\documentclass[amsart,12pt,oneside,english]{article}
\pdfoutput=1
\usepackage{amsthm, amsmath, amsfonts, amsxtra, amssymb, euscript, mathrsfs, MnSymbol, verbatim, enumerate, multicol, multirow, color,babel, geometry, tikz,tikz-cd, tikz-3dplot, tkz-graph, array, hyperref, thm-restate, thmtools, datetime, graphicx, tensor, braket, slashed, adjustbox, mathtools, bbding, esint, pgfplots, ytableau, float, rotating, mathdots, savesym, cite, wasysym, amscd, pifont, setspace, wrapfig, picture, subcaption, caption, tabularx,youngtab,parskip}
\usepackage[numbers,sort&compress]{natbib}
\usepackage[normalem]{ulem} 
\usepackage[utf8]{inputenc}
\usepackage[all]{xy}
\usepackage[noabbrev]{cleveref}
\numberwithin{equation}{section}
\numberwithin{figure}{section}
\usetikzlibrary{arrows, positioning, decorations.pathmorphing, decorations.markings, decorations.pathreplacing, decorations.markings, matrix, patterns, snakes, shapes}
\tikzset{
on each segment/.style={
decorate,
decoration={
show path construction,
moveto code={},
lineto code={
\path [#1]
(\tikzinputsegmentfirst) -- (\tikzinputsegmentlast);
},
curveto code={
\path [#1] (\tikzinputsegmentfirst)
.. controls
(\tikzinputsegmentsupporta) and (\tikzinputsegmentsupportb)
..
(\tikzinputsegmentlast);
},
closepath code={
\path [#1]
(\tikzinputsegmentfirst) -- (\tikzinputsegmentlast);
},
},
},
mid arrow/.style={postaction={decorate,decoration={markings,mark=at position .7 with {\arrow[#1]{stealth}}}}},
}
\hypersetup{colorlinks=true}
\hypersetup{linkcolor=black}
\hypersetup{citecolor=black}
\hypersetup{urlcolor=black}

\makeatletter
\def\oversortoftilde#1{\mathop{\vbox{\m@th\ialign{##\crcr\noalign{\kern3\p@}%
				\sortoftildefill\crcr\noalign{\kern3\p@\nointerlineskip}%
				$\hfil\displaystyle{#1}\hfil$\crcr}}}\limits}

\def\sortoftildefill{$\m@th \setbox\z@\hbox{$-$}%
	\braceld\leaders\vrule \@height\ht\z@ \@depth\z@\hfill\braceru$}

\makeatother

\theoremstyle{plain}
\newtheorem*{thm*}{Theorem}
\newtheorem{thm}{Theorem}[section]

\newtheorem{alg}[thm]{Algorithm}

\theoremstyle{definition}

\newtheorem*{defn*}{Definition}

\makeatother
\newcommand{\calo}{\mathcal{O}}

\newcommand{\calh}{\mathcal{H}}

\begin{document}

\begin{titlepage}
% Report number
\vspace*{-3cm} 
\begin{flushright}
{\tt CALT-TH-2020-016}\\
\end{flushright}
\begin{center}
\vspace{2.5cm}
{\LARGE\bfseries The infinite-dimensional HaPPY code:\\ }
{\Large\bfseries entanglement wedge reconstruction and dynamics \\  }
\vspace{2cm}
{\large
Elliott Gesteau$^{1,2}$ and Monica Jinwoo Kang$^3$\\}
\vspace{.6cm}
{ $^1$ Département de Mathématiques et Applications, Ecole Normale Supérieure}\par\vspace{-.3cm}
{Paris, 75005, France}\par
{ $^2$ Perimeter Institute for Theoretical Physics}\par\vspace{-.3cm}
{Waterloo, Ontario N2L 2Y5, Canada}\par
\vspace{.2cm}
{ $^3$ Walter Burke Institute for Theoretical Physics, California Institute of Technology}\par\vspace{-.3cm}
{  Pasadena, CA 91125, U.S.A.}\par
\vspace{.2cm}
\vspace{.6cm}

\scalebox{.95}{\tt  egesteau@perimeterinstitute.ca, monica@caltech.edu}\par
\vspace{2cm}
{\bf{Abstract}}\\
\end{center}
{We construct an infinite-dimensional analog of the HaPPY code as a growing series of stabilizer codes defined respective to their Hilbert spaces. The Hilbert spaces are related by isometric maps, which we define explicitly. We construct a Hamiltonian that is compatible with the infinite-dimensional HaPPY code and further study the stabilizer of our code, which has an inherent fractal structure. We use this result to study the dynamics of the code and map a nontrivial bulk Hamiltonian to the boundary. We find that the image of the mapping is scale invariant, but does not create any long-range entanglement in the boundary, therefore failing to reproduce the features of a CFT. This result shows the limits of the HaPPY code as a model of the AdS/CFT correspondence, but also hints that the relevance of quantum error correction in quantum gravity may not be limited to the CFT context.}
\\
\vfill 
%{ \today\   at \ \currenttime\par}
\end{titlepage}

\tableofcontents
\newpage
\section{Introduction}
One of the most well understood examples of holography is the AdS/CFT correspondence \cite{Maldacena:1997re,Gubser:1998bc,Witten:1998qj}.
In the past decade, much progress has been made in the understanding of bulk reconstruction in AdS/CFT via quantum error correction, the first attempt being made in \cite{Almheiri:2014lwa}. In particular, the Ryu--Takayanagi formula \cite{RT2006}, which relates the entanglement entropy of the boundary conformal field theory to the area of an extremal surface in the bulk spacetime, plays a central role for understanding how spacetime geometry emerges from boundary entanglement. In fact, the Ryu--Takayanagi formula has later been proven to be equivalent to the conservation of bulk and boundary relative entropies, and to the reconstruction of the \textit{entanglement wedge} of a boundary region in the context of finite-dimensional Hilbert spaces \cite{Jafferis:2015del,Harlow:2016vwg,DongHarlowWall,Faulkner:2013ana}. These results use the formalism of operator algebra quantum error correction, in which the space of quantum field states of a fixed background geometry is seen as a \textit{code subspace} of the whole Hilbert space of the underlying quantum gravity theory with a semi-classical bulk.\footnote{QFT on a fixed bulk geometry can only be defined in the semiclassical limit of the theory.}

Unlike a causal wedge, an entanglement wedge can reach very deep within the bulk. In particular, it can contain the interior of an AdS black hole, which makes it an ideal candidate to solve important quantum gravity questions such as the black hole information paradox, as proposed first in \cite{Almheiri:2019psf,Penington:2019npb}. In order to understand entanglement wedge reconstruction, many toy models have been considered. A particularly successful class of toy models are the ones based on tensor networks  \cite{Pastawski:2015qua,Kang:2019dfi,Heydeman:2018qty}, including the HaPPY code. The HaPPY code is without a doubt the most celebrated holographic quantum error correcting code based on tensor networks. Its stunning successes include a redundant encoding of the bulk degrees of freedom on the boundary, as well as a well-defined Ryu--Takayanagi formula \cite{Pastawski:2015qua}. However, one of the main limitations of the HaPPY code is that it is a finite code, whereas one expects the boundary theory in AdS/CFT to have an infinite-dimensional Hilbert space.

In this paper, we construct an infinite-dimensional analog of the HaPPY code in a mathematical manner, and study some of its dynamical properties. We investigate entropic considerations separately in our companion paper \cite{MonicaElliott2}. In terms of the dynamics, it has often been claimed that the HaPPY code pushes a field theory in the bulk to a CFT on the boundary. As explicit operator pushing is tedious, it is quite hard to find precise statements which can probe this claim.\footnote{By operator pushing, we mean operators of the bulk to be mapped to the operators on the boundary. This is different from state pushing, which maps the state of the bulk to the state on the boundary. This difference is the key that complicates the analysis of the HaPPY code.} The best approach which we are aware of is the one of Majorana dimers, however, it only maps bulk states without interactions to the boundary \cite{Jahn:2019nmz}. The main goal of this paper is to question whether a physical theory in the bulk, which exhibits nontrivial correlation functions, will be mapped to a boundary CFT-like theory. Our analysis therefore culminates in the study of the bulk-to-boundary mapping of a nontrivial bulk theory.

We start by giving a precise mathematical meaning to the idea of ``pushing the boundary to infinity." Our setting involves direct systems of Hilbert spaces (as in \cite{Kang:2019dfi}), and we show that in the direct limit, an isometry between the infinite-dimensional code and physical spaces can be rigorously defined. We build an algorithm for a recursive map that reconstructs the HaPPY code as a stabilizer code at every level by incorporating Bell pairs in the bulk.

Once the infinite-dimensional setup is constructed, we move on to studying the structure of the stabilizer of the infinite-dimensional HaPPY code. We find an explicit family of fractal stabilizer generators, which interestingly echoes the concept of \textit{uberholography} introduced by Pastawski and Preskill \cite{Pastawski:2016qrs}. This suggests that any bulk local operator can be encoded in a fractal region of the boundary.

We then move on to our main goal of constructing a nontrivial bulk theory and mapping it to the boundary. In order to do so, we construct a bulk Hamiltonian from trapeze-shaped building blocks, and show that these building blocks create enough bulk entaglement so that the theory in the bulk has a physical meaning; in particular, it exhibits a phase transition. However, we show that the bulk-to-boundary map disentangles our Hamiltonian vastly, to a point where the boundary image will indeed be scale-invariant, but fails to create any global long-range entanglement. Therefore, it contradicts the claim that the HaPPY code should simulate a CFT, in which we expect an algebraic decay of the correlation functions.

Our result shows that the HaPPY code is too much of a disentangler to map the correlations all the way to the infinite boundary, at least for our choice of Hamiltonian. Since we have demonstrated that for certain choices of Hamiltonians the HaPPY code does not provide the expected holographic toy model for the AdS/CFT, it raises the question of which Hamiltonians, if any, would lead to a proper holographic model. On the other hand, we have shown that within the framework of quantum error correction, it is possible to encode a physically nontrivial theory in the code subspace of a very de-correlated boundary Hamiltonian. This Hamiltonian couples only finite numbers of images of bulk interactions together except in some very specific cases, and may nevertheless still need to be regulated in some cases.

The rest of the paper is organized as follows. In Section \ref{sec:MERA}, we explain the connection between tensor networks and renormalization by introducing the multi-scale entanglement renormalization ansatz (MERA) and showing its similarity with the HaPPY code. This allows us to relate isometries and the renormalization group flow.
We then comment on the structure of the HaPPY code in Section \ref{sec:pentagon}. We focus on the geometry of the pentagon tiling and find recursive relations between the numbers of bulk and boundary qubits. In Section \ref{sec:isometryH}, we construct bulk and boundary Hilbert spaces using direct limit techniques, and construct the isometries between them by using Bell pairs. In Section \ref{sec:BoundaryStabilizerGen}, we construct stabilizer generators for each level of the code, and extensively study the fractal patterns that arise from the successive pushes of these operators. In Section \ref{sec:trapezeHamiltonian}, we construct the trapeze Hamiltonian, which yields a bulk theory mimicking local quantum field theory in curved spacetime. We study this theory through a mean-field approximation in the bulk and discover a phase transition, verifying a nontrivial behavior of the bulk quantum field theory. In Section \ref{sec:bulktoboundarytrapeze}, we construct the \textit{uberholographic} bulk-to-boundary map for this trapeze Hamiltonian and study the resulting boundary correlation functions in the infinite-dimensional limit. In particular, we show that these correlation functions do not match with the ones of a CFT. We prove a technical result on these correlation functions in Appendix \ref{sec:trapezeEW}. This result will be at the center of our analysis of entropic properties in the new setup of our companion paper \cite{MonicaElliott2}.  We conclude in Section \ref{sec:conclusion} by emphasizing the shortcomings of the infinite-dimensional analog of the HaPPY code as a toy model and explain how the same issues call for a new formalism to understand its entropic properties, which will be discussed in \cite{MonicaElliott2}.

\section{Tensor networks, MERA, and renormalization}
\label{sec:MERA}

In this section, we will outline how renormalization, as required in a holographic picture, is encoded in the tensor network. The simplest renormalization scheme is block-spin renormalization, which is implemented by tree tensor networks. For a one-dimensional lattice, block-spin renormalization consists in grouping neighboring spins into blocks and projecting (or truncating) the Hilbert space of each block onto the Hilbert space of a single spin \cite{Kadanoff}. The more involved multi-scale entanglement renormalization ansatz (MERA) does not use a tree tensor network. Before applying the projection to a block of spins, the MERA tensor network applies unitary transformations, called disentanglers, that reduce the entanglement between neighboring blocks \cite{Vidal:2007hda,Vidal:2008zz}. This aspect is well depicted in Figure \ref{fig:merasetup}, which provides the schematic diagram of the MERA tensor network.

The MERA tensor network given in Figure \ref{fig:merasetup} depicts an isometry $u$ that maps states in a long distance theory to states in a short distance theory. Note that $uu^\dagger$ is a projection. This tensor network also defines a renormalization group flow for operators. If $\calo$ is an operator that acts on the Hilbert space of the short-distance theory, it can be mapped to an operator in the long-distance theory as $\calo \rightarrow  u^\dagger \calo u$. \emph{Primary operators} are local operators (that is, operators that act on a single spin or a small number of neighboring spins) that are mapped to themselves in this way, up to a constant multiplicative factor.

The disentanglers are important because they reduce the rank of the reduced density matrix of a single block. In particular, note that the expectation value of an operator $\calo$ in a state $\rho$ only depends on the matrix elements of $\calo$ in the nonzero eigenspace of $\rho$ (which is the orthocomplement of $\ker \rho$), whose dimensionality is the rank of $\rho$. Ideally, we want the rank of $\rho$ to not exceed the dimensionality of the Hilbert space of each spin in the long distance theory so that correlation functions can be computed with high accuracy in the long distance theory.

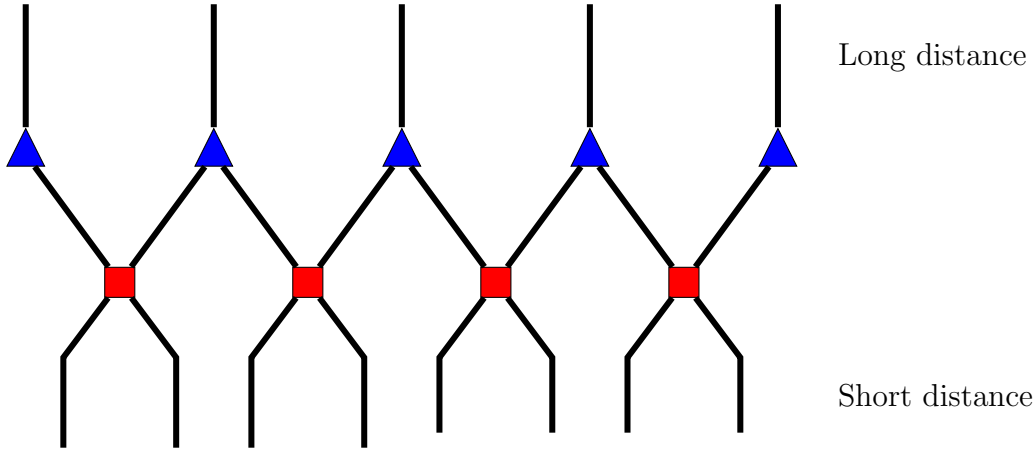
\begin{figure}[H]
\centering
\begin{tikzpicture}
\node[draw,fill=blue,isosceles triangle,isosceles triangle stretches,shape border rotate=90,minimum
width=0.5cm,minimum height=0.5cm,inner sep=0pt] (T1) at (-5,1.7) {};
\node[draw,fill=blue,isosceles triangle,isosceles triangle stretches,shape border rotate=90,minimum
width=0.5cm,minimum height=0.5cm,inner sep=0pt] (T2) at (-2.5,1.7) {};
\node[draw,fill=blue,isosceles triangle,isosceles triangle stretches,shape border rotate=90,minimum
width=0.5cm,minimum height=0.5cm,inner sep=0pt] (T3) at (0,1.7) {};
\node[draw,fill=blue,isosceles triangle,isosceles triangle stretches,shape border rotate=90,minimum
width=0.5cm,minimum height=0.5cm,inner sep=0pt] (T4) at (2.5,1.7) {};
\node[draw,fill=blue,isosceles triangle,isosceles triangle stretches,shape border rotate=90,minimum
width=0.5cm,minimum height=0.5cm,inner sep=0pt] (T5) at (5,1.7) {};
\node[draw,fill=red,rectangle,minimum width=0.4cm,minimum height=0.4cm,inner sep=0pt] (R1) at (-3.75,0) {};
\node[draw,fill=red,rectangle,minimum width=0.4cm,minimum height=0.4cm,inner sep=0pt] (R2) at (-1.25,0) {};
\node[draw,fill=red,rectangle,minimum width=0.4cm,minimum height=0.4cm,inner sep=0pt] (R3) at (1.25,0) {};
\node[draw,fill=red,rectangle,minimum width=0.4cm,minimum height=0.4cm,inner sep=0pt] (R4) at (3.75,0) {};
\node[draw=none,label={[label distance=0.1mm]0:Long distance}] (C1) at (5.5,3) {};
\draw[draw,line width=0.8mm] (T1)--(R1)--(T2)--(R2)--(T3)--(R3)--(T4)--(R4)--(T5);
\draw[draw,line width=0.8mm] (-4.5,-2.2)--(-4.5,-1)--(R1)--(-3,-1)--(-3,-2.2);
\draw[draw,line width=0.8mm] (-2,-2.2)--(-2,-1)--(R2)--(-.5,-1)--(-.5,-2.2);
\draw[draw,line width=0.8mm] (2,-2)--(2,-1)--(R3)--(.5,-1)--(.5,-2);
\draw[draw,line width=0.8mm] (4.5,-2)--(4.5,-1)--(R4)--(3,-1)--(3,-2);
\draw[draw,line width=0.8mm] (T1)--(-5,3.7);
\draw[draw,line width=0.8mm] (T2)--(-2.5,3.7);
\draw[draw,line width=0.8mm] (T3)--(0,3.7);
\draw[draw,line width=0.8mm] (T4)--(2.5,3.7);
\draw[draw,line width=0.8mm] (T5)--(5,3.7);
\node[draw=none,label={[label distance=0.1mm]0:Short distance}] (C2) at (5.5,-1.5) {};
\end{tikzpicture}
\caption{The tensor network of MERA. Each red square represents a disentangler. Each blue triangle represents an isometry from up to down. This tensor network defines an isometry from the Hilbert space of the spins on the top row (which correspond to a ``long distance'' description of a theory) to the Hilbert space of the spins on the bottom row (which correspond to a ``short distance'' description of the same theory).}
\label{fig:merasetup}
\end{figure}

The HaPPY code is similar to MERA, except that its tensor network contains bulk indices and is based on the geometry of the Poincare disk. In other words, the HaPPY code contains holographic data, which makes it more complex. The bulk indices are important for defining local bulk operators, although locality below the size of a single (pentagon) tile cannot be achieved. Just as MERA is organized in a hierarchy of levels that correspond to different renormalization scales, the HaPPY code is naturally organized into a sequence of levels (see Figure \ref{fig:level1tensornetwork}), which we interpret as renormalization scales in the boundary theory. Each level corresponds to a finite number of boundary qubits that comprise a Hilbert space called $\widetilde{\calh}_n$ for $n \in \mathbb{N}$. The Hilbert space of the boundary theory (of the infinite-dimensional tensor network) is defined as (the completion of) the disjoint union of all of the $\widetilde{\calh}_n$, subject to an equivalence relation that identifies states in different $\widetilde{\calh}_n$ that can be related by renormalization group flow. This physically justifies why we write
\begin{align}
\widetilde{\calh}_1 \rightarrow \widetilde{\calh}_2 \rightarrow \widetilde{\calh}_3 \rightarrow \cdots,
\end{align}
where each right arrow ($\rightarrow$) corresponds to the isometry that embeds the Hilbert space at one scale into the Hilbert space of the next smallest scale. These isometries should arise from the HaPPY code. They depend on the state of the bulk qubits. More precisely, the isometry from $\widetilde{\calh}_n$ to $\widetilde{\calh}_{n+1}$ depends on the state of the newly introduced bulk qubits in the larger tensor network corresponding to $\widetilde{\calh}_{n+1}$. These bulk qubits must be in a state which is uniquely determined by the isometric map, in order to have a well-defined renormalization group flow provided by the tensor network. Different bulk states lead to different renormalization group flows, which correspond to different boundary theories. This is similar to the notion that in AdS/CFT, different boundary conditions for bulk fields correspond to turning on different sources in the boundary CFT. We want to use the HaPPY code to better understand the origin of the extrapolate dictionary in AdS/CFT.\footnote{MERA has been used in the context of AdS/CFT first in \cite{Swingle:2009bg}.} This necessitates studying the dynamics of the infinite-dimensional analog of the HaPPY code in a holographic context.

%The correlators of bulk operators approach the correlators of boundary operators as the bulk operators are moved towards the boundary. For a single bulk operator $\phi$ that is dual to $\calo$ on the boundary, we write
%\[ \phi(r,x) \rightarrow r^\Delta \calo(x), \]
%where $r$ is the radial coordinate in the bulk and $\Delta$ is a scaling dimension. If we can achieve this expression in the continuum HaPPY code, then we will have confirmed the relationship between radial direction and renormalization group flow.

\section{The pentagon tiling structure of the HaPPY code}
\label{sec:pentagon}
In this section, we explain how to push the structure of the HaPPY code to an infinite boundary in a rigorous manner. The HaPPY code, invented in \cite{Pastawski:2015qua}, is based on a pentagonal tiling of the hyperbolic plane, where in each tile lies a perfect tensor, most often constructed out of a stabilizer code like the five qubit code, which is the one we will study in this paper. We chose a pentagonal tiling as it is the simplest holographic code leading to nontrivial quantum error correction. Each bulk tensor contains dangling legs, each corresponding to the code subspace of the network, and is mapped to the boundary.

\begin{figure}[H]
	\vspace{2mm}
	\centering
	\includegraphics[width=0.45\linewidth]{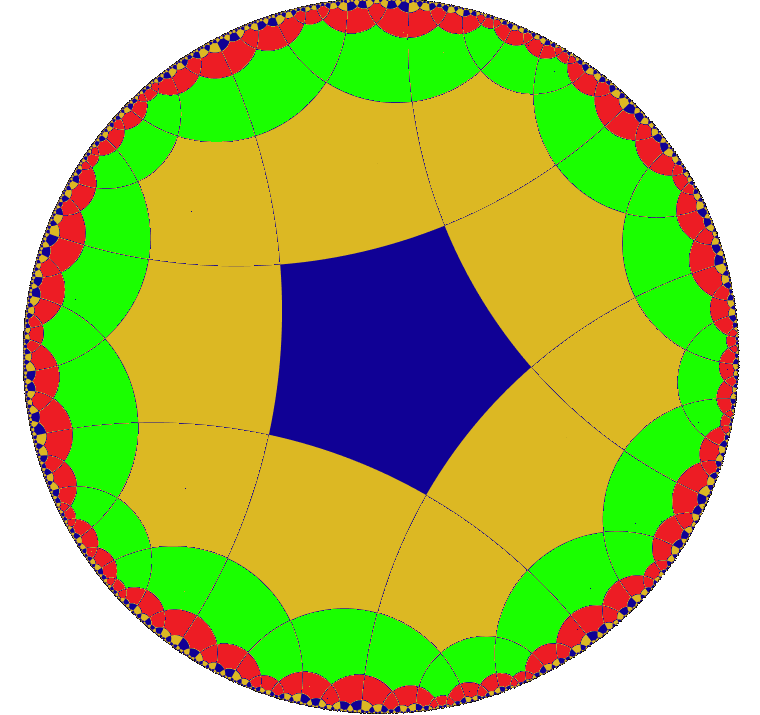}
	\caption{The same level tiles share the same color in this picture. Level 1 consists of the blue tile in the center, level 2 consists of the yellow tiles surrounding level 1, level 3 consists of the green tiles, level 4 consists of the red tiles, etc.}
	\label{fig:levels}
\end{figure}

\subsection{From the pentagon tiling to the tensor network}

The Poincare disk is separated into individual pentagonal tiles by geodesic lines. We can choose to label the center as the first pentagon tile. We call this Level 1; more precisely, level 1 consists simply of the center pentagon, as represented in blue in Figure \ref{fig:levels}. 

Then, the levels can be rigorously defined in a recursive manner. Let us define the notion of the \emph{level n}. Given the level 1 tile, the level 2 tiles are defined to be all the tiles that share any edges and/or vertices with the level 1 tile. Likewise, the level 3 tiles are defined to be all tiles (except for tiles that are already in levels 1 or 2) that share any edges and/or vertices with level 2 tiles.

With this definition of the levels in hand, we can now describe its usage for producing a state of finitely-many ($\widetilde{N}_n$) qubits. Let us start at level 1. We put a bulk qubit in the level 1 tile and then draw boundary legs through each edge of the level 1 tile. The result is presented on the left in Figure \ref{fig:level1tensornetwork}.

\begin{figure}[!h]
	\centering
	\vspace{1mm}
	\includegraphics[width=0.45\linewidth]{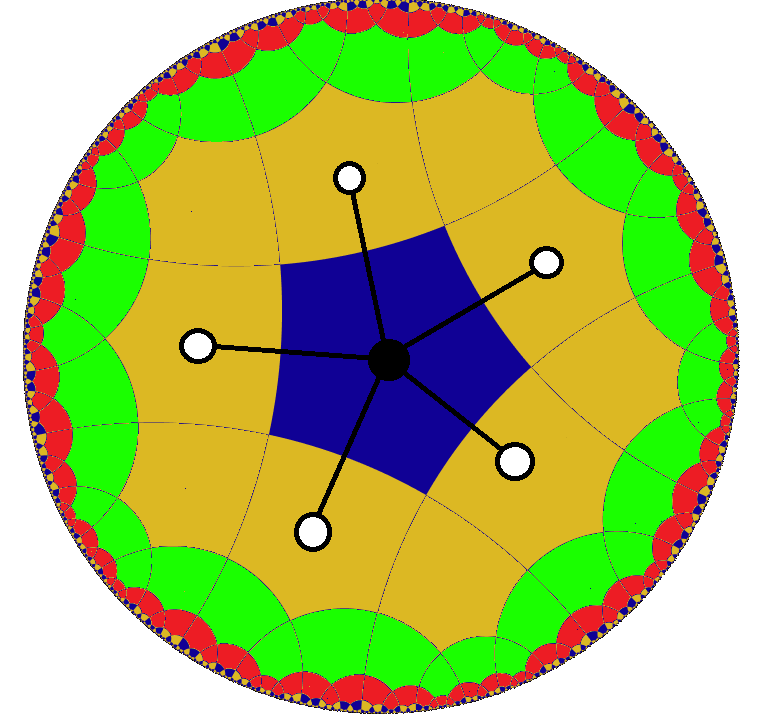}
	\includegraphics[width=0.45\linewidth]{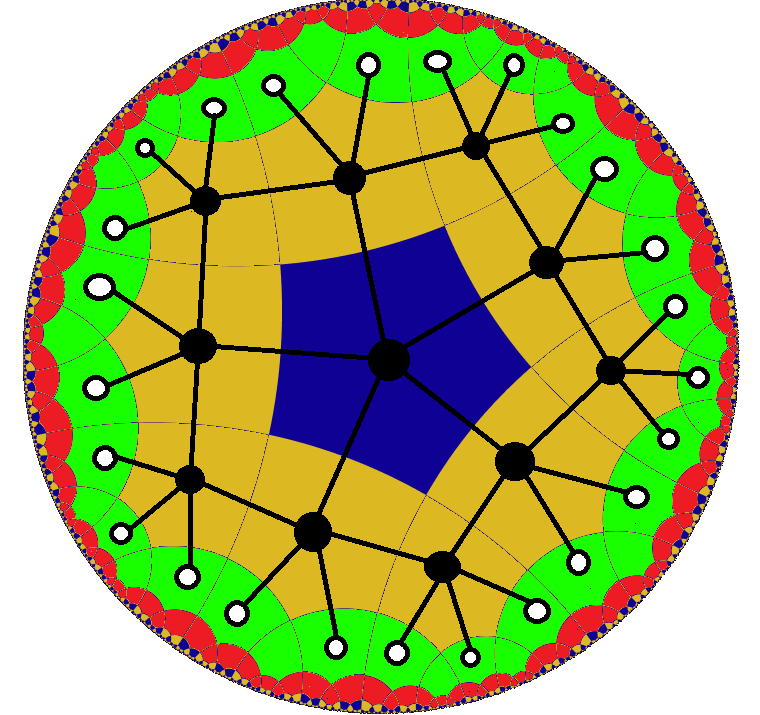}
	\vspace{1mm}
	\caption{{\bf Left}: The tensor network produced at level 1. There is one bulk qubit denoted in black and five boundary qubits denoted in white. In this case, $\widetilde{h}_1 = 5$. {\bf Right}: The level 2 tensor network. Here, $\widetilde{h}_2 = 25$ since there are 25 boundary (white) qubits.}
	\label{fig:level1tensornetwork}
\end{figure}

\begin{table}[H]
\arraycolsep=8pt\def\arraystretch{1.3}
\begin{center}
$
\begin{array}{|c|c||c|c|c|c|c|c|c|}
\hline 
\text{Levels} & n & 1 & 2 & 3 & 4 & 5 & 6 & \cdots \\ 
\hline
\hline
\text{\# Bulk qubits} & N_n & 1 & 11 & 51 & 201 & 761 & 2851 & \cdots \\ 
\hline
\text{\# Boundary qubits} & \widetilde{N}_n & 5 & 25 & 95 & 355 & 1325 & 4945 & \cdots \\ 
\hline 
\end{array} 
$
\end{center}
\caption{Numbers of bulk qubits (denoted as ${N}_n$) and boundary qubits (denoted as $\widetilde{N}_n$) at various levels $n$.}
\label{table:Nn}
\end{table}

The \emph{level 1 tensor network} is defined to be a map from one bulk qubit to five boundary qubits. We define $\widetilde{N}_n$ to be the number of boundary qubits of the level $n$ tensor network. Thus, $\widetilde{h}_1 = 5$. %To be concrete, we will put the bulk qubit in the state $\ket{\downarrow}$. 
Then, the level 1 tensor network maps the bulk state to a particular state of $\widetilde{h}_1 = 5$ boundary qubits, which we will call \emph{the level 1 boundary state.}

The \emph{level n tensor network} is constructed as follows. First, put a bulk qubit in every tile in levels $1,2,\cdots,n$. Then, draw a connecting leg through any edge that borders at least one tile with a bulk qubit. Of course, each edge borders two tiles. A leg that connects two bulk qubits defines a contraction of two tensor indices. A leg that connects a bulk qubit with a tile in level $n+1$ is a ``dangling leg'' and it defines a boundary qubit. For example, the level 2 tensor network is represented on the right in Figure \ref{fig:level1tensornetwork}.

The bulk (boundary) state of the tensor network is defined as the tensor product of the states of each qubit in the bulk (boundary). This defines a state of $\widetilde{N}_n$ boundary qubits and $N_n$ bulk qubits at every level $n$. We list $N_n$ and $\widetilde{N}_n$ at the first few levels in Table \ref{table:Nn}.

\subsection{Counting bulk and boundary qubits at each level}
\label{sec:AsymototesHn}

Starting at the level 1 tile, one can move to the ten level 2 tiles by moving through one of the five vertices of the level 1 tile or one of the five edges of the level 1 tile. After moving through an edge, one can move through one of the two edges or one vertex to get to the next level. After moving through a vertex, one can move through one of three edges or two vertices to get to the next level. See Figure \ref{fig:increasinglevels} as an example.

\begin{figure}[H]
\centering
\includegraphics[width=0.45\linewidth]{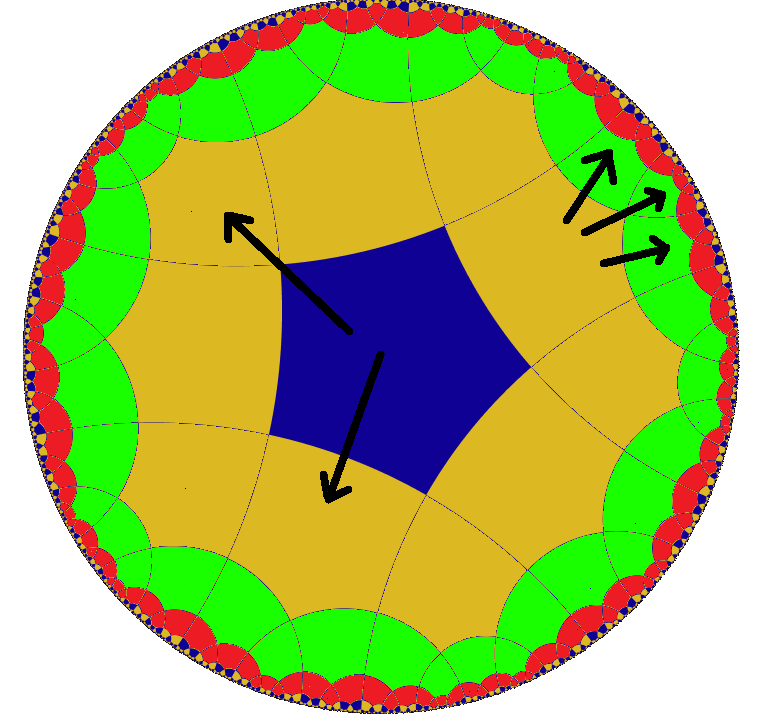}
\vspace{3mm}
\caption{We show how one can move through an edge or a vertex to move from one level to another.}
\label{fig:increasinglevels}
\end{figure}
We let $\left(\begin{array}{c}
e \\ 
v
\end{array} \right)$ refer to the number of edges $e$ and vertices $v$ at a given level. The level 1 tile has 5 edges and 5 vertices, so we represent it with $\left(\begin{array}{c}
5 \\ 
5
\end{array} \right)$. The number of edges and vertices at the $n$th level is obtained by acting on this vector $n$ times with the matrix
\[ \left(\begin{array}{cc}
2 & 3 \\ 
1 & 2
\end{array}\right),  \]
which is in $SL(2,\mathbb{Z})$.

So if $e_n$ and $v_n$ denote the number of edges and vertices at the $n$th level, then we have
\begin{align}
\left(\begin{array}{c}
e_n \\ 
v_n
\end{array} \right) =  \left(\begin{array}{cc}
2 & 3 \\ 
1 & 2
\end{array}\right)^{(n-1)} \left(\begin{array}{c}
5 \\ 
5
\end{array} \right)
\label{eq:edgesvertices}
\end{align}
The number of black (bulk) qubits at level $n$ is $e_{n-1} + v_{n-1}$ and the number of white (boundary) qubits at level $n$ is $e_n$. More precisely, by denoting $N_n$ as the number of bulk qubits at level $n$ and $\widetilde{N}_n$ as the number of boundary qubits at level $n$, 
\begin{align}
\begin{aligned}
N_n&=\sum_{k=0}^{n-1} \left( e_k+v_k \right),\\
\widetilde{N}_n &=e_n.
\label{eq:NNtilde}
\end{aligned}
\end{align}

We find recursive relations for the numbers of bulk and boundary qubits $N_n$ and $\widetilde{N}_n$, utilizing equation \eqref{eq:NNtilde}, which we list in Theorem \ref{thm:RecursiveRelation}.
 \medskip
\begin{thm}
Let $N_n$ be the number of bulk qubits at level $n$, and $\widetilde{N}_n$ be the number of boundary qubits at level $n$. Then the recursion relations of (the infinite-dimensional analog of) the HaPPY code are given by
\begin{align*}
\begin{cases}
& N_{n+1} = 6 N_n - 9 N_{n-1} + 2 N_{n-2} - 8 \\
& \widetilde{N}_{n+1} = 2 N_{n+1} - N_n + 4 .
\end{cases}
\end{align*}
\label{thm:RecursiveRelation}
\end{thm}
\begin{proof}
The edges and vertices at an arbitrary level $n$ satisfy the relation in equation \eqref{eq:edgesvertices}. It follows that that the recursive relations of $e_n$ and $v_n$ are
\begin{align*}
\begin{cases}
& e_n=2e_{n-1}+3v_{n-1},\quad e_1=5, \\
& v_n=e_{n-1}+2v_{n-1},\quad v_1=5.
\end{cases}
\end{align*}
Then, the number of boundary qubits ($\widetilde{N}_n$) and the number of bulk qubits ($N_n$) are given by
\begin{align*}
\widetilde{N}_n=2e_{n-1}+3v_{n-1}=e_n,\quad N_{n+1}-N_n=e_n+v_n.
\end{align*}
By considering the recursive relations we want to prove, we consider the differences between $n=k$ and $n=k-1$:
\begin{align*}
\begin{cases}
& N_{n+1}-N_n = 6 (N_n-N_{n-1}) - 9 (N_{n-1}-N_{n-2}) + 2 (N_{n-2}-N_{n-3}) ,\\
& \widetilde{N}_{n+1}-\widetilde{N}_n = 2 (N_{n+1}-N_{n}) - (N_n-N_{n-1}) .
\end{cases}
\end{align*}
We show that our $N_n$ and $\widetilde{N}_n$ satisfy these recursive relations:
\begin{align*}
 6 (N_n-N_{n-1}) -9 (N_{n-1}-N_{n-2}) & + 2 (N_{n-2}-N_{n-3}) \\
& = 6 (e_{n-1}+v_{n-1})-9 (e_{n-2}+v_{n-2})+2(e_{n-3}+v_{n-3}) \\
& =(2e_{n-1}+3v_{n-1})+(e_{n-1}+2v_{n-1})+3e_{n-1}+v_{n-1}\\
&\qquad\ -3(2e_{n-2}+3v_{n-2})-3e_{n-2}+(2e_{n-3}+3v_{n-3})-v_{n-3} \\
& =e_n+v_n +(e_{n-2}+2v_{n-2})-2e_{n-2}-v_{n-3} \\
& =e_n+v_n+2(e_{n-3}+2v_{n-3})-(2e_{n-3}+3v_{n-3})-v_{n-3}
	=e_n+v_n \\
& =N_{n+1}-N_n,
\end{align*}
\begin{align*}
2 (\widetilde{N}_{n+1}-\widetilde{N}_{n}) - (\widetilde{N}_n-\widetilde{N}_{n-1}) &= 2 (e_n+v_n) - (e_{n-1}+v_{n-1}) = (2e_n+3v_n)-v_n-(e_{n-1}+v_{n-1})\\
&=e_{n+1}-(e_{n-1}+2v_{n-1}+e_{n-1}+v_{n-1})=e_{n+1}-e_n\\
&=\widetilde{N}_{n+1}-\widetilde{N}_n.
\end{align*}
Furthermore, we show that for the initial condition of the series, the following relations hold for the HaPPY code:
\begin{align*}
& N_3 = 6 N_3 - 9 N_2 + 2 N_1 - 8 \\
& \widetilde{N}_2 = 2 N_2 - N_1 + 4 ,\quad\widetilde{N}_3 = 2 N_3 - N_2 + 4 ,\quad\widetilde{N}_4 = 2 N_4 - N_4 + 4 .
\end{align*}
Hence, we conclude that the recursive relation given by the theorem holds by induction.
\end{proof}

We can also find closed form formulas for $N_n$ and $\widetilde{N}_n$. In order to do so, we start with equation \eqref{eq:edgesvertices} with the initial value
\begin{align}
\left(\begin{array}{c}
e_0 \\ 
v_0
\end{array} \right) = \left(\begin{array}{c}
0 \\ 
1
\end{array} \right) .
\end{align}
The edges and vertices at level $n$ can be then computed in closed form formulas as
\begin{align}
e_n&=-\frac{5}{2}\left(\left(1+\sqrt{3}\right) \left(2-\sqrt{3}\right)^n+\left(1-\sqrt{3}\right) \left(2+\sqrt{3}\right)^n\right), \\
v_n&=\frac{5}{2\sqrt{3}} \left(\left(1+\sqrt{3}\right) \left(2-\sqrt{3}\right)^n-\left(1-\sqrt{3}\right) \left(2+\sqrt{3}\right)^n\right).
\end{align}
Utilizing these, we find the number of bulk qubits $N_n$ and the number of boundary qubits $\widetilde{N}_n$ at level $n$ to be
\begin{align}
N_n &=-5+\frac{5}{\sqrt{3}+3}\left(\left(2-\sqrt{3}\right)^{n-1}+\left(2+\sqrt{3}\right)^n\right), \label{eq:Nn} \\
\widetilde{N}_n&=\frac{5}{1-\sqrt{3}}\left(\left(2-\sqrt{3}\right)^n-\left(2+\sqrt{3}\right)^{n-1}\right). \label{eq:Nntilde}
\end{align}
We find the ratio of the numbers between the bulk qubits and boundary qubits at each level $n$ to be
\begin{align}
\frac{N_n}{\widetilde{N}_n}=-\frac{\left(1-\left(2-\sqrt{3}\right)^n\right) \left(1-\left(2-\sqrt{3}\right)^{n-1}\right)}{\sqrt{3} \left(1-\left(2-\sqrt{3}\right)^{2 n-1}\right)} .
\end{align}
We plot numeric values of this ratio for small values of $n$ in Figure \ref{fig:BulkBoundaryNodeRatio}. As the levels continue increasing, the ratio converges quickly to the asymptotic value
\begin{align}
\lim_{n\to\infty} \frac{N_n}{\widetilde{N}_n}\approx 0.577.
\end{align}

\begin{figure}[H]
\centering
\includegraphics[scale=0.68]{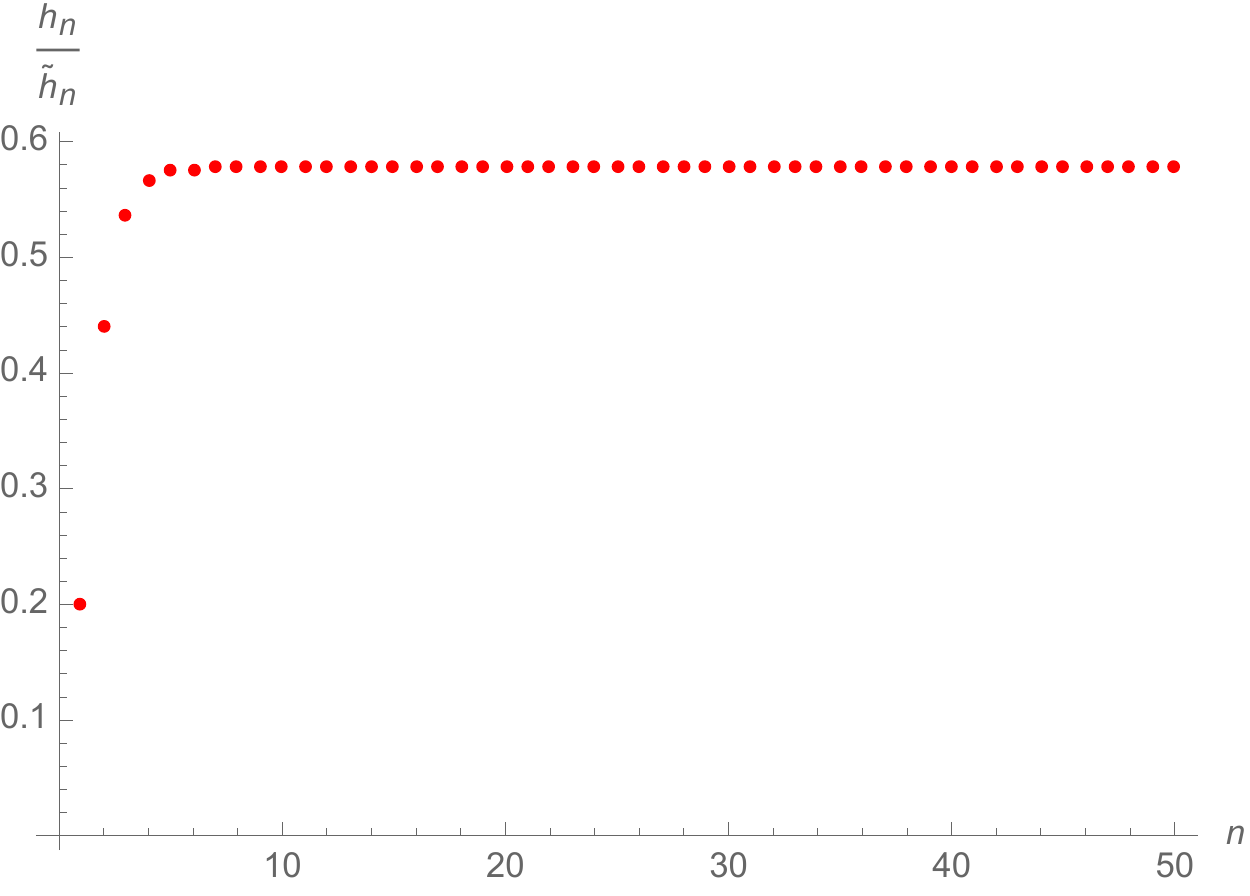}
\caption{The ratio between the number of the bulk qubits ($N_n$) and the number of the boundary qubits ($\widetilde{N}_n$) at level $n$. We see that it asymptotes quickly to $0.577$.}
\label{fig:BulkBoundaryNodeRatio}
\end{figure}

At each node of the tensor network, we place in general a {\it qudit}, which has a $d$-dimensional Hilbert space. Let $\calh_n$ for $n \in \mathbb{N}$ denote the Hilbert space of the black (bulk) qudits at level $n$. Similarly, $\widetilde{\calh}_n$ denotes the Hilbert space of the white (boundary) qudits at level $n$.\footnote{See Section \ref{sec:isometryH} for more details and the isometric maps.} Then, the dimensions of the Hilbert spaces are given by
\begin{align}
\mathrm{dim} \calh_n^{(d)}=d^{N_n},\quad 
\mathrm{dim} \widetilde{\calh}_n^{(d)}=d^{\widetilde{N}_n}.
\end{align}
For the case of the HaPPY code, we place qubits at each node; hence we consider $d=2$. Since we only consider the constructions with qubits, we drop this superindex $(d)$ for the rest of the paper.
We find the ratio between the dimensions of the bulk and boundary Hilbert spaces for our HaPPY code to be
\begin{align}
\frac{\mathrm{dim} \calh_n^{(d=2)}}{\mathrm{dim} \widetilde{\calh}_n^{(d=2)}}=2^{\frac{5}{\sqrt{3}}\left(-\sqrt{3}-\left(2+\sqrt{3}\right)^{n-1}+\left(2-\sqrt{3}\right)^{n-1}\right)}.
\label{eq:HHtRatio}
\end{align}

We find the ratio of the dimensions of the Hilbert spaces to be quickly vanishing for a sufficiently large $n$, as shown in Figure \ref{fig:dimHratio}. More precisely, Figure \ref{fig:dimHratio} is a lin-log plot that portrays the logarithmic value of the ratio, $\log \frac{\mathrm{dim} \calh_n}{\mathrm{dim} \widetilde{\calh}_n}$, with respect to level $n$; we find the logarithmic ratio to be fastly diverging negatively. It follows that the ratio converges to zero:
\begin{align}
\lim_{n\to\infty} \frac{\mathrm{dim} \calh_n}{\mathrm{dim} \widetilde{\calh}_n}=0.
\end{align}
We list the numeric values of this ratio for the first five levels in Table \ref{table:HnRatioValues}. The ratio between the bulk and boundary Hilbert spaces is nearly zero with an order of $10^{-170}$ when considering five levels of HaPPY code! This vast dominance of the size of the boundary Hilbert space is the reason why the bulk physics can be approximated only in the near-boundary region.

\begin{table}[H]
\arraycolsep=10pt\def\arraystretch{1.3}
    \centering
    \begin{tabular}{|c|l|}
    \hline
    Level $(n)$ & Ratio $\left(\frac{\mathrm{dim} \calh_n}{\mathrm{dim} \widetilde{\calh}_n}\right)$ \\
    \hline
    \hline
    1 & 0.0625 \\
    2 & 0.0000610352 \\
    3 & $5.684341886080802\times 10^{-14}$ \\
    4 & $4.37906\times 10^{-47}$ \\
    5 & $1.65608\times 10^{-170}$ \\
    \hline
    \end{tabular}
\caption{The ratio between the dimensions of the bulk and boundray Hilbert spaces $\left(\frac{\mathrm{dim} \calh_n}{\mathrm{dim} \widetilde{\calh}_n}\right)$ decreases fastly to zero as the level increases. By the fifth level, the value is incredibly small, on the order of $10^{-170}$.}
    \label{table:HnRatioValues}
\end{table}

\begin{figure}[H]
\centering
\includegraphics[scale=0.63]{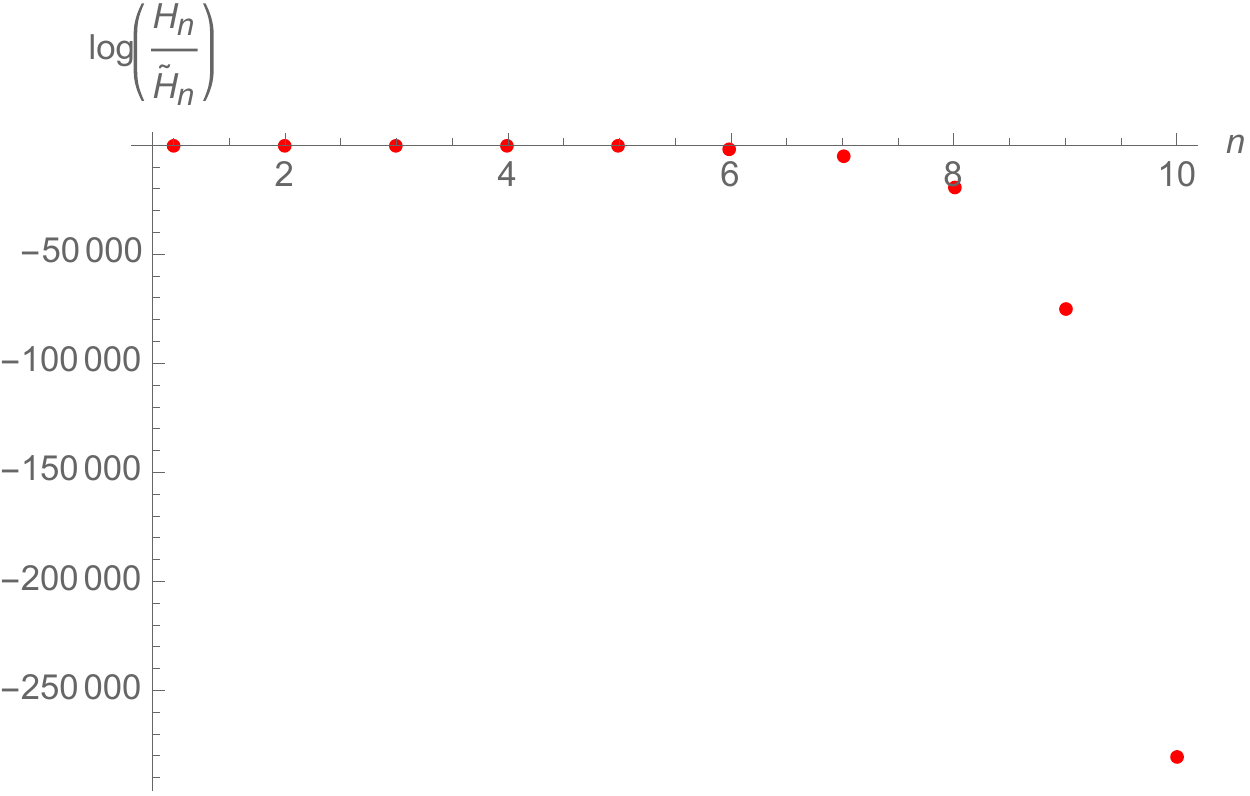}
\includegraphics[scale=0.63]{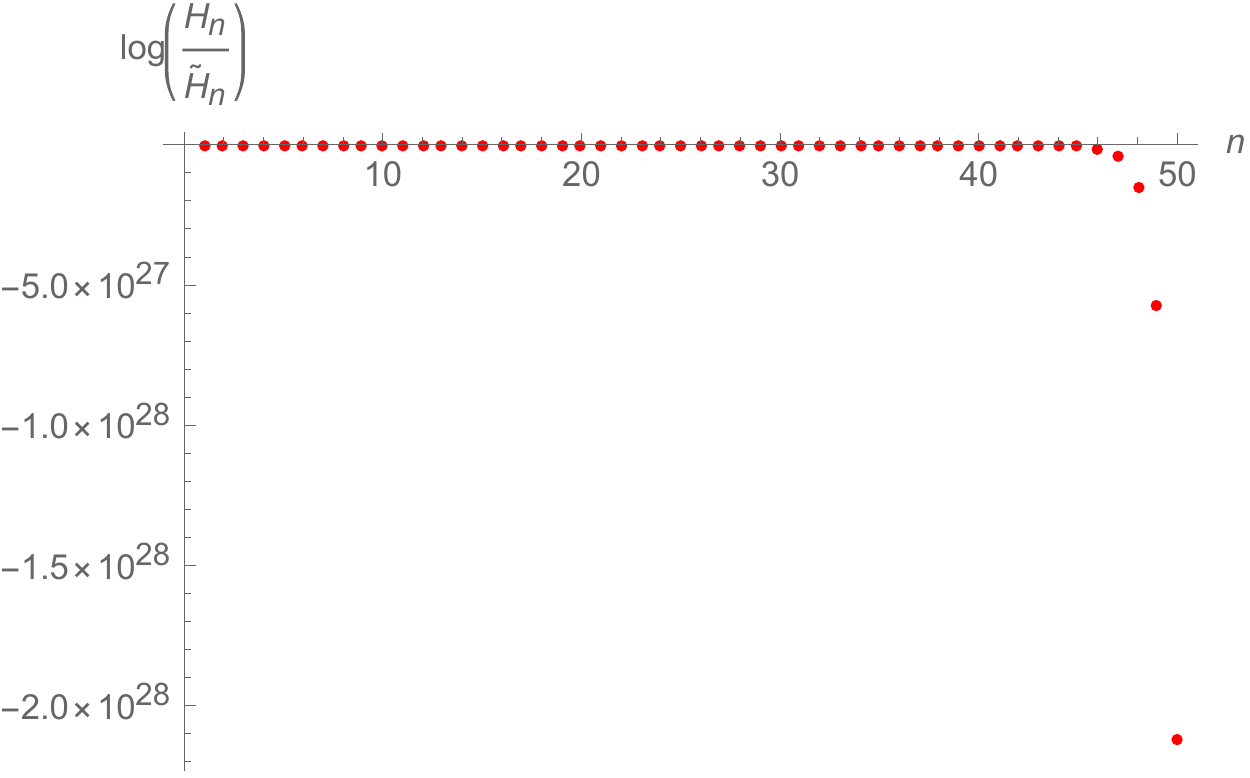}
\vspace{5mm}
\caption{The plots represent the ratio between the dimensions of the bulk and boundary Hilbert spaces as level grows, on a logarithmic scale. This demonstrates how quickly the ratio converges to zero, as the log dramatically diverges negatively.}
\label{fig:dimHratio}
\end{figure}

It is interesting to study the rate of convergence to zero of this ratio. We find that it is nearly double exponentially decaying to zero. We can see that from Figure \ref{fig:dimHratiologlog} where the double logarithmic value of the ratio between the dimensions of $\calh$ and $\widetilde{\calh}$ is plotted as the level $n$ gets larger. In fact, we quickly find an approximately linear graph as $n$ increases. 

\begin{figure}[H]
    \centering
    \includegraphics[scale=0.63]{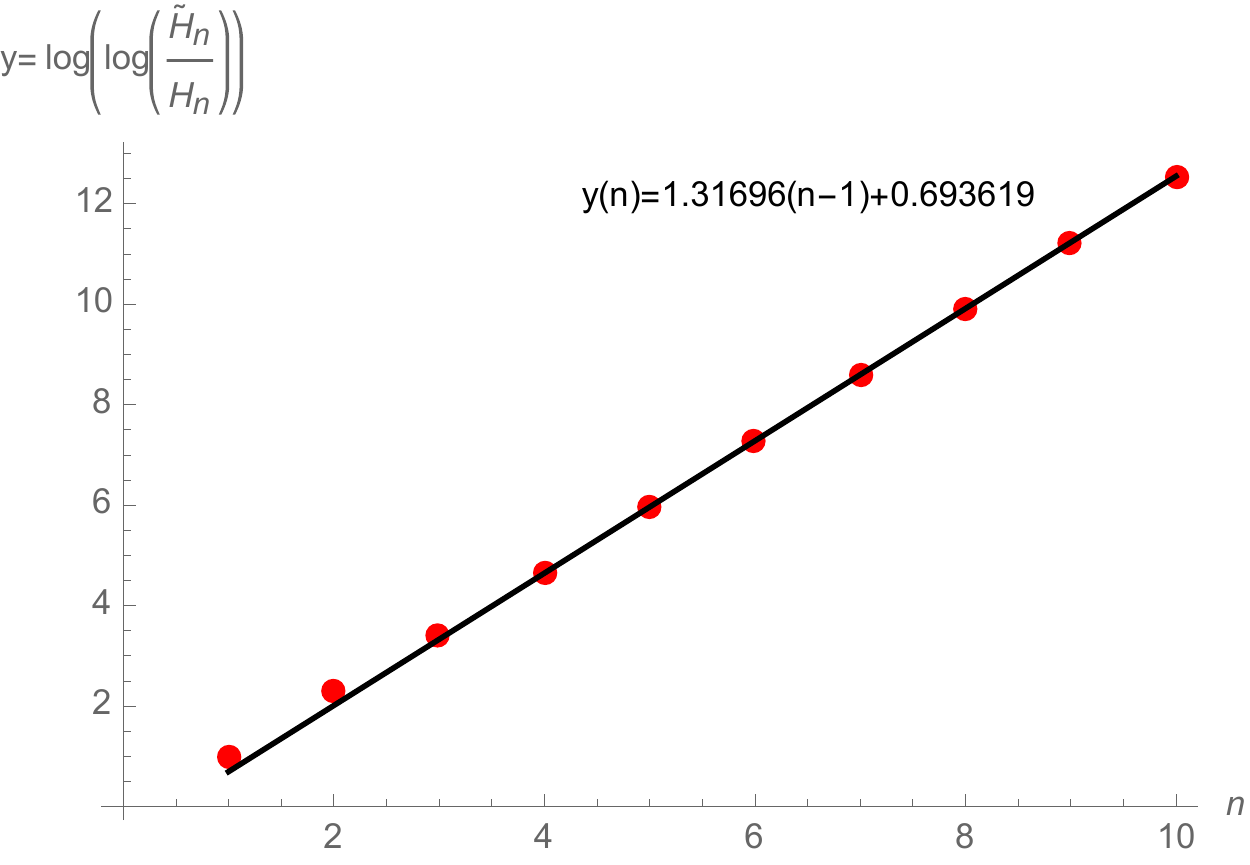}
    \includegraphics[scale=0.63]{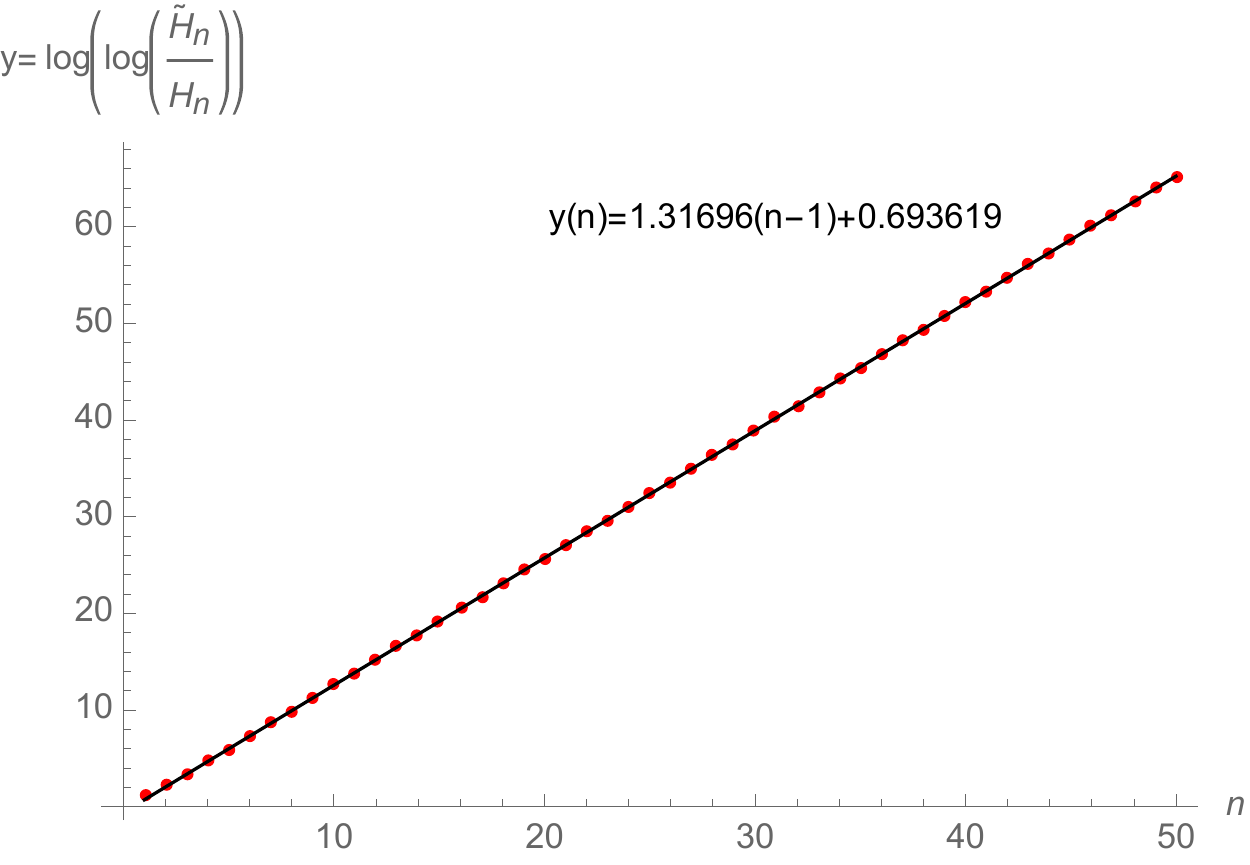}
    \caption{The red points are the plots of the ratio between the dimensions of the bulk and boundary Hilbert spaces $\frac{\mathrm{dim} \widetilde{\calh}_n}{\mathrm{dim} \calh_n}$ in double logarithmic scale as level $n$ grows. We see that the double-logarithmic ratio gets approximately linear with respect to the level; we can conclude that the ratio $\frac{\mathrm{dim} \widetilde{\calh}_n}{\mathrm{dim} \calh_n}$ is nearly double-exponentially decaying as the level grows. The black lines are the fitted linear plots of the ratio with respect to $n$.}
    \label{fig:dimHratiologlog}
\end{figure}

We can investigate this nearly-linear relation of the double-logarithmic ratio in an algebraic manner. Recall equation \eqref{eq:HHtRatio} and take its inverse, so that we can take the double-logarithmic function of the ratio. Then we have
\begin{align}
\log\log \frac{\mathrm{dim} \widetilde{\calh}_n}{\mathrm{dim} \calh_n}=\log\log \left( 2^{\frac{5}{\sqrt{3}}\left(\sqrt{3}+\left(2+\sqrt{3}\right)^{n-1}-\left(2-\sqrt{3}\right)^{n-1}\right)} \right).
\end{align}
With logarithmic properties, we expand:
\begin{align}
\log\log \frac{\mathrm{dim} \widetilde{\calh}_n}{\mathrm{dim} \calh_n} =\log \frac{5}{\sqrt{3}}+\log \left(\sqrt{3}+\left(2+\sqrt{3}\right)^{n-1}-\left(\frac{1}{2+\sqrt{3}}\right)^{n-1}\right) +\log(\log 2) .
\end{align}
We want to find the leading order terms of this expression. For such, we factorize the argument of the log in the second term as
\begin{align}
\log\log \frac{\mathrm{dim} \widetilde{\calh}_n}{\mathrm{dim} \calh_n} =
\log \frac{5}{\sqrt{3}}+\log \left(\frac{\left(\left(2+\sqrt{3}\right)^{n-1}+\frac{\sqrt{3}+\sqrt{7}}{2}\right) \left(\left(2+\sqrt{3}\right)^{n-1}+\frac{\sqrt{3}-\sqrt{7}}{2}\right)}{\left(2+\sqrt{3}\right)^{n-1}} \right) +\log(\log 2).
\end{align}
Then with more logarithmic properties and expansions, we get
\begin{align}
\begin{split}
\log\log \frac{\mathrm{dim} \widetilde{\calh}_n}{\mathrm{dim} \calh_n} &=
\log \frac{5}{\sqrt{3}}+\log\left(\left(2+\sqrt{3}\right)^{n-1}+\frac{\sqrt{3}+\sqrt{7}}{2}\right)+\log\left(\left(2+\sqrt{3}\right)^{n-1}+\frac{\sqrt{3}-\sqrt{7}}{2}\right)\\
&\qquad -\log\left(2+\sqrt{3}\right)^{n-1}+\log(\log 2) \\
&=\log \frac{5}{\sqrt{3}}+\log\left(2+\sqrt{3}\right)^{n-1}+\log\left(1+\frac{\sqrt{3}+\sqrt{7}}{2}\left(2-\sqrt{3}\right)^{n-1}\right)\\
&\qquad+\log\left(1+\frac{\sqrt{3}-\sqrt{7}}{2}\left(2-\sqrt{3}\right)^{n-1}\right) +\log(\log 2) \\
&=(n-1)\log\left(2+\sqrt{3}\right)+\sum_{k=1}^\infty \frac{\left(2-\sqrt{3}\right)^{k(n-1)}}{k} \left( \left(\frac{\sqrt{3}+\sqrt{7}}{2}\right)^k+\left(\frac{\sqrt{3}-\sqrt{7}}{2}\right)^k \right)\\
&\qquad+\log \frac{5}{\sqrt{3}}+\log(\log 2).
\end{split}
\end{align}
The second term is small relative to the other terms and vanishes as $n$ goes to infinity. Thus we find that, even for relatively small values of $n$, the double logarithmic ratio of the dimensions of the Hilbert spaces between bulk and boundary is well approximated by the linear function:
\begin{align}
\begin{split}
\log\log \frac{\mathrm{dim} \widetilde{\calh}_n}{\mathrm{dim} \calh_n}\quad\xrightarrow{n\rightarrow\infty}\quad & (n-1)\log\left(2+\sqrt{3}\right)+\log \frac{5}{\sqrt{3}}+\log(\log 2) \\
&\quad\sim 1.31696 (n-1)+0.693619 .
\end{split}
\end{align}
The numerical values we found exactly matches the fit we found in Figure \ref{fig:dimHratiologlog}. We can thus definitely state the ratio to be dominated by
\begin{align}
\frac{\mathrm{dim} \calh_n}{\mathrm{dim} \widetilde{\calh}_n}\quad\sim\quad 2^{-\frac{5}{\sqrt{3}} \left(2+\sqrt{3}\right)^{n-1}} ,
\end{align}
which converges quickly to zero.

\section{Isometries between Hilbert spaces}
\label{sec:isometryH}

In this section, we construct a mathematical setup for the boundary (physical) and bulk (code) Hilbert spaces of an infinite-dimensional HaPPY code with its boundary pushed to infinity. We also give an explicit example of a system of isometries, constructed utilizing Bell pairs, between the different bulk layers that lifts to the boundary in a way that respects the quantum error correcting structure of the HaPPY code.

\subsection{Setup: the maps between Hilbert spaces}
Let $\calh_n$ for $n \in \mathbb{N}$ denote the Hilbert space of the black (bulk) qubits at level $n$. For all $n \in \mathbb{N}$, $\calh_n$ is identified with (or, mapped into) a subspace of $\calh_{n+1}$.
The relationship between the $\calh_n$ Hilbert spaces is expressed as
\[ \calh_1 \rightarrow \calh_2 \rightarrow \calh_3 \rightarrow \cdots \] 
where each right arrow ($\rightarrow$) denotes the map from one Hilbert space into the next Hilbert space. We will call $\calh_n$ the {\it code pre-Hilbert space at level $n$}.

For each $\calh_n$, we can use the tensor network at level $n$ to map the black (bulk) qubits at level $n$ into their corresponding white (boundary) qubits. The white (boundary) qubits at level $n$ make up the {\it physical pre-Hilbert space at level $n$}, or $\widetilde{\calh}_n$. We represent the isometry produced by the tensor network with a down arrow ($\downarrow$), so the relationship between $\calh_n$ and $\widetilde{\calh}_{n}$ is given by
\[ \begin{array}{c}
\text{\it code pre-Hilbert space at level }n\quad : \\ 
\\ 
\text{\it physical pre-Hilbert space at level }n\ :
\end{array}\
\begin{array}{c}
\calh_n \\ 
\downarrow \\ 
\widetilde{\calh}_n
\end{array} 
 \]
We also want a way to isometrically map a state in $\widetilde{\calh}_n$ into a state in $\widetilde{\calh}_{n + 1}$. This is done by putting tensors in all tiles of level $n+1$ and using the tensor network to map the white qubits in level $n$ to the white qubits in level $n+1$. 
Moreover, in order for the map to behave correctly with the quantum error correcting structure of the HaPPY code, we want the following diagram to be commutative\footnote{For example, one can see that if the bulk indices of these tensors (which correspond to black qubits) are put into a reference state, then the diagram is clearly commutative.}
\begin{align}
\begin{array}{ccccccc}
	\text{\it code pre-Hilbert spaces}\quad : & \calh_1 & \rightarrow & \calh_2 & \rightarrow & \quad\cdots \\ 
	& \downarrow &  & \downarrow &  & \\ 
	\text{\it physical pre-Hilbert spaces}\ : & \widetilde{\calh}_1 & \rightarrow & \widetilde{\calh}_2 & \rightarrow & \quad\cdots & \quad.
\end{array} 
\label{eq:commutativediagram}
\end{align}

For both the code and physical pre-Hilbert spaces, the right arrow ($\rightarrow$) means that one Hilbert space is isometrically mapped into, and identified with a subspace of, the next Hilbert space.
The {\it code pre-Hilbert space} (for the infinite-dimensional case) is the disjoint union of all the $\calh_n$ subject to the identifications given above.\footnote{Mathematically, we are taking the inductive limit of the direct system of Hilbert spaces.} The {\it physical pre-Hilbert space} is defined as the disjoint union of all the $\widetilde{\calh}_n$ subject to the identifications given above.

%A bulk operator on $\calh_n$ can be mapped to a bulk operator on $\calh_{n+1}$ that only acts nontrivially on the qubits in level $n+1$ that are also in level $n$. A bulk operator on $\calh_n$ can also be mapped to a boundary operator on $\widetilde{\calh}_n$, though there are many ways to do this.\footnote{Alternatively, quantum systems can be constructed directly at the operator level. We study the HaPPY code directly at the level of operator pushing in our companion paper \cite{MonicaElliott2}.}

An advantage of this picture is that it reproduces bulk physics all the way to the AdS radius, which can be of interest when thinking of deep bulk objects like black holes. On the other hand, it is not clear if the successive mappings of a bulk operator will converge in some sense all the way up to the boundary. Moreover, the physical interpretation of this picture is somewhat questionable, as the HaPPY code is better suited to describe physics close to the boundary.

Table \ref{table:dimensions} relates the dimensions of bulk and boundary Hilbert spaces for the first few levels of the infinite-dimensional HaPPY code.

\begin{table}[H]
\arraycolsep=8pt\def\arraystretch{1.4}
\begin{center}
$
\begin{array}{|c|c||c|c|c|c|c|c|c|}
\hline 
\text{Level} & n & 1 & 2 & 3 & 4 & 5 & 6 & \cdots \\ 
\hline
\hline
\text{Bulk Hilbert space} & \mathrm{dim} \calh_n & 2^1 & 2^{11} & 2^{51} & 2^{201} & 2^{761} & 2^{2851} & \cdots \\ 
\hline 
\text{Boundary Hilbert space} & \mathrm{dim} \widetilde{\calh}_n & 2^5 & 2^{25} & 2^{95} & 2^{355} & 2^{1325} & 2^{4945} & \cdots \\ 
\hline 
\end{array} 
$
\end{center}
\caption{Dimensions of bulk and boundary Hilbert spaces at various levels.}
\label{table:dimensions}
\end{table}

\subsection{Bulk Hilbert Space Maps}

The HaPPY code at level $n$ provides a map between $\calh_n$ and $\widetilde{\calh}_{n}$; however, the tensor network does not provide a map between two successive bulk Hilbert spaces or, alternatively, two boundary Hilbert spaces. A natural guess for a map between successive bulk Hilbert spaces is to take a state in $\mathcal{H}_n$ and take a tensor product with the requisite number of extra qubits in a reference state (say, $\ket{\uparrow}$), to make a state inside of $\mathcal{H}_{n+1}$. This defines an injective isometric map between $\mathcal{H}_n$ and $\mathcal{H}_{n+1}$. However, the problem with this approach or alike is that the corresponding boundary-to-boundary map in the commutative diagram, given in equation \eqref{eq:commutativediagram}, will be extremely complicated and is not expected to provide a realization of long-range entanglement. In this section, we will give an alternative way of constructing a map between the two successive bulk Hilbert spaces $\calh_n$ and $\calh_{n+1}$, in such a way that it lifts nicely to the boundary. This will give an example that realizes the commutative diagram given in equation \eqref{eq:commutativediagram} both at the level of the bulk and at the level of the boundary.

To construct the bulk to bulk map, we first investigate level $1$ of the HaPPY code. The HaPPY code does not provide a natural way to directly map 1 central qubit to 11 qubits, which would correspond to the number of level $1$ bulk qubits to the number of level $2$ bulk qubits. Instead, it just gives a way of mapping the central qubit to the 5 neighboring boundary qubits. Therefore, there are 6 qubit degrees of freedom left, that correspond to the central node and the 5 other level 2 nodes that are not on the first boundary level.  To define the map, we will need to introduce a Bell pair. Recalling that the tensors of the HaPPY code are perfect tensors, and they maximally entangle any bi-partition of their legs. Indeed, the central tensor is supposed to be maximally entangled to the rest of the network, and it sounds natural to encode that information in the remaining qubits. As a result, if a qubit lives on the central tile of the HaPPY code, then it will be maximally entangled with its image by the bulk isometry. As such, these two elements of the Hilbert space can be identified with a Bell pair. Consequently, by successively grabbing bulk Bell pairs, and pushing some of their elements through the tensor network, we will be able to reconstruct the Hilbert spaces of the HaPPY code. Note that the HaPPY tensor network sometimes pushes two legs into the same tile, so it is not surprising that we are introducing Bell pairs which are pushed through isometries defined by the tensor network.

Using these ingredients, we present a natural way using Bell pairs to define an isometry from $\calh_{n}$ to $\calh_{n+1}$ for $n \in \mathbb{N}$. We briefly summarize the procedure and then give examples.
 \medskip
\begin{alg} \label{alg:HaPPYbulk}
The bulk to bulk map $\calh_n\to\calh_{n+1}$ of the HaPPY code at level $n\in\mathbb{N}$ is constructed by the following. The initial data of the algorithm $\calh_1 \rightarrow \calh_2$ ($1 \rightarrow 11$) is
\begin{subequations}
\begin{align*}
\calh_1 \rightarrow \calh_2\ :\quad 
1 & \xrightarrow{\hspace*{.4cm}\calh_1 \rightarrow \widetilde{\calh}_1\hspace*{.4cm}} 5 \\
&\left\{ \begin{array}{l}
	1 \xrightarrow{\hspace*{1.15cm}} 1\\ 
	1 \xrightarrow{\hspace*{.1cm}\calh_1 \rightarrow \widetilde{\calh}_1} 5 ,
	\end{array} \right.  
\end{align*}
\end{subequations}
where the last two rows incorporate a Bell pair.
For the rest, $n\geq 2$, the map $\calh_{n} \rightarrow \calh_{n+1}$ is constructed differently for $n$ even and $n$ odd except for the top row, which is always the map $\calh_n\to\widetilde{\calh}_n$.
\begin{itemize}
\item {\bf $\mathbf{n}$ even}: We start with the map $\calh_{n-1} \rightarrow \calh_n$ and advance the outputs of all Bell pairs horizontally as in the commutative diagram ($\calh_{k} \rightarrow \calh_{k+1}$ or $\widetilde{\calh}_{k} \rightarrow \widetilde{\calh}_{k+1}$ for $k<n$) except the last row of the last Bell pair. Then, we add a new Bell pair and map it under $\calh_1\to\widetilde{\calh}_1$:
$$ \left\{ \begin{array}{c}
1 \longrightarrow 5\\ 
1 \longrightarrow 5
\end{array} \right. $$
\item {\bf $\mathbf{n}$ odd}: We start with the map $\calh_{n-1} \rightarrow \calh_n$ and advance the outputs of all Bell pairs horizontally except the last row of the last Bell pair.
\end{itemize}
\end{alg}

We explore first few levels of this algorithm as examples. To start the algorithm, we require the first map $\calh_1 \rightarrow \calh_2$ ($1 \rightarrow 11$) as the initial data in Algorithm \ref{alg:HaPPYbulk}. The top line represents the input qubit mapped to five qubits via the HaPPY tensor network, which corresponds to $\calh_1 \rightarrow \widetilde{\calh}_1$. The second line introduces a maximally entangled Bell pair, and one of the qubits is mapped to five qubits via the tensor network. The total output is 11 qubits.

Now we look into $\calh_2 \rightarrow \calh_3$ ($11 \rightarrow 51$), which is under the category of $n$ even. The algorithm provides this map to be
\begin{subequations}
\begin{align}
\calh_2 \rightarrow \calh_3\ :\quad 
11 & \xrightarrow{\hspace*{1.1cm}} 25 \\
&\left\{ \begin{array}{c}
1 \longrightarrow 11\\
\hspace{-0.18cm} 1 \longrightarrow 5
\end{array} \right.
\label{eq:bp2}
\\
&\left\{ \begin{array}{c}
1 \longrightarrow 5\\ 
1 \longrightarrow 5
\end{array} \right. 
\label{eq:bp3} 
\end{align}
\end{subequations}
The rule is that the top line of the map always maps the input qubits, which correspond to 11 qubits in this case, to the corresponding number of qubits in the same column of Table \ref{table:dimensions}, which corresponds to 25 qubits in this case. In other words, this is the map for $\calh_2 \rightarrow \widetilde{\calh}_2$. The output of the first row of the Bell pair for $\calh_1\to\calh_2$ is advanced horizontally according to Table \ref{table:dimensions}, which corresponds to $\calh_1 \rightarrow \calh_2$, and the result is copied into equation \eqref{eq:bp2}. Furthermore, a new Bell pair is added in equation \eqref{eq:bp3}, and each qubit is mapped to five qubits, using $\calh_1 \rightarrow \widetilde{\calh}_1$ individually.

For $\calh_3 \rightarrow \calh_4$ ($51 \rightarrow 201$), the map is given by:
\begin{subequations}
\begin{align}
\calh_3 \rightarrow \calh_4\ :\quad 
51 & \xrightarrow{\hspace*{1.1cm}} 95 \\
&\left\{ \begin{array}{c}
1 \longrightarrow 51\\ 
1 \longrightarrow 25
\end{array} \right.
\label{eq:bp4}
\\
&\left\{ \begin{array}{c}
1 \longrightarrow 25\\ 
\hspace{-0.18cm} 1 \longrightarrow 5
\end{array} \right. 
\label{eq:bp5} 
\end{align}
\end{subequations}
To arrive at equations \eqref{eq:bp4} and \eqref{eq:bp5} from equations \eqref{eq:bp2} and \eqref{eq:bp3}, we horizontally advance (i.e. move one cell to the right in Table \ref{table:dimensions}) the output of every row in equations \eqref{eq:bp2} and \eqref{eq:bp3} except for the last row of \eqref{eq:bp3}.
	
To obtain further isometries, we repeat the pattern of the previous two steps. For $\calh_4 \rightarrow \calh_5$ ($201 \rightarrow 761$), the map is given by:
\begin{subequations}
\begin{align}
\calh_4 \rightarrow \calh_5\ :\quad 
201 & \xrightarrow{\hspace*{1.1cm}} 355 \\
&\left\{ \begin{array}{c}
1 \longrightarrow 201\\ 
\hspace{-0.18cm} 1 \longrightarrow 95
\end{array} \right.
\label{eq:bp6}
\\
&\left\{ \begin{array}{c}
1 \longrightarrow 95\\ 
\hspace{-0.18cm} 1 \longrightarrow 5
\end{array} \right. 
\label{eq:bp7}
\\
&\left\{ \begin{array}{c}
1 \longrightarrow 5\\ 
1 \longrightarrow 5
\end{array} \right. 
\label{eq:bp8} 
\end{align}
\end{subequations}
To determine the output of the first two Bell pairs, we horizontally advanced the outputs of equations \eqref{eq:bp4} and \eqref{eq:bp5} (except for the last line). Then, we introduced another Bell pair.

For $\calh_5 \rightarrow \calh_6$ ($761 \rightarrow 2851$), the map is given by:
\begin{subequations}
\begin{align}
\calh_5 \rightarrow \calh_6\ :\quad 
761 & \xrightarrow{\hspace*{1.1cm}} 1325 \label{eq:toplineex}\\
&\left\{ \begin{array}{c}
1 \longrightarrow 761\\ 
1 \longrightarrow 355
\end{array} \right.
\label{eq:bp9}
\\
&\left\{ \begin{array}{c}
1 \longrightarrow 355\\ 
\hspace{-0.18cm} 1 \longrightarrow 25
\end{array} \right. 
\label{eq:bp10}
\\
&\left\{ \begin{array}{c}
1 \longrightarrow 25\\ 
\hspace{-0.18cm} 1 \longrightarrow 5
\end{array} \right. 
\label{eq:bp11} 
\end{align}
\end{subequations}
This results from horizontally advancing the outputs of every row in equations \eqref{eq:bp6} to \eqref{eq:bp8} except the last row of equation \eqref{eq:bp8}.
	
For $\calh_6 \rightarrow \calh_7$ ($2851 \rightarrow 10651$), we repeat the pattern we have already established:
\begin{subequations}
\begin{align}
\calh_6 \rightarrow \calh_7\ :\quad 
2851 & \xrightarrow{\hspace*{1.1cm}} 4945 \\
&\left\{ \begin{array}{c}
1 \longrightarrow 2851\\ 
1 \longrightarrow 1325
\end{array} \right.
\label{eq:bp12}
\\
&\left\{ \begin{array}{c}
1 \longrightarrow 1325\\ 
\hspace{-0.38cm} 1 \longrightarrow 95
\end{array} \right. 
\label{eq:bp13}
\\
&\left\{ \begin{array}{c}
1 \longrightarrow 95\\ 
\hspace{-0.18cm} 1 \longrightarrow 5
\end{array} \right. 
\label{eq:bp14} 
\\
&\left\{ \begin{array}{c}
1 \longrightarrow 5\\ 
1 \longrightarrow 5
\end{array} \right. 
\label{eq:bp15} 
\end{align}		
\end{subequations}

Following the algorithm for $\calh_n\to\calh_{n+1}$, the output, after incorporating the number of bulk qubits $N_n$ of $\calh_n$ along with the Bell pairs, sums up to be the number of qubits $N_{n+1}$ of $\calh_{n+1}$. This may appear non-trivial from Algorithm \ref{alg:HaPPYbulk}, however, we can show that the algorithm is giving a map consistent with $N_n$ and $\widetilde{N}_n$ found for the HaPPY code in equations \eqref{eq:Nn} and \eqref{eq:Nntilde}. In fact, the algorithm dictates the following relations on $N_n$ and $N_{n+1}$ for $n$ odd and even:
\begin{align}
\begin{aligned}
& n=2k+1:&& N_{n}=N_{2k+1}=\widetilde{N}_{2k}+N_{2k}+2\sum_{i=0}^{k-1}\widetilde{N}_{2i+1}+5,\\
& n=2k+2: && N_{n}=N_{2k+2}=\widetilde{N}_{2k+1}+N_{2k+1}+2\sum_{i=1}^{k}\widetilde{N}_{2i}+5,
\end{aligned}
\end{align}
where $k\geq 1$, $k\in\mathbb{N}$. The above recursive relations of $N_n$ and $\widetilde{N}_n$ established in Algorithm \ref{alg:HaPPYbulk} are satisfied with the equations \eqref{eq:Nn} and \eqref{eq:Nntilde}, which provides $N_n$ and $\widetilde{N}_n$ for the HaPPY code. Hence, we can deduce that the algorithm works as desired, mapping bulk to bulk recursively. It is quite remarkable that the bulk (and thus boundary as in Section \ref{sec:boundaryalgorithm}) Hilbert spaces for the HaPPY code can be related at different levels via such a Bell pair construction, indeed it would be interesting to explore in more detail the physical ramifications of this map.

\subsection{Boundary Hilbert Space Maps} \label{sec:boundaryalgorithm}

For an arbitrary level $n \in \mathbb{N}$, the map between the boundary Hilbert spaces $\widetilde{\calh}_n$ and $\widetilde{\calh}_{n+1}$ is directly determined by the map between the bulk Hilbert spaces $\calh_{n}$ and $\calh_{n+1}$. This ensures the commutativity of the diagram given in equation \eqref{eq:commutativediagram}. In fact, we construct the boundary to boundary map based on the HaPPY code bulk to boundary map $\calh_n\to\widetilde{\calh}_n$ and the established bulk to bulk map $\calh_n\to\calh_{n+1}$, as presented in Algorithm \ref{alg:HaPPYboundary}.

\begin{alg} \label{alg:HaPPYboundary}
The boundary to boundary map $\widetilde{\calh}_n\to\widetilde{\calh}_{n+1}$ of the HaPPY code at level $n\in\mathbb{N}$ is constructed by the following. We first perform Algorithm \ref{alg:HaPPYbulk} for the bulk to bulk map $\calh_n\to\calh_{n+1}$. We first send the input $\widetilde{N}_n$ of $\widetilde{\calh}_n$ to $\widetilde{N}_n$ of $\calh_{n+1}$. For the rest, we copy the outputs of all Bell pairs from $\calh_n\to\calh_{n+1}$. We collect them and then apply the map from $\calh_{n+1} \rightarrow \widetilde{\calh}_{n+1}$.
\end{alg}

For example, we consider the case when $n=5$, namely, the boundary to boundary map $\widetilde{\calh}_5 \rightarrow \widetilde{\calh}_6$. The top line of the map is replaced with a line where the output number of qubits is mapped identically to itself: $1325 \rightarrow 1325$. We next copy all Bell pairs from $\calh_5\to\calh_6$, given in \cref{eq:bp9,eq:bp10,eq:bp11}. Then, the total number of output qubits equals the number of qubits in $\calh_{n+1}$.  Finally, we advance this through the HaPPY code bulk to boundary map $\calh_{n+1}\to\widetilde{\calh}_{n+1}$, as we can summarize the step as
\begin{equation}
\left.\begin{array}{l l}
\widetilde{\calh}_5 \rightarrow \widetilde{\calh}_6\ :\quad 
1325 & \xrightarrow{\hspace*{1.1cm}} 1325 \\
& \left\{ \begin{array}{c}
1 \longrightarrow 761\\ 
1 \longrightarrow 355
\end{array} \right.
\\
& \left\{ \begin{array}{c}
1 \longrightarrow 355\\ 
\hspace{-0.18cm} 1 \longrightarrow 25
\end{array} \right. 
\\
& \left\{ \begin{array}{c}
1 \longrightarrow 25\\ 
\hspace{-0.18cm} 1 \longrightarrow 5
\end{array} \right. 
\end{array} \right] \longrightarrow 4995 .
\end{equation}

We conclude that via Algorithms \ref{alg:HaPPYbulk} and \ref{alg:HaPPYboundary}, we have established the bulk to bulk map and the boundary to boundary map recursively; along with the HaPPY code bulk to boundary isometry, the commutative diagram of Hilbert spaces has been constructed
\begin{align}
\begin{array}{cccccc}
	 \calh_n & \rightarrow & \calh_{n+1}& \\
	 \downarrow &  & \downarrow &  & \\ 
	 \widetilde{\calh}_n & \rightarrow & \widetilde{\calh}_{n+1} & .
\end{array} 
\label{eq:CDleveln}
\end{align}
This provides the commutative diagram for all levels of the HaPPY code in equation \eqref{eq:commutativediagram} when $n$ is recursively applied. Having this quantum error correcting code requires investigating its role as a stabilizer code at each level, which we study in Section \ref{sec:BoundaryStabilizerGen}.

Knowing such maps is useful for determining the correlation functions in the limit where the tensor network is infinitely large, which is important for computing physical observables via holographic tensor networks. Even though this indeed provides a rigid recursive understanding of isometries between Hilbert spaces, this however does not provide a description associated to Hamiltonian theories. For understanding thermal states of the associated quantum field theory, we study Hamiltonian descriptions of the bulk from the infinite-dimensional HaPPY code in Sections \ref{sec:trapezeHamiltonian} and \ref{sec:bulktoboundarytrapeze}. We will define a natural Hamiltonian from the bulk and explore the patterns of its maps at the level of operators.

\section{Stabilizer generators at each level}
\label{sec:BoundaryStabilizerGen}

The next task is to construct some explicit operator maps from the bulk to the boundary. It will enable us to test whether the HaPPY code really keeps its promises as a model of AdS/CFT: is an entangled bulk theory mapped to a sufficiently entangled theory on the boundary that it exhibits a conformal algebraic decay in its correlation functions?

The HaPPY code is in its essence a stabilizer code, at least for the historic choice of tensors coming from the five-qubit code. It then seems natural to start our explicit study of operator pushing by looking at its stabilizer.
We want to find some explicit construction for a nice set of stabilizer generators of the HaPPY code at every level, in terms of the action of the Pauli group on boundary qubits. We take the conventions for the Pauli operators to be
\begin{align}
X=\begin{bmatrix}0&1\\1&0\\\end{bmatrix},\quad Y=\begin{bmatrix}0&-i\\i&0\\\end{bmatrix},\quad Z=\begin{bmatrix}1&0\\0&-1\\\end{bmatrix}.
\label{eq:Paulis}
\end{align}
Then the Pauli group\footnote{Technically, in order for this to have a group structure, one would need to consider these elements up to powers of $i$, but the name is usual in quantum error correction literature.} is the set
\begin{align}
\Pi=\{X,Y,Z,1\} .
\end{align}
Note that these operators act on a single qubit, which corresponds to a state in a two-dimensional Hilbert space. Then for a quantum stabilizer code over more than one qubit, one would like to express stabilizer generators as tensor products of Pauli matrices.

The original perfect tensor used in the HaPPY paper is constructed utilizing the 5-qubit code. In what follows, one should bear in mind that the stabilizer of the 5-qubit code is generated by the cyclic permutations of the string of Pauli operators:
\begin{align}
S=YXXY1,
\end{align}
where we have suppressed the tensor product symbols ($\otimes$) between Paulis.
We note that this $S$ is usually taken to be of the form $XZZX1$ in the literature. However one can go from this expression to ours by multiplying it by $ZZX1X$, which gives $Y1YXX$, and then doing a circular permutation of the Paulis. The converse procedure is also valid.

We are first interested in knowing how many stabilizer generators we need for each level of the code. For this we first need to utilize the number of boundary qubits $\widetilde{N}_n$ in a setting that is friendly to the graphic pattern of the HaPPY code. We observe the HaPPY code in terms of the bulk nodes individually linking to the boundary nodes. With the exception of the first level, the most outside level of the bulk nodes each provides either two or three boundary nodes, which we call {\it 2-clusters} and {\it 3-clusters}. Let us denote the number of 2-clusters (2 boundary legs for the pentagons) as $a_n$ and that of 3-clusters (3 boundary legs for the pentagons) as $b_n$, where $n\geq2, n\in\mathbb{N}$.\footnote{For the consistency of the notation for $n$ as the level defined in Section \ref{sec:pentagon}, $a_n$ and $b_n$ sequences start from $n=2$.} The first level, however,  is special in terms of bulk and boundary nodes, where there are no 2- or 3-clusters for the boundary qubits, since there is a single center node with five boundary nodes directly. Then, as Figure \ref{fig:level1tensornetwork} depicts, we can safely determine the recursive relations of $a_n$ and $b_n$ to be \begin{align}
\begin{cases}
& a_n=2a_{n-1}+3b_{n-1},\quad a_2=5, \\
& b_n=a_{n-1}+2b_{n-1},\quad b_2=5.
\end{cases}
\end{align}
This 2-cluster and 3-cluster series $a_n$ and $b_n$ are directly correlated with the counting of of edges and vertices observed earlier in Section \ref{sec:AsymototesHn}. The counting of edges ($e_n$) and vertices ($v_n$) at level $n$ satisfy recursive relations given in equation \eqref{eq:edgesvertices}, which originates from the $SL(2,\mathbb{Z})$ symmetry. For the consistency of the notation for $n$, the number of 2-clusters $a_n$ and that of 3-clusters $b_n$ correspond to the number of edges $e_{n-1}$ and that of vertices $v_{n-1}$, respectively.

To put each level of the HaPPY code as a stabilizer code, we write the number of physical qubits $p_n$ at level $n$ as
\begin{align}
p_n=\widetilde{N}_n=2a_n+3b_n=a_{n+1},
\end{align}
and the number of logical qubits $\ell_n$ as
\begin{align}
\ell_{n+1}-\ell_n=N_{n+1}-N_n=a_{n+1}+b_{n+1}.
\end{align}
The stabilizer code we have setup here is the series of $[p_n=\widetilde{N}_n, \ell_n=N_n]$ quantum error correcting codes defined at each level $n$, encoding $\ell_n$ logical qubits and $p_n$ physical qubits, whose rate is $\ell_n/p_n$. 

In other words, we have introduced a user-friendly graphical way to count the required qubits of the HaPPY code at each level via the number of 2-clusters ($a_n$) and that of 3-clusters ($b_n$) with their recursive relations 
\begin{align}
\begin{cases}
& \widetilde{N}_n=a_{n+1},\quad N_n=N_{n-1}+a_n+b_n,\\
& a_n=2a_{n-1}+3b_{n-1},\quad b_n=a_{n-1}+2b_{n-1}.
\end{cases}
\end{align}
We list for first few levels $n=2,3,4,5,6$ in Table \ref{table:anbnhnhtn} the numeric values of the number of 2-clusters $a_n$, the number of 3-clusters $b_n$, the number of logical qubits (or identically the number of bulk qubits) $\ell_n=N_n$, and the number of physical qubits (or identically the number of boundary qubits) $p_n=\widetilde{N}_n$.

\begin{table}[H]
\arraycolsep=8pt\def\arraystretch{1.4}
\begin{center}
$
\begin{array}{|c|c||c|c|c|c|c|c|}
\hline 
\text{Levels} & n & 2 & 3 & 4 & 5 & 6 & \cdots \\ 
\hline
\hline
\text{\# 2-clusters} & a_n & 5 & 25 & 95 & 355 & 1325 & \cdots \\ 
\hline
\text{\# 3-clusters} & b_n & 5 & 15 & 55 & 205 & 765 & \cdots \\ 
\hline
\text{\# Logical/Bulk qubits} & \ell_n=N_n & 11 & 51 & 201 & 761 & 2851 & \cdots \\ 
\hline
\text{\# Physical/Boundary qubits} & p_n=\widetilde{N}_n & 25 & 95 & 355 & 1325 & 4945 & \cdots \\ 
\hline 
\end{array} 
$
\end{center}
\caption{Numbers of 2-clusters ($a_n$) and 3-clusters ($b_n$) at various levels $n$. They form the logical ($\ell_n$) and physical qubits ($p_n$), corresponding to bulk ($N_n$) and boundary qubits ($\widetilde{N}_n$), respectively.}
\label{table:anbnhnhtn}
\end{table}

The stabilizer of the $n$-th level of the HaPPY code $\mathcal{S}$ is an abelian subgroup of the $\widetilde{N}_n$-fold Pauli group $\Pi^{\widetilde{N}_n}$:
\begin{align}
\Pi_n = \left\{\ e^{i\phi} A_1\otimes\cdots\otimes A_{\widetilde{N}_n}\ \middle|\ A_j\in\Pi\ \left(\forall j\in\{1,\cdots ,\widetilde{N}_n\}\right),\ \phi\in\left\{0,\frac{1}{2}\pi,\pi,\frac{3}{2}\pi\right\}\right\} .
\end{align}
The stabilizer $\mathcal{S}$ has a minimal representation in terms of $\widetilde{N}-N$ independent generators
\begin{align}
\left\{ g_1,\cdots,g_{\widetilde{N}-N}\right\},
\end{align}
where $\forall j\in\{ 1,\cdots,\widetilde{N}-N\}$, $g_j\in\mathcal{S}$.
We find the recursive relation for the number of stabilizer generators $g_n$ for our infinite-dimensional analog of the HaPPY code to be
\begin{align}
\begin{split}
g_{n+1}-g_n &=(\widetilde{N}_{n+1}-N_{n+1})-(\widetilde{N}_n-N_n)=\widetilde{N}_{n+1}-\widetilde{N}_n-(N_{n+1}-N_n)\\
&=(2a_{n+1}+3b_{n+1})-a_{n+1}-(a_{n+1}+b_{n+1})=2b_{n+1}.
\end{split}
\end{align}

The problem then boils down to finding $g_n$ generators for each $n$. In what follows, we present an algorithmic approach to determine the stabilizer generators recursively from level $n$ to level $n+1$.\footnote{For a review on stabilizer codes with similar methods, see \cite{Rubio}.}  First, we will define a nice bulk-to-boundary map that pushes the $g_n$ generators already known at level $n$ to the boundary of level $n+1$. Then we find a set of $2b_{n+1}$ independent new generators. The latter part is straightforward: a set of $2b_{n+1}$ generators is given by the two families described on Figure \ref{fig:FamiliesBoundary}. Each of the two family members at level $n+1$ can be centered on the central descendant of a 3-cluster at level $n$, and there are $b_n$ such central descendants, yielding the desired total of $2b_n$ independent generators.

\begin{figure}[H]
\begin{subfigure}{0.5\textwidth}
\begin{center}
\begin{tikzpicture}[scale=0.8]
\draw (-9,0)--(-1,0);
\draw (-8,0)--(-8.5,1);
\draw (-8,0)--(-7.5,1);
\draw (-5,0)--(-5.5,1);
\draw (-5,0)--(-4.5,1);
\draw (-2,0)--(-2.5,1);
\draw (-2,0)--(-1.5,1);
\draw (-8,0)--(-8,-.5);
\draw (-5,0)--(-5,-.5);
\draw (-2,0)--(-2,-.5);
\node[draw=none,label={[label distance=0.1mm]north:Z}] (C) at (-8.5,1) {};
\node[draw=none,label={[label distance=0.1mm]north:Z}] (C) at (-7.5,1) {};
\node[draw=none,label={[label distance=0.1mm]north:Z}] (C) at (-5.5,1) {};
\node[draw=none,label={[label distance=0.1mm]north:Z}] (C) at (-4.5,1) {};
\node[draw=none,label={[label distance=0.1mm]north:Z}] (C) at (-2.5,1) {};
\node[draw=none,label={[label distance=0.1mm]north:Z}] (C) at (-1.5,1) {};
\node[draw=none,label={[label distance=0.1mm]west:X}] (C) at (-9,0) {};
\node[draw=none,label={[label distance=0.1mm]east:X}] (C) at (-1,0) {};
\node[draw=none,label={[label distance=0.1mm]east:BULK}] (C) at (-5.85,-1) {};
\end{tikzpicture}
\end{center}
\vspace{3mm}

\subcaption{Family 1}
\end{subfigure}
\begin{subfigure}{0.5\textwidth}
\begin{center}
\begin{tikzpicture}[scale=0.8]
\draw (-9,0)--(-1,0);
\draw (-8,0)--(-8.5,1);
\draw (-8,0)--(-7.5,1);
\draw (-5,0)--(-5.5,1);
\draw (-5,0)--(-4.5,1);
\draw (-2,0)--(-2.5,1);
\draw (-2,0)--(-1.5,1);
\draw (-8,0)--(-8,-.5);
\draw (-5,0)--(-5,-.5);
\draw (-2,0)--(-2,-.5);
\node[draw=none,label={[label distance=0.1mm]north:X}] (C) at (-8.5,1) {};
\node[draw=none,label={[label distance=0.1mm]north:X}] (C) at (-7.5,1) {};
\node[draw=none,label={[label distance=0.1mm]north:X}] (C) at (-5.5,1) {};
\node[draw=none,label={[label distance=0.1mm]north:X}] (C) at (-4.5,1) {};
\node[draw=none,label={[label distance=0.1mm]north:X}] (C) at (-2.5,1) {};
\node[draw=none,label={[label distance=0.1mm]north:X}] (C) at (-1.5,1) {};
\node[draw=none,label={[label distance=0.1mm]west:Y}] (C) at (-9,0) {};
\node[draw=none,label={[label distance=0.1mm]east:Y}] (C) at (-1,0) {};
\node[draw=none,label={[label distance=0.1mm]east:BULK}] (C) at (-5.85,-1) {};
\end{tikzpicture}
\end{center}
\vspace{3mm}

\subcaption{Family 2}
\end{subfigure}
\vspace{2mm}

\caption{Families of boundary generators with Paulis XYZ as denoted in equation \eqref{eq:Paulis}. A letter on a boundary node corresponds to one application of the Pauli matrix it labels.} 
\label{fig:FamiliesBoundary}
\end{figure}
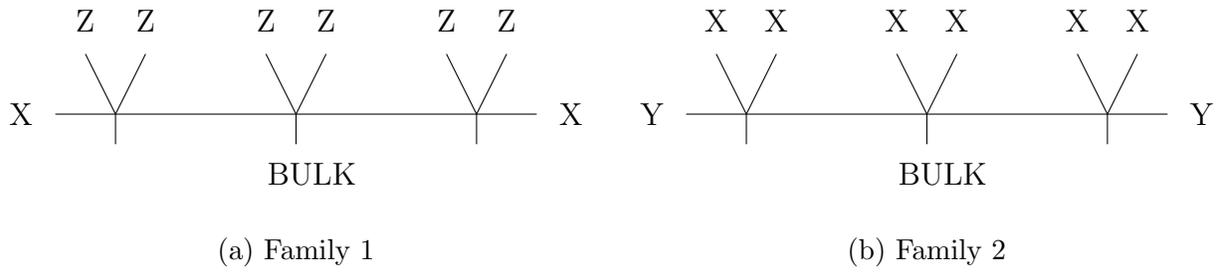

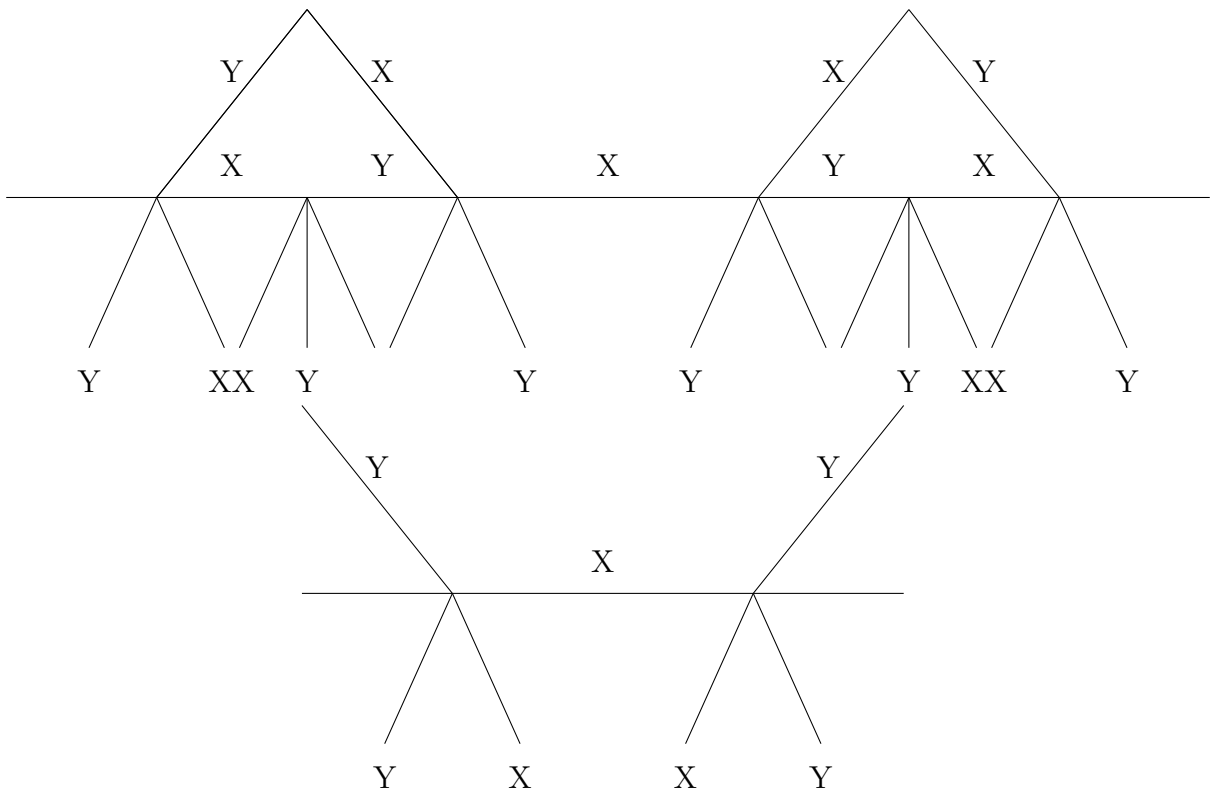
\begin{figure}[H]
\centering
\begin{tikzpicture}
\draw (-8,-2.5)--(-6,-2.5)--node [midway, label= X] {}(-4,-2.5)--node [midway, label= Y] {}(-2,-2.5)--node [midway, label= X] {}(2,-2.5)--node [midway, label= Y] {}(4,-2.5)--node [midway, label= X] {}(6,-2.5)--(8,-2.5);
\draw (-6,-2.5)--node [midway, label= Y] {}(-4,0);
\draw (-4,0)--node [midway, label= X] {}(-2,-2.5);
\draw (2,-2.5)--node [midway, label= X] {}(4,0);
\draw (4,0)--node [midway, label= Y] {}(6,-2.5);
\draw (-6,-2.5)--(-4,0);
\draw (-4,0)--(-2,-2.5);
\draw (-6,-2.5)--(-6-0.9,-4.5);
\draw (-6,-2.5)--(-6+0.9,-4.5);
\draw (-2+0.9,-4.5)--(-2,-2.5);
\draw (-2-0.9,-4.5)--(-2,-2.5);
\draw (-4,-2.5)--(-4+0.9,-4.5);
\draw (-4,-2.5)--(-4,-4.5);
\draw (-4,-2.5)--(-4-0.9,-4.5);
\draw (6,-2.5)--(6+0.9,-4.5);
\draw (6,-2.5)--(6-0.9,-4.5);
\draw (2+0.9,-4.5)--(2,-2.5);
\draw (2-0.9,-4.5)--(2,-2.5);
\draw (4,-2.5)--(4+0.9,-4.5);
\draw (4,-2.5)--(4,-4.5);
\draw (4,-2.5)--(4-0.9,-4.5);
\node[draw=none,label={[label distance=0.1mm]south: Y}] (C) at (-6-0.9,-4.5) {};
\node[draw=none,label={[label distance=0.1mm]south: XX}] (C) at (-5,-4.5) {};
\node[draw=none,label={[label distance=0.1mm]south: Y}] (C) at (-4,-4.5) {};
\node[draw=none,label={[label distance=0.1mm]south: Y}] (C) at (6+0.9,-4.5) {};
\node[draw=none,label={[label distance=0.1mm]south: XX}] (C) at (5,-4.5) {};
\node[draw=none,label={[label distance=0.1mm]south: Y}] (C) at (4,-4.5) {};
\node[draw=none,label={[label distance=0.1mm]south: Y}] (C) at (-2+0.9,-4.5) {};
\node[draw=none,label={[label distance=0.1mm]south: Y}] (C) at (2-0.9,-4.5) {};
\end{tikzpicture}
\begin{tikzpicture}
\draw (-4,0)--node [midway, label= Y] {}(-2,-2.5);
\draw (4,0)--node [midway, label= Y] {}(2,-2.5);
\draw (-4,-2.5)--node [midway, label= X] {}(4,-2.5);
\draw (-2+0.9,-4.5)--(-2,-2.5);
\draw (-2-0.9,-4.5)--(-2,-2.5);
\draw (2+0.9,-4.5)--(2,-2.5);
\draw (2-0.9,-4.5)--(2,-2.5);
\node[draw=none,label={[label distance=0.1mm]south: X}] (C) at (-2+0.9,-4.5) {};
\node[draw=none,label={[label distance=0.1mm]south: X}] (C) at (2-0.9,-4.5) {};
\node[draw=none,label={[label distance=0.1mm]south: Y}] (C) at (-2-0.9,-4.5) {};
\node[draw=none,label={[label distance=0.1mm]south: Y}] (C) at (2+0.9,-4.5) {};
\end{tikzpicture}
\vspace{5mm}
\caption{Pushing the patterns. $X$ and $Y$ are the Pauli matrices of \eqref{eq:Paulis}. YXXY is pushed to YXXY11YY11YXXY, while YY is pushed to YXXY. A letter on a boundary node corresponds to one application of the Pauli matrix it labels. A letter on an internal line corresponds to the insertion to two Paulis labelled by this letter, that square to one but make it explicit that the stabilizer structure of the code is respected.}
\label{patterns}
\end{figure}

\begin{figure}[H]
\centering
\vspace{5mm}
\begin{tikzpicture}[scale=1.35]
\draw (-6,-2.5)-- (-4,-2.5);
\draw (0,-2.5)--(0,-1.5); 
\draw (-4,-2.5)-- node [midway, label=Y] {} (0,-2.5);
\draw (0,-2.5)-- node [midway, label=Y] {} (4,-2.5);
\draw (0,-2.5)-- node [midway, label=X] {} (-2+0.9,-4.5);
\draw (0,-2.5)-- node [midway, label=X] {} (2-0.9,-4.5);
\draw (4,-2.5)-- (6,-2.5);
\draw (-2+0.9,-4.5)-- node [midway, label=\tiny X] {} (0,-4.5);
\draw (2-0.9,-4.5)-- node [midway, label=\tiny X] {} (0,-4.5);
\draw (-2+0.9,-4.5)-- (-2+0.2,-4.5);
\draw (-2+0.2,-4.5)-- node [midway, label=\tiny Y] {} (-2-0.5,-4.5);
\draw (-4+0.9,-4.5)-- (-2-0.5,-4.5);
\draw (-4+0.9,-4.5)-- node [midway, label=\tiny X] {} (-4+0.45,-4.5);
\draw (-4,-4.5)-- node [midway, label=\tiny X] {} (-4+0.45,-4.5);
\draw (-4,-4.5)-- node [midway, label=\tiny Y] {} (-4-0.45,-4.5);
\draw (-4-0.45,-4.5)-- node [midway, label=\tiny X] {} (-4-0.9,-4.5);
\draw (-4-0.9,-4.5) -- (-4-1.5,-4.5);
\draw (2-0.9,-4.5)-- (2-0.2,-4.5);
\draw (2-0.2,-4.5)-- node [midway, label=\tiny Y] {} (2+0.5,-4.5);
\draw (4-0.9,-4.5)-- (2+0.5,-4.5);
\draw (4-0.9,-4.5)-- node [midway, label=\tiny X] {} (4-0.45,-4.5);
\draw (4,-4.5)-- node [midway, label=\tiny X] {} (4-0.45,-4.5);
\draw (4,-4.5)-- node [midway, label=\tiny Y] {} (4+0.45,-4.5);
\draw (4+0.45,-4.5)-- node [midway, label=\tiny X] {} (4+0.9,-4.5);
\draw (4+0.9,-4.5) -- (4+1.5,-4.5);
\draw (-4,-2.5)-- node [midway, label=X] {} (-4+0.9,-4.5);
\draw (-4,-2.5)-- node [midway, label=X] {} (-4,-4.5);
\draw (-4,-2.5)-- node [midway, label=Y] {} (-4-0.9,-4.5);
\draw (4,-2.5)-- node [midway, label=Y] {} (4+0.9,-4.5);
\draw (4,-2.5)-- node [midway, label=X] {} (4,-4.5);
\draw (4,-2.5)-- node [midway, label=X] {} (4-0.9,-4.5);
\draw (-4-0.9-0.2,-6)--(-4-0.9,-4.5);
\draw (-4-0.9+0.1,-6)--(-4-0.9,-4.5);
\draw (-4-0.65,-6)--(-4-0.45,-4.5);
\draw (-4-0.45,-6)--(-4-0.45,-4.5);
\draw (-4-0.25,-6)--(-4-0.45,-4.5);
\draw (-4-0.1,-6)--(-4,-4.5);
\draw (-4+0.1,-6)--(-4,-4.5);
\draw (-4+0.65,-6)--(-4+0.45,-4.5);
\draw (-4+0.45,-6)--(-4+0.45,-4.5);
\draw (-4+0.25,-6)--(-4+0.45,-4.5);
\draw (-4+0.9-0.1,-6)--(-4+0.9,-4.5);
\draw (-4+0.9+0.2,-6)--(-4+0.9,-4.5);
\draw (-2+0.9-0.4,-6)--(-2+0.9,-4.5);
\draw (-2+0.9+0.4,-6)--(-2+0.9,-4.5);
\draw (8-6-0.9-0.4,-6)--(8-6-0.9,-4.5);
\draw (8-6-0.9+0.4,-6)--(8-6-0.9,-4.5);
\draw (8-4-0.9-0.2,-6)--(8-4-0.9,-4.5);
\draw (8-4-0.9+0.1,-6)--(8-4-0.9,-4.5);
\draw (8-4-0.65,-6)--(8-4-0.45,-4.5);
\draw (8-4-0.45,-6)--(8-4-0.45,-4.5);
\draw (8-4-0.25,-6)--(8-4-0.45,-4.5);
\draw (8-4-0.1,-6)--(8-4,-4.5);
\draw (8-4+0.1,-6)--(8-4,-4.5);
\draw (8-4+0.65,-6)--(8-4+0.45,-4.5);
\draw (8-4+0.45,-6)--(8-4+0.45,-4.5);
\draw (8-4+0.25,-6)--(8-4+0.45,-4.5);
\draw (8-4+0.9-0.1,-6)--(8-4+0.9,-4.5);
\draw (8-4+0.9+0.2,-6)--(8-4+0.9,-4.5);
\draw (0,-4.5)--(0,-6);
\draw (0,-4.5)--(0.5,-6);
\draw (0,-4.5)--(-0.5,-6);
\node[draw=none,label={[label distance=0.1mm]south:\tiny YXXY}] (C) at (-4-0.8,-6) {};
\node[draw=none,label={[label distance=0.1mm]south:\tiny YY}] (C) at (-4+0.2,-6) {};
\node[draw=none,label={[label distance=0.1mm]south:\tiny YY}] (C) at (-4+0.75,-6) {};
\node[draw=none,label={[label distance=0.1mm]south:\tiny YY}] (C) at (4-0.2,-6){};
\node[draw=none,label={[label distance=0.1mm]south:\tiny YY}] (C) at (4-0.75,-6){};
\node[draw=none,label={[label distance=0.1mm]south:\tiny YXXY}] (C) at (4+0.8,-6){};
\node[draw=none,label={[label distance=0.1mm]south:\tiny YY}] (C) at (0.6,-6){};
\node[draw=none,label={[label distance=0.1mm]south:\tiny YY}] (C) at (-0.6,-6){};
\end{tikzpicture}
\vspace{3mm}
\caption{Pushing stabilizer strings to start the fractal pattern. The initial stabilizer pattern YXXXXXXY is pushed to a boundary string made of YXXY's and YY's at the next level.}
\label{pushing}
\end{figure}

\begin{figure}[H]
\centering
\begin{tikzpicture}[scale=3.6]
\draw (0,0)--node [midway, label=\tiny Y] {}(1,0);
\draw (0,0)--(0.31,0.95);
\draw (0,0)-- node [midway, label=\tiny Y] {}(-0.81,0.59);
\draw (0,0)-- node [midway, label=\tiny X] {}(0.31,-0.95);
\draw (0,0)--node [midway, label=\tiny X] {}(-0.81,-0.59);
\draw (1,0)--(1.3,0.3);
\draw (1,0)--(1.3,-0.3);
\draw (1,0)--(0.31,0.95);
\draw (1,0)--node [midway, label=\tiny X] {}(0.655,-0.475)--node [midway, label=\tiny Y] {}(0.31,-0.95);
\draw (0.31,0.95)--(-0.81,0.59);
\draw (-0.81,-0.59)--node [midway, label=\tiny X] {}(-0.255,-0.77)--node [midway, label=\tiny X] {}(0.31,-0.95);
\draw (-0.81,-0.59)--node [midway, label=\tiny Y] {}(-0.81,0)--node [midway, label=\tiny X] {}(-0.81,0.59);
\draw[rotate around={72: (0,0)}] (1,0)--(1.3,0.3);
\draw[rotate around={2*72: (0,0)}] (1,0)--(1.3,0.3);
\draw[rotate around={3*72: (0,0)}] (1,0)--(1.3,0.3);
\draw[rotate around={4*72: (0,0)}] (1,0)--(1.3,0.3);
\draw[rotate around={72: (0,0)}] (1,0)--(1.3,-0.3);
\draw[rotate around={2*72: (0,0)}] (1,0)--(1.3,-0.3);
\draw[rotate around={3*72: (0,0)}] (1,0)--(1.3,-0.3);
\draw[rotate around={4*72: (0,0)}] (1,0)--(1.3,-0.3);
\draw (-0.81,0)--(-1,0);
\draw (-0.81,0)--(-1,0.2);
\draw (-0.81,0)--(-1,-0.2);
\draw[rotate around={72: (0,0)}] (-0.81,0)--(-1,0);
\draw[rotate around={2*72: (0,0)}] (-0.81,0)--(-1,0);
\draw[rotate around={3*72: (0,0)}] (-0.81,0)--(-1,0);
\draw[rotate around={4*72: (0,0)}] (-0.81,0)--(-1,0);
\draw[rotate around={72: (0,0)}] (-0.81,0)--(-1,0.2);
\draw[rotate around={2*72: (0,0)}] (-0.81,0)--(-1,0.2);
\draw[rotate around={3*72: (0,0)}] (-0.81,0)--(-1,0.2);
\draw[rotate around={4*72: (0,0)}] (-0.81,0)--(-1,0.2);
\draw[rotate around={72: (0,0)}] (-0.81,0)--(-1,-0.2);
\draw[rotate around={2*72: (0,0)}] (-0.81,0)--(-1,-0.2);
\draw[rotate around={3*72: (0,0)}] (-0.81,0)--(-1,-0.2);
\draw[rotate around={4*72: (0,0)}] (-0.81,0)--(-1,-0.2);
\node[draw=none,label={[label distance=0.1mm]east:\tiny X}] (C) at (1.3,-0.3) {};
\node[draw=none,label={[label distance=0.1mm]east:\tiny Y}] (C) at (1.3,0.3) {};
\node[draw=none,label={[label distance=0.1mm]east:\tiny Y}] (C) at (0,-1.4) {};
\node[draw=none,label={[label distance=0.1mm]east:\tiny X}] (C) at (0.8,-0.4) {};
\node[draw=none,label={[label distance=0.1mm]east:\tiny Y}] (C) at (0.65,-0.6) {};
\node[draw=none,label={[label distance=0.1mm]east:\tiny Y}] (C) at (-0.3,-1.1) {};
\node[draw=none,label={[label distance=0.1mm]east:\tiny Y}] (C) at (-0.75,-0.9) {};
\node[draw=none,label={[label distance=0.1mm]east:\tiny Y}] (C) at (-1.05,-1.1) {};
\node[draw=none,label={[label distance=0.1mm]east:\tiny Y}] (C) at (-1.25,0) {};
\node[draw=none,label={[label distance=0.1mm]east:\tiny X}] (C) at (-1.25,0.3) {};
\node[draw=none,label={[label distance=0.1mm]east:\tiny X}] (C) at (-1.55,0.55) {};
\node[draw=none,label={[label distance=0.1mm]east:\tiny Y}] (C) at (-1.15,1.1) {};
\end{tikzpicture}
\caption{In the special case of the central tensor, the level one stabilizer generator YXXY1 is pushed to YXXY11YY11YXXY11111111. We omit $1$ for convenience in the graph.}
\label{central}
\end{figure}

We now need a systematic way of pushing generators to the boundary. Interestingly, there exists a systematic way to push them into patterns of 
\begin{align}
S_1:=YXXY\quad\text{and}\quad S_2:=YY,
\end{align}
separated by a lot of qubit identity operators. More precisely, in $S_1$, the $X$'s should arise from different nodes, whereas in $S_2$, the $Y$'s should arise from different nodes.
Under those conditions, Figure \ref{patterns} explicitly shows that there is a systematic way to push $S_1$ to
\begin{align}
S_111S_211S_1,
\end{align}
and $S_2$ to $S_1$ with the $S_1$ and $S_2$ patterns in the right cluster configuration, therefore defining a sort of fractal of $S_1$'s and $S_2$'s. In addition, Figure \ref{pushing} below shows how to push stabilizer strings a first time to start the fractal, and Figure \ref{central} treats the specific case of the central tensor. Everything works strictly analogously for the second family of generators.

Then, one can recursively define a sequence $(S_n)$ of mappings of the string $YY$ to the boundary. The first few terms of the sequence $(S_n)$ are:
\begin{subequations} \label{eqs:YYmappings}
\begin{align}
&\ \, YY\\
&\ \, YXXY\\
&\ \, YXXY11YY11YXXY\\
&\begin{array}{l}
YXXY11YY11YXXY1111YXXY1111YXXY11YY11YXXY111111111111\\
11YXXY11YY11YXXY11111111111111YXXY11YY11YXXY1111YXXY\\
1111YXXY11YY11YXXY
\end{array}
\end{align}
\end{subequations}

It can be seen that the proportion of non-trivial Paulis rapidly goes to zero in these patterns, as the construction of the pattern at the next level essentially consists in making holes in the $YXXY$ patterns. This is reminiscent of the construction of fractals like the Cantor set, and interestingly echoes the concept of \textit{uberholography} introduced by Pastawski and Preskill in \cite{Pastawski:2016qrs}: a local operator in the bulk - in this case, the identity - is encoded in a fractal on the boundary. Note that our choice of bulk to boundary map as expressed in the figures of this section is the least disentangling, as it minimizes the size of the boundary region towards which the operators are pushed. It is then natural to see this map as a particularly nice \textit{entanglement wedge reconstruction} map: indeed, the smallest connected region of the boundary that supports its image will be dual to the bulk region on which our map acts nontrivially, and in analogy with the AdS/CFT correspondence, we shall call this region the entanglement wedge of the corresponding boundary region.

In order to illustrate the construction of stabilizer generators at each level, let us explicitly give the generators of the HaPPY code for the two first levels. 

At level one, we have 5 boundary qubits and 1 bulk qubit, therefore the number of generators is 
\begin{align}
g_1=\widetilde{N}_1-N_1=4,
\end{align}
and the HaPPY code reduces to the five-qubit code. The generators are four of the five cyclic permutations of the Pauli string $YXXY1$.

At level two, there are 14 generators, as specified in Table \ref{table:anbnhnhtn}. These include the $2b_2=10$ generators that can be obtained from sticking the two patterns of Figure \ref{fig:FamiliesBoundary} on each side of the pentagon of Figure \ref{fig:level1tensornetwork}. The four remaining generators correspond to the four pushes of the central stabilizer generators, which are explicitly illustrated on Figure \ref{central}, up to rotation. We list the 14 generators for the level two in Table \ref{tb:14generators}.\\

\begin{table}[H]
\arraycolsep=8pt\def\arraystretch{1.3}
\begin{center}
$
\begin{array}{|c|}
\hline
\text{Generators}\\
\hline
\hline
 11YXXXXXXY111111111111111\\
 1111111YXXXXXXY1111111111\\
 111111111111YXXXXXXY11111\\
 11111111111111111YXXXXXXY\\
 XXXXY11111111111111111YXX\\
 11XZZZZZZX111111111111111\\
 1111111XZZZZZZX1111111111\\
 111111111111XZZZZZZX11111\\
 11111111111111111XZZZZZZX\\
 ZZZZX11111111111111111XZZ\\
 YXXY11YY1YY11YXXY11111111\\
 11111YXXY11YY1YY11YXXY111\\
 XY11111111YXXY11YY1YY11YX\\
 Y1YY11YXXY11111111YXXY11Y\\
\hline
\end{array}
$
\caption{14 generators for the level two HaPPY code.}
\label{tb:14generators}
\end{center}
\end{table}

\section{The trapeze Hamiltonian}
\label{sec:trapezeHamiltonian}
The main goal of this section is to construct a bulk theory which mimics local quantum field theory in curved spacetime on the HaPPY code. 
In this section, we construct the dynamics of a local field theory in the bulk, which we explicitly map to the boundary. In order to do so, we introduce an Ising-like bulk Hamiltonian, and then show that it maps to an interesting fractal boundary pattern through the tensor network. With this Hamiltonian, we study the bulk physics.

\subsection{Trapeze operators and the bulk Hamiltonian}
The first thing to do is to define the bulk Hamiltonian. We will want the bulk field theory to be local, as is expected of quantum field theory in a curved background in the semiclassical limit of AdS/CFT. In spirit, we want to find an Ising-like interaction between the bulk qubits, which maps to the boundary in a simple enough way. It turns out that a particularly simple choice of such an interaction is the \textit{trapeze interaction}, which is a local interaction between 8 spins. We define a \textit{trapeze} in the bulk to be any interaction given by eight spins: two neighboring spins on the same bulk level, four neighboring descendants of those two spins, and the two 3-cluster parents in between. Figure \ref{fig:bulktrapeze} gives examples of such trapezes.\\

\begin{figure}[H]
\centering
\begin{tikzpicture}
\draw (-5,1)--(-4,0);
\draw (5,1)--(4,0);
\draw (-8,0)--(8,0);
\draw (-8,-2.5)--(8,-2.5);
\draw (-8,-4.5)--(8,-4.5);
\draw (-4,0)--(4,0);
\draw (-6,-2.5)--(-4,0);
\draw (-4,0)--(-2,-2.5);
\draw (2,-2.5)--(4,0);
\draw (4,0)--(6,-2.5);
\draw (-6,-2.5)--(-4,0);
\draw (-4,0)--(-2,-2.5);
\draw (-6,-2.5)--(-6-0.9,-4.5);
\draw (-6,-2.5)--(-6+0.5,-4.5);
\draw (-2+0.9,-4.5)--(-2,-2.5);
\draw (-2-0.5,-4.5)--(-2,-2.5);
\draw (-4,-2.5)--(-4+0.9,-4.5);
\draw (-4,-2.5)--(-4,-4.5);
\draw (-4,-2.5)--(-4-0.9,-4.5);
\draw (6,-2.5)--(6+0.9,-4.5);
\draw (6,-2.5)--(6-0.5,-4.5);
\draw (2+0.5,-4.5)--(2,-2.5);
\draw (2-0.9,-4.5)--(2,-2.5);
\draw (4,-2.5)--(4+0.9,-4.5);
\draw (4,-2.5)--(4,-4.5);
\draw (4,-2.5)--(4-0.9,-4.5);
\draw (-6-0.9-0.4,-6)--(-6-0.9,-4.5);
\draw (-6-0.9+0.2,-6)--(-6-0.9,-4.5);
\draw (-6-0.2-0.3,-6)--(-6-0.2,-4.5);
\draw (-6-0.2,-6)--(-6-0.2,-4.5);
\draw (-6-0.2+0.3,-6)--(-6-0.2,-4.5);
\draw (-6+0.5-0.2,-6)--(-6+0.5,-4.5);
\draw (-6+0.5+0.2,-6)--(-6+0.5,-4.5);
\draw (-4-0.9-0.2,-6)--(-4-0.9,-4.5);
\draw (-4-0.9+0.1,-6)--(-4-0.9,-4.5);
\draw (-4-0.65,-6)--(-4-0.45,-4.5);
\draw (-4-0.45,-6)--(-4-0.45,-4.5);
\draw (-4-0.25,-6)--(-4-0.45,-4.5);
\draw (-4-0.1,-6)--(-4,-4.5);
\draw (-4+0.1,-6)--(-4,-4.5);
\draw (-4+0.65,-6)--(-4+0.45,-4.5);
\draw (-4+0.45,-6)--(-4+0.45,-4.5);
\draw (-4+0.25,-6)--(-4+0.45,-4.5);
\draw (-4+0.9-0.1,-6)--(-4+0.9,-4.5);
\draw (-4+0.9+0.2,-6)--(-4+0.9,-4.5);
\draw (-2-0.5-0.2,-6)--(-2-0.5,-4.5);
\draw (-2-0.5+0.2,-6)--(-2-0.5,-4.5);
\draw (-2+0.2-0.3,-6)--(-2+0.2,-4.5);
\draw (-2+0.2,-6)--(-2+0.2,-4.5);
\draw (-2+0.2+0.3,-6)--(-2+0.2,-4.5);
\draw (-2+0.9-0.2,-6)--(-2+0.9,-4.5);
\draw (-2+0.9+0.4,-6)--(-2+0.9,-4.5);
\draw (8-6-0.9-0.4,-6)--(8-6-0.9,-4.5);
\draw (8-6-0.9+0.2,-6)--(8-6-0.9,-4.5);
\draw (8-6-0.2-0.3,-6)--(8-6-0.2,-4.5);
\draw (8-6-0.2,-6)--(8-6-0.2,-4.5);
\draw (8-6-0.2+0.3,-6)--(8-6-0.2,-4.5);
\draw (8-6+0.5-0.2,-6)--(8-6+0.5,-4.5);
\draw (8-6+0.5+0.2,-6)--(8-6+0.5,-4.5);
\draw (8-4-0.9-0.2,-6)--(8-4-0.9,-4.5);
\draw (8-4-0.9+0.1,-6)--(8-4-0.9,-4.5);
\draw (8-4-0.65,-6)--(8-4-0.45,-4.5);
\draw (8-4-0.45,-6)--(8-4-0.45,-4.5);
\draw (8-4-0.25,-6)--(8-4-0.45,-4.5);
\draw (8-4-0.1,-6)--(8-4,-4.5);
\draw (8-4+0.1,-6)--(8-4,-4.5);
\draw (8-4+0.65,-6)--(8-4+0.45,-4.5);
\draw (8-4+0.45,-6)--(8-4+0.45,-4.5);
\draw (8-4+0.25,-6)--(8-4+0.45,-4.5);
\draw (8-4+0.9-0.1,-6)--(8-4+0.9,-4.5);
\draw (8-4+0.9+0.2,-6)--(8-4+0.9,-4.5);
\draw (8-2-0.5-0.2,-6)--(8-2-0.5,-4.5);
\draw (8-2-0.5+0.2,-6)--(8-2-0.5,-4.5);
\draw (8-2+0.2-0.3,-6)--(8-2+0.2,-4.5);
\draw (8-2+0.2,-6)--(8-2+0.2,-4.5);
\draw (8-2+0.2+0.3,-6)--(8-2+0.2,-4.5);
\draw (8-2+0.9-0.2,-6)--(8-2+0.9,-4.5);
\draw (8-2+0.9+0.4,-6)--(8-2+0.9,-4.5);
\draw[draw,line width=1mm] (-2,-2.5)--(2,-2.5);
\draw[draw,line width=1mm] (-6,-2.5)--(-4,-2.5);
\draw[draw,line width=1mm] (4,-2.5)--(6,-2.5);
\draw[draw,line width=1mm] (-2,-2.5)--(-2+0.9,-4.5);
\draw[draw,line width=1mm] (-2,-2.5)--(-2-0.5,-4.5);
\draw[draw,line width=1mm] (2,-2.5)--(2-0.9,-4.5);
\draw[draw,line width=1mm] (2,-2.5)--(2+0.5,-4.5);
\draw[draw,line width=1mm] (-6,-2.5)--(-6-0.9,-4.5);
\draw[draw,line width=1mm] (-6,-2.5)--(-6+0.5,-4.5);
\draw[draw,line width=1mm] (-4,-2.5)--(-4-0.9,-4.5);
\draw[draw,line width=1mm] (-4,-2.5)--(-4,-4.5);
\draw[draw,line width=1mm] (6,-2.5)--(6+0.9,-4.5);
\draw[draw,line width=1mm] (6,-2.5)--(6-0.5,-4.5);
\draw[draw,line width=1mm] (4,-2.5)--(4+0.9,-4.5);
\draw[draw,line width=1mm] (4,-2.5)--(4,-4.5);
\draw[draw,line width=1mm] (-6-0.9,-4.5)--(-4,-4.5);
\draw[draw,line width=1mm] (-2-0.5,-4.5)--(2.5,-4.5);
\draw[draw,line width=1mm] (4,-4.5)--(6+0.9,-4.5);
\end{tikzpicture}
\caption{Trapeze shapes (in bold) and how they arise in the bulk.}
\label{fig:bulktrapeze}
\end{figure}
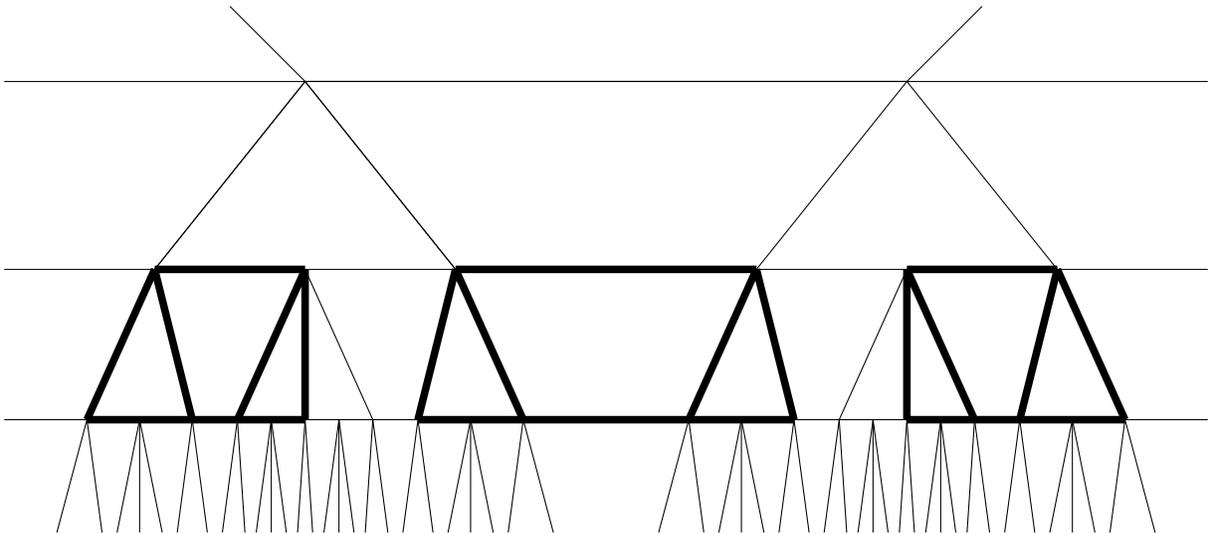

Once a trapeze is chosen, the corresponding \textit{trapeze interaction} is encoded by the tensor product of the $Z$ Pauli matrix applied to every trapeze site. It is the exact analog of an Ising interaction in the case of eight qubits, which makes it more degenerate but identical in spirit. The choice of such a trapeze interaction is motivated by the fact that it maps very simply to the closest layer of the bulk to the boundary, as shown in Figure \ref{fig:firstlevel}.\\

\begin{figure}[H]
\centering
\begin{tikzpicture}
\draw (-4,0)--(-4,1);
\draw (4,0)--(4,1);
\draw (-8,0)--(8,0);
\draw (-8,-2.5)--(-6,-2.5)--node [midway, label=\tiny Z] {}(-4,-2.5)--node [midway, label=\tiny Y] {}(-2,-2.5)--node [midway, label=\tiny Y] {}(2,-2.5)--node [midway, label=\tiny Y] {}(4,-2.5)--node [midway, label=\tiny Z] {}(6,-2.5)--(8,-2.5);
\draw (-4,0)--(4,0);
\draw (-6,-2.5)--node [midway, label=\tiny Y] {}(-4,0);
\draw (-4,0)--node [midway, label=\tiny Z] {}(-2,-2.5);
\draw (2,-2.5)--node [midway, label=\tiny Z] {}(4,0);
\draw (4,0)--node [midway, label=\tiny Y] {}(6,-2.5);
\draw (-6,-2.5)--(-4,0);
\draw (-4,0)--(-2,-2.5);
\draw (-6,-2.5)--(-6-0.9,-4.5);
\draw (-6,-2.5)--(-6+0.9,-4.5);
\draw (-2+0.9,-4.5)--(-2,-2.5);
\draw (-2-0.9,-4.5)--(-2,-2.5);
\draw (-4,-2.5)--(-4+0.9,-4.5);
\draw (-4,-2.5)--(-4,-4.5);
\draw (-4,-2.5)--(-4-0.9,-4.5);
\draw (6,-2.5)--(6+0.9,-4.5);
\draw (6,-2.5)--(6-0.9,-4.5);
\draw (2+0.9,-4.5)--(2,-2.5);
\draw (2-0.9,-4.5)--(2,-2.5);
\draw (4,-2.5)--(4+0.9,-4.5);
\draw (4,-2.5)--(4,-4.5);
\draw (4,-2.5)--(4-0.9,-4.5);
\node[draw,circle,thick,scale=1,black,label={[label distance=1mm]south:Z}] (C) at (-4,0) {};
\node[draw,circle,thick,scale=1,black,label={[label distance=1mm]south:Z}] (C) at (4,0) {};
\node[draw,circle,thick,scale=1,black,label={[label distance=1mm]north:Z}] (C) at (-6,-2.5) {};
\node[draw,circle,thick,scale=1,black,label={[label distance=1mm]north:Z}] (C) at (-2,-2.5) {};
\node[draw,circle,thick,scale=1,black,label={[label distance=1mm]north:Z}] (C) at (2,-2.5) {};
\node[draw,circle,thick,scale=1,black,label={[label distance=1mm]north:Z}] (C) at (6,-2.5) {};
\node[draw,circle,thick,scale=1,black,label={[label distance=1mm]north:Z}] (C) at (-4,-2.5) {};
\node[draw,circle,thick,scale=1,black,label={[label distance=1mm]north:Z}] (C) at (4,-2.5) {};
\node[draw=none,label={[label distance=0.5mm]south:YY}] (C) at (-5,-4.5) {};
\node[draw=none,label={[label distance=0.5mm]south:YY}] (C) at (5,-4.5) {};
\end{tikzpicture}
\caption{Mapping a trapeze interaction to the closest layer of the bulk to the boundary. The trapeze interaction is represented by the $Z$ Pauli matrices acting on the bulk nodes. It is mapped onto the first boundary level through operator pushing, and just gives two pairs of $Y$ Pauli matrices.}
\label{fig:firstlevel}
\end{figure}
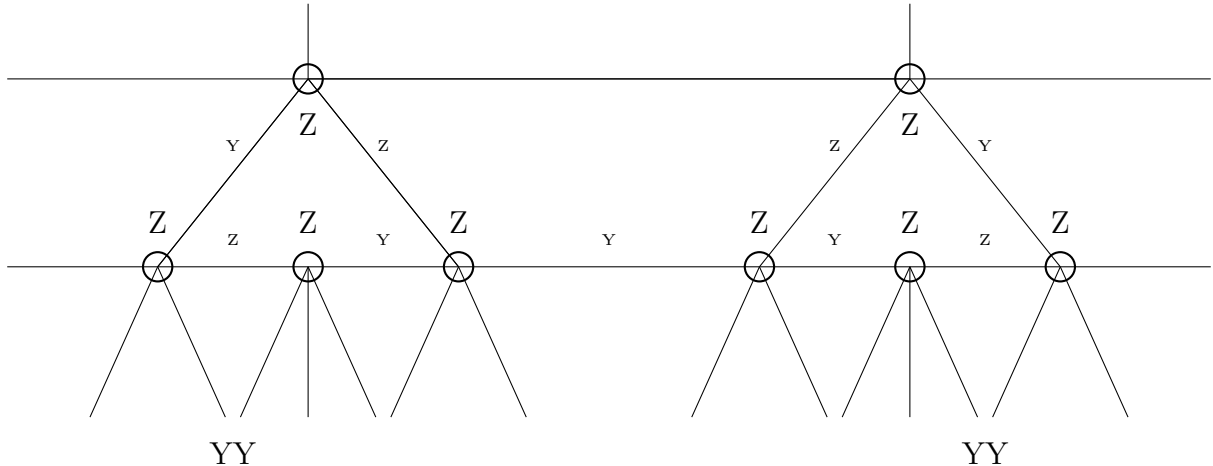

The problem of mapping a trapeze operator to the boundary therefore amounts to mapping the Pauli string $YY$, coming from two neighboring different bulk nodes, to the boundary. This will give rise to an interesting fractal pattern coming from the stabilizer structure, which will be our first step towards a qualitative understanding of the bulk-to-boundary map.

\subsection{Trapeze field theory: a mean-field approximation}
Before giving more details on the bulk-to-boundary map, we now show that the trapeze operator is the building block of a field theory in the bulk which has nontrivial behavior. In order to do so, we prove that it exhibits long range correlations, as a phase transition will happen in the bulk.

Let us start with the bulk Hamiltonian $$H_{\mathrm{bulk}}:=\sum_{\mathrm{trapezes}}\ \prod_{i\in\mathrm{trapeze}}Z_i.$$ The idea is to perform a mean field theory analysis of the bulk Hamiltonian and see whether it exhibits a phase transition. The difficulty here is that unlike in the Ising model, all bulk nodes don't have an identical role in the interactions. Here is our setup, which is a bit involved as we need need to distinguish five different kinds of bulk qubit variables:
\begin{itemize}
\item$s_1$ is the mean magnetization of 3-cluster parents in between two 3-cluster descendants,
\item$s_2$ is the mean magnetization of 3-cluster parents in between two 2-cluster descendants,
\item$s_3$ is the mean magnetization of 2-cluster descendants,
\item$s_4$ is the mean magnetization of 3-cluster extreme descendants,
\item$s_5$ is the mean magnetization of 3-cluster middle descendants.
\end{itemize}

Then the plan is simply to look at all the trapezes that each spin is in, and do regular mean field theory! There is one remaining issue: some spins will be in the upper level of the considered spin and have to be controlled. For this we have to introduce parameters $p$, $q$ and $r$ which are respectively the asymptotic proportions of 2-cluster descendants, 3-cluster extreme descendants and 3-cluster middle descendants. Then the unknown upper spins (and their consequences for the network) are taken into account by a weighted average over $p$, $q$ and $r$. Then in some additional specific cases for the 2-cluster, we will require the probability, that we denote by $\alpha$, of a 2-cluster descendant having a 3-cluster descendant next to it. Counting arguments similar to the ones in Section \ref{sec:AsymototesHn} yield:
\begin{align}
p=\frac{3-\sqrt{3}}{2},\quad q=\frac{\sqrt{3}-1}{3},\quad r=\frac{\sqrt{3}-1}{6},\quad \alpha=\frac{5+\sqrt{3}}{12}.
\end{align}
This analysis yields the following system:
\begin{subequations} \label{eqs:system}
\begin{align}
s_1&=\mathrm{tanh}\left( \,\beta\, \left(ps_2^2s_3^3s_4s_5+qs_1s_2s_3^2s_4^2s_5+rs_1s_2s_3^2s_4s_5^2+s_1s_2s_3^2s_4s_5^2+s_1s_2s_3^2s_4^2s_5\right) \right),\label{eq:s1}\\
s_2&=\mathrm{tanh}\left(\,\beta\,\left(2s_1s_2s_3^3s_4s_5+p(s_1s_2s_3^3s_4s_5+\alpha s_2s_3^5s_4+(1-\alpha)s_2s_3^6)\right.\right.\nonumber \\
&\quad\quad\quad\quad\quad\quad\quad\quad\quad\quad\quad\quad\quad\quad\quad\quad\quad\quad\left.\left.+q(s_1^2s_3^2s_4^2s_5+s_2s_3^5s_4)+2rs_1^2s_3^2s_4s_5^2\right)\right),\label{eq:s2}\\
s_3&=\mathrm{tanh}\left(\,\beta\,\left(s_1s_2^2s_3^2s_4s_5+p(s_1s_2^2s_3^2s_4s_5+s_2^2s_3^4\left(\frac{s_3+s_4}{2}\right)+\alpha s_2^2s_3^4s_4+(1-\alpha)s_2^2s_3^5)\right.\right.\nonumber \\
&\quad\quad\quad\quad\quad\quad\left.\left.+q(s_2^2s_3^4s_4+s_1^2s_2s_3s_4^2s_5+s_2^2s_3^4\left(\frac{s_3+s_4}{2}\right))+r(2s_1s_2^2s_3s_4s_5^2+s_2^2s_3^4s_4)\right)\right),\label{eq:s3}\\
s_4&=\mathrm{tanh}\left(\,\beta\,\left(s_1^2s_2s_3^2s_4s_5+s_2^2s_3^5+ps_1s_2^2s_3^3s_5+qs_1^2s_2s_3^2s_4s_5+rs_1^2s_2s_3^2s_5^2\right)\right),\label{eq:s4}\\
s_5&=\mathrm{tanh}\left(\,\beta\,\left(2s_1^2s_2s_3^2s_4s_5+2ps_1s_2^2s_3^3s_4+2qs_1^2s_2s_3^2s_4^2+2rs_1^2s_2s_3^2s_4s_5\right)\right) ,\label{eq:s5}
\end{align}
\end{subequations}
where $\beta$ is the inverse temperature.

We give a derivation of this system below. We derive the five magnetizations $s_1$, $s_2$, $s_3$, $s_4$, and $s_5$ individually in ascending order.

For $s_1$, we consider a given 3-cluster parent inside two 3-cluster descendants. It forms two trapezes with the lower boundary level, each made of half the corresponding 3-cluster and one of the two neighboring 2-clusters. These two trapezes contribute respectively 
\begin{align*}
& s_1s_2s_3^2s_4s_5^2\quad \times \quad \text{the spin of our bulk node},\\
& s_1s_2s_3^2s_4^2s_5\quad \times \quad \text{the spin of our bulk node}.
\end{align*}
 Then, the qubit is involved in a third interaction with the half 3-cluster of the upper level and the 2-cluster next to it. Depending if this 2-cluster comes from a 2-cluster (probability $p$), a 3-cluster extreme (probability $q$) or the center of a 3-cluster (probability $r$), this trapeze contributes respectively 
\begin{align*}
& ps_2^2s_3^3s_4s_5\quad \times \quad \text{the spin of our bulk node},\\
& qs_1s_2s_3^2s_4^2s_5\quad \times \quad \text{the spin of our bulk node},\\
& rs_1s_2s_3^2s_4s_5^2\quad \times \quad \text{the spin of our bulk node}.
\end{align*}
 We then sum all three contributions and multiply by minus the spin. (The spin is either $+1$ or $-1$.) Taking a thermal average at inverse temperature $\beta$ then yields $s_1$ as presented in equation \eqref{eq:s1}. 
 
 For $s_2$, the situation is a bit more complicated. Let us consider a given 3-cluster parent inside two 2-cluster descendants. It is involved in two interactions with the lower level, containing half the 3-cluster, and the corresponding neighboring 2-cluster. These two interactions lead to the contribution $$2s_1s_2s_3^3s_4s_5\quad \times \quad \text{the spin of our bulk node}.$$In order to describe the interactions involving the upper level, the upper 2-cluster parent can itself be
 \begin{enumerate}
\item	a descendant of a 2-cluster with probability $p$,
\item	the extreme descendant of a 3-cluster with probability $q$, 
\item	the middle descendant of a 3-cluster with probability $r$. 
\end{enumerate}
In the first case, our bulk qubit of interest will be involved in a trapeze interaction that involves a 2-cluster and a 3-cluster, and another one that involves two 2-clusters. The first trapeze always contributes $$s_1s_2s_3^3s_4s_5\quad \times \quad \text{the spin of our bulk node},$$ while the other one contributes in a different way depending on the nature of the parent of the second 2-cluster, which is a 2-cluster descendent with probability $\alpha$, and a 3-cluster extreme descendent with probability $1-\alpha$. The contribution of this trapeze is then $$\alpha s_2s_3^5s_4+(1-\alpha)s_2s_3^6\quad \times \quad \text{the spin of our bulk node}.$$ The overall contribution of this first case, up to multiplication by the spin of our bulk node, is then given by
\begin{align}
1:\quad p(s_1s_2s_3^3s_4s_5+\alpha s_2s_3^5s_4+(1-\alpha)s_2s_3^6).
\end{align}
In the second case, our bulk qubit is involved in two upper level interactions, formed respectively by one 2-cluster and one 3-cluster (contribution $qs_1^2s_3^2s_4^2s_5$) and two 2-clusters, the other 2-cluster parent being a 2-cluster descendent itself (contribution $qs_2s_3^5s_4$). The overall contribution of this second case, up to multiplication by the spin of our bulk node, is then: 
\begin{align}
2:\quad q(s_1^2s_3^2s_4^2s_5+s_2s_3^5s_4).
\end{align}
Finally, in the third case, our bulk qubit is involved in two identical interactions between a 2-cluster and half a 3-cluster, the contribution up to the bulk node spin is then:
\begin{align}
3:\quad 2rs_1^2s_3^2s_4s_5^2.
\end{align}
Summing all these contributions and then taking a thermal average at inverse temperature $\beta$ over the bulk qubit gives equation \eqref{eq:s2}.

For $s_3$, it is the most complicated case. Let us consider a given 2-cluster descendant. This descendent can contribute to the trapeze Hamiltonian in several ways. First, it will always interact with the trapeze it forms with its neighboring 3-cluster parent and the lower level, which contributes 
$$s_1s_2^2s_3^2s_4s_5\quad \times \quad \text{the spin of our bulk node}.$$ 
Then, for the 3 other trapeze interactions, we need to separate several cases: the parent 2-cluster node of our spin can itself be 
\begin{enumerate}
\item	a descendant of a 2-cluster with probability $p$,
\item	the extreme descendant of a 3-cluster with probability $q$, 
\item	the middle descendant of a 3-cluster with probability $r$. 
\end{enumerate}

Let us first consider the first case: a descendant of a 2-cluster with probability $p$. In this case, there are two situations we have to consider:
\begin{enumerate}[(1)]
\item	the spin of interest lies outside its 2-cluster grandparent,
\item	the spin of interest lies inside its 2-cluster grandparent. 
\end{enumerate}
If it lies outside (probability $\frac{1}{2}$), then it will be interacting in three different trapezes. 
On the upper level, they will respectively contain a 2-cluster and a 3-cluster, and two 2-clusters. The former trapeze always contributes 
$$1(1)\text{ upper level -- a 2-cluster and a 3-cluster}:\quad s_1s_2^2s_3^2s_4s_5,$$ 
while the latter's contribution depends on the nature of the parent of the second 2-cluster, which is a 2-cluster descendant with probability $\alpha$, and a 3-cluster extreme descendant with probability $1-\alpha$. The contribution for the latter is then
$$1(1)\text{ upper level -- two 2-clusters}:\quad \alpha s_2^2s_3^4s_4+(1-\alpha)s_2^2s_3^5.$$ 
Finally, on the lower level, our spin interacts with another trapeze made out of 2-clusters, which this time contributes 
$$1(1)\text{ lower level}:\quad s_2^2s_3^5.$$ 
If our spin lies inside the parent 2-cluster (probability $\frac{1}{2}$), then on the upper level, it participates in the two same trapeze interactions as before:
\begin{align*}
&1(2)\text{ upper level -- a 2-cluster and a 3-cluster}:\quad s_1s_2^2s_3^2s_4s_5 ,\\
&1(2)\text{ upper level -- a 2-cluster and a 3-cluster}:\quad \alpha s_2^2s_3^4s_4+(1-\alpha)s_2^2s_3^5 . 
\end{align*}
However, the lower trapeze interaction is still made of two 2-clusters, but its contribution is slightly different as 
$$1(2)\text{ lower level}:\quad s_2^2s_3^4s_4.$$
Hence, the overall contribution of the first case is 
\begin{align}
1:\quad p\left(s_1s_2^2s_3^2s_4s_5+s_2^2s_3^4\left(\frac{s_3+s_4}{2}\right)+\alpha s_2^2s_3^4s_5+(1-\alpha)s_2^2s_3^5\right).
\end{align}
Now consider the second case: the extreme descendant of a 3-cluster with probability $q$. Similarly from the previous case, we consider contributions from  
\begin{enumerate}[(1)]
\item	the qubit lies outside of its 3-cluster grandparent,
\item	the qubit lies inside of its 3-cluster grandparent.
\end{enumerate}
If it lies outside, then on the upper level, the trapezes respectively contribute 
$$2(1)\text{ upper level}:\quad s_2^2 s_3^4s_4\quad\text{and}\quad s_1^2s_2s_3s_4^2s_5,$$
while the lower one contributes 
$$2(1)\text{ lower level}:\quad s_2^2 s_3^5.$$
If it lies inside, then the trapezes on the upper level are unchanged, 
$$2(2)\text{ upper level}:\quad s_2^2 s_3^4s_4\quad\text{and}\quad s_1s_2s_3s_4^2s_5^2,$$
but on the lower level, the interaction becomes 
$$2(2)\text{ lower level}:\quad s_2^2s_3^4s_4.$$ 
Therefore, the overall contribution of the second case is 
\begin{align}
2:\quad q\left(s_2^2s_3^4s_4+s_1s_2s_3s_4^2s_5+s_2^2s_3^4\left(\frac{s_3+s_4}{2}\right)\right).
\end{align}
Finally, we consider the third case: the middle descendant of a 3-cluster with probability $r$. In this case, there is only one possibility: the two upper level trapezes contribute 
$$3\text{ upper level}:\quad s_1s_2^2s_3s_4s_5^2,$$
whilst the lower one contributes 
$$3\text{ lower level}:\quad s_2^2s_3^4s_4.$$
The overall contribution of the third case is
\begin{align}
3:\quad r(2s_1s_2^2s_3s_4s_5^2+s_2^2s_3^4s_4).
\end{align}
We sum all three contributions and multiply by minus the spin. Taking a thermal average at inverse temperature $\beta$ then yields $s_3$ as presented in equation \eqref{eq:s3}. 

For $s_4$, the situation is simpler. Consider a 3-cluster extreme descendent in the bulk. It will be in two trapeze interactions with the lower level, one involving two 2-clusters and one involving one 2-cluster and half a 3-cluster. They respectively contribute
\begin{align*}
& s_2^2s_3^5\quad \times \quad \text{the spin of our bulk node},\\
& s_1^2s_2s_3^2s_4s_5\quad \times \quad \text{the spin of our bulk node}.
\end{align*}
On the upper level, our spin will be involved in one trapeze interaction between a 2-cluster and half the 3-cluster parent. The form of this interaction will depend on the 2-cluster parent, which can be a 2-cluster descendant (probability $p$), a 3-cluster extreme descendant (probability $q$), or a 3-cluster middle descendant (probability $r$). Their contributions are respectively:
\begin{align*}
& ps_1s_2^2s_3^3s_5\quad \times \quad \text{the spin of our bulk node},\\
& qs_1^2s_2s_3^2s_4s_5\quad \times \quad \text{the spin of our bulk node},\\
& rs_1^2s_2s_3^2s_5^2\quad \times \quad \text{the spin of our bulk node}.
\end{align*}
Summing all these contributions and taking a thermal average at inverse temperature $\beta$ yields equation \eqref{eq:s4}.

Finally, $s_5$ is the most symmetric case. Consider a 3-cluster middle descendent in the bulk. It will participate in two identical interactions with the lower level, which each involve the underlying 2-cluster and one of its neighboring 3-clusters. They will contribute 
$$2s_1^2s_2s_3^2s_4s_5\quad \times \quad \text{the spin of our bulk node}.$$ 
Then the interactions with the upper level will also be identical, and each involve a neighboring 2-cluster. Their form will as usual depend on the nature of the neighboring 2-cluster parent, and depending, will yield contributions of the forms 
\begin{align*}
& 2ps_1s_2^2s_3^3s_4\quad \times \quad \text{the spin of our bulk node},\\
& 2qs_1^2s_2s_3^2s_4^2\quad \times \quad \text{the spin of our bulk node},\\
& 2rs_1^2s_2s_3^2s_4s_5\quad \times \quad \text{the spin of our bulk node}.
\end{align*}
We sum all these contributions and then take a thermal average at inverse temperature $\beta$. This yields the last magnetization $s_5$ of the system as presented in equation \eqref{eq:s5}.

Thus we have derived the system of five mean magnetizations $s_1$, $s_2$, $s_3$, $s_4$, and $s_5$, given by \cref{eqs:system}.

\subsection{Bulk phase transition}
A numerical analysis shows that the system of five magnetizations $s_1$, $s_2$, $s_3$, $s_4$, and $s_5$, given by \cref{eqs:system}, has a nontrivial solution only at low enough temperature, and therefore exhibits a phase transition. The graph being two-dimensional, we can't expect mean-field theory results to be exact, but we expect that its connectivity is high enough that it still exhibits the right qualitative behavior. As an example, we give a possible evolution of the average total magnetizations $s_1 s_2 s_3 s_4 s_5$ with the inverse temperature $\beta$ on Figure \ref{fig:trans}. 

\begin{figure}[H]
    \centering
    \includegraphics[scale=0.56]{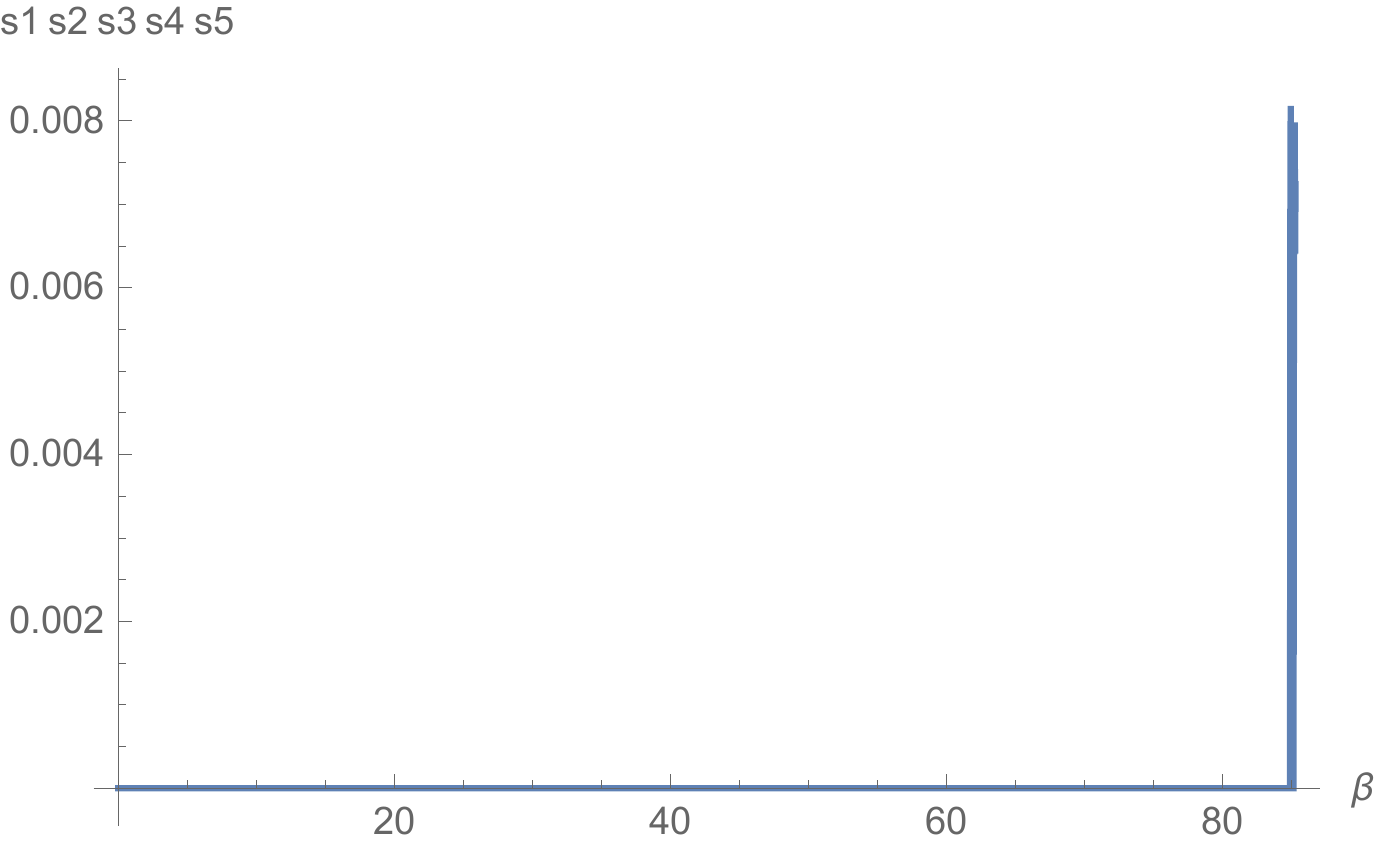}
    \includegraphics[scale=0.56]{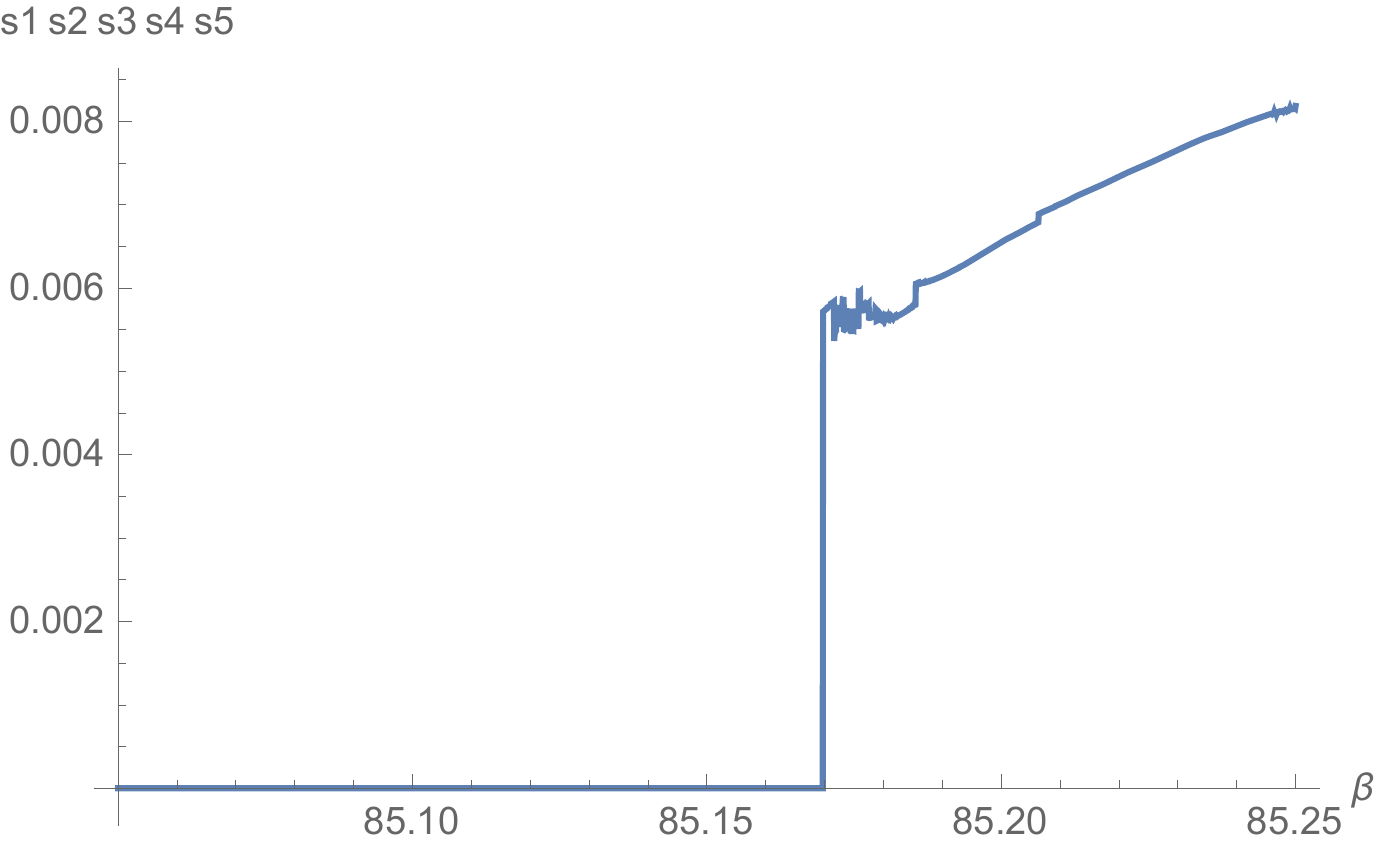}
    \caption{A possible outcome for the average total magnetization $s_1s_2s_3s_4s_5$ of trapeze mean field theory in the bulk in terms of the inverse temperature $\beta$. At low $\beta$, the curve is identically zero and then jumps to a nonzero value around $\beta=85.17$, which is characteristic of a phase transition. The plot on the right is zoomed in on the same average total magnetization in terms of the inverse temperature around the critical point.}
    \label{fig:trans}
\end{figure}

\section{The bulk-to-boundary map: the trapeze Hamiltonian}
\label{sec:bulktoboundarytrapeze}
In Section \ref{sec:trapezeHamiltonian}, we constructed the trapeze Hamiltonian that gives rise to a physical theory in the bulk and showed that it exhibits a phase transition. In this section, we study how this theory maps to the boundary. Mapping a trapeze operator to the boundary is done by mapping separately the two pairs of $Y$, and then taking the tensor product of the mapping. From Section \ref{sec:BoundaryStabilizerGen}, we know how to \textit{uberholographically} map $YY$ to the boundary in the framework of our consideration on the stabilizer, as demonstrated in \cref{eqs:YYmappings}, which makes this task straightforward.

Now, it is not enough to just map one trapeze operator to the boundary. Instead, we want to map all the local field theory, i.e. all of the trapeze operators within an entanglement wedge, and express their mapping in terms of strings of qubits. In order to do so, we shall first reason within an entanglement wedge defined by a bulk trapeze. Define a \textit{dashed line} in the bulk as a semi-infinite line from the bulk to the boundary, which always separates two neighboring 2-clusters. A dashed line defines an entanglement wedge pointing at a trapeze made of two 2-clusters, as shown in Figure \ref{fig:dottedline}.

In order to understand how all the trapezes in the entanglement wedge will map to the boundary, it is first useful to understand the structure of the network around the dashed line in terms of 2-clusters and 3-clusters. This can be done through a \textit{shifting and unfolding} procedure, which we now introduce.

\begin{figure}[H]
\centering
\begin{tikzpicture}
\draw (-5,1)--(-4,0);
\draw (5,1)--(4,0);
\draw (-8,0)--(8,0);
\draw (-8,-2.5)--(8,-2.5);
\draw (-8,-4.5)--(8,-4.5);
\draw (-4,0)--(4,0);
\draw (-6,-2.5)--(-4,0);
\draw (-4,0)--(-2,-2.5);
\draw (2,-2.5)--(4,0);
\draw (4,0)--(6,-2.5);
\draw (-6,-2.5)--(-4,0);
\draw (-4,0)--(-2,-2.5);
\draw (-6,-2.5)--(-6-0.9,-4.5);
\draw (-6,-2.5)--(-6+0.5,-4.5);
\draw (-2+0.9,-4.5)--(-2,-2.5);
\draw (-2-0.5,-4.5)--(-2,-2.5);
\draw (-4,-2.5)--(-4+0.9,-4.5);
\draw (-4,-2.5)--(-4,-4.5);
\draw (-4,-2.5)--(-4-0.9,-4.5);
\draw (6,-2.5)--(6+0.9,-4.5);
\draw (6,-2.5)--(6-0.5,-4.5);
\draw (2+0.5,-4.5)--(2,-2.5);
\draw (2-0.9,-4.5)--(2,-2.5);
\draw (4,-2.5)--(4+0.9,-4.5);
\draw (4,-2.5)--(4,-4.5);
\draw (4,-2.5)--(4-0.9,-4.5);
\draw (-6-0.9-0.4,-6)--(-6-0.9,-4.5);
\draw (-6-0.9+0.2,-6)--(-6-0.9,-4.5);
\draw (-6-0.2-0.3,-6)--(-6-0.2,-4.5);
\draw (-6-0.2,-6)--(-6-0.2,-4.5);
\draw (-6-0.2+0.3,-6)--(-6-0.2,-4.5);
\draw (-6+0.5-0.2,-6)--(-6+0.5,-4.5);
\draw (-6+0.5+0.2,-6)--(-6+0.5,-4.5);
\draw (-4-0.9-0.2,-6)--(-4-0.9,-4.5);
\draw (-4-0.9+0.1,-6)--(-4-0.9,-4.5);
\draw (-4-0.65,-6)--(-4-0.45,-4.5);
\draw (-4-0.45,-6)--(-4-0.45,-4.5);
\draw (-4-0.25,-6)--(-4-0.45,-4.5);
\draw (-4-0.1,-6)--(-4,-4.5);
\draw (-4+0.1,-6)--(-4,-4.5);
\draw (-4+0.65,-6)--(-4+0.45,-4.5);
\draw (-4+0.45,-6)--(-4+0.45,-4.5);
\draw (-4+0.25,-6)--(-4+0.45,-4.5);
\draw (-4+0.9-0.1,-6)--(-4+0.9,-4.5);
\draw (-4+0.9+0.2,-6)--(-4+0.9,-4.5);
\draw (-2-0.5-0.2,-6)--(-2-0.5,-4.5);
\draw (-2-0.5+0.2,-6)--(-2-0.5,-4.5);
\draw (-2+0.2-0.3,-6)--(-2+0.2,-4.5);
\draw (-2+0.2,-6)--(-2+0.2,-4.5);
\draw (-2+0.2+0.3,-6)--(-2+0.2,-4.5);
\draw (-2+0.9-0.2,-6)--(-2+0.9,-4.5);
\draw (-2+0.9+0.4,-6)--(-2+0.9,-4.5);
\draw (8-6-0.9-0.4,-6)--(8-6-0.9,-4.5);
\draw (8-6-0.9+0.2,-6)--(8-6-0.9,-4.5);
\draw (8-6-0.2-0.3,-6)--(8-6-0.2,-4.5);
\draw (8-6-0.2,-6)--(8-6-0.2,-4.5);
\draw (8-6-0.2+0.3,-6)--(8-6-0.2,-4.5);
\draw (8-6+0.5-0.2,-6)--(8-6+0.5,-4.5);
\draw (8-6+0.5+0.2,-6)--(8-6+0.5,-4.5);
\draw (8-4-0.9-0.2,-6)--(8-4-0.9,-4.5);
\draw (8-4-0.9+0.1,-6)--(8-4-0.9,-4.5);
\draw (8-4-0.65,-6)--(8-4-0.45,-4.5);
\draw (8-4-0.45,-6)--(8-4-0.45,-4.5);
\draw (8-4-0.25,-6)--(8-4-0.45,-4.5);
\draw (8-4-0.1,-6)--(8-4,-4.5);
\draw (8-4+0.1,-6)--(8-4,-4.5);
\draw (8-4+0.65,-6)--(8-4+0.45,-4.5);
\draw (8-4+0.45,-6)--(8-4+0.45,-4.5);
\draw (8-4+0.25,-6)--(8-4+0.45,-4.5);
\draw (8-4+0.9-0.1,-6)--(8-4+0.9,-4.5);
\draw (8-4+0.9+0.2,-6)--(8-4+0.9,-4.5);
\draw (8-2-0.5-0.2,-6)--(8-2-0.5,-4.5);
\draw (8-2-0.5+0.2,-6)--(8-2-0.5,-4.5);
\draw (8-2+0.2-0.3,-6)--(8-2+0.2,-4.5);
\draw (8-2+0.2,-6)--(8-2+0.2,-4.5);
\draw (8-2+0.2+0.3,-6)--(8-2+0.2,-4.5);
\draw (8-2+0.9-0.2,-6)--(8-2+0.9,-4.5);
\draw (8-2+0.9+0.4,-6)--(8-2+0.9,-4.5);
\draw[draw,line width=1mm] (0,1)--(0,0);
\draw[draw,line width=1mm] (0,-1)--(0,-2);
\draw[draw,line width=1mm] (0,-3)--(0,-4);
\draw[draw,line width=1mm] (0,-5)--(0,-6);
\end{tikzpicture}
\vspace{3mm}

\caption{A dashed line defined that separates two neighboring 2-clusters. We observe a symmetry between this dashed line.}
\label{fig:dottedline}
\end{figure}
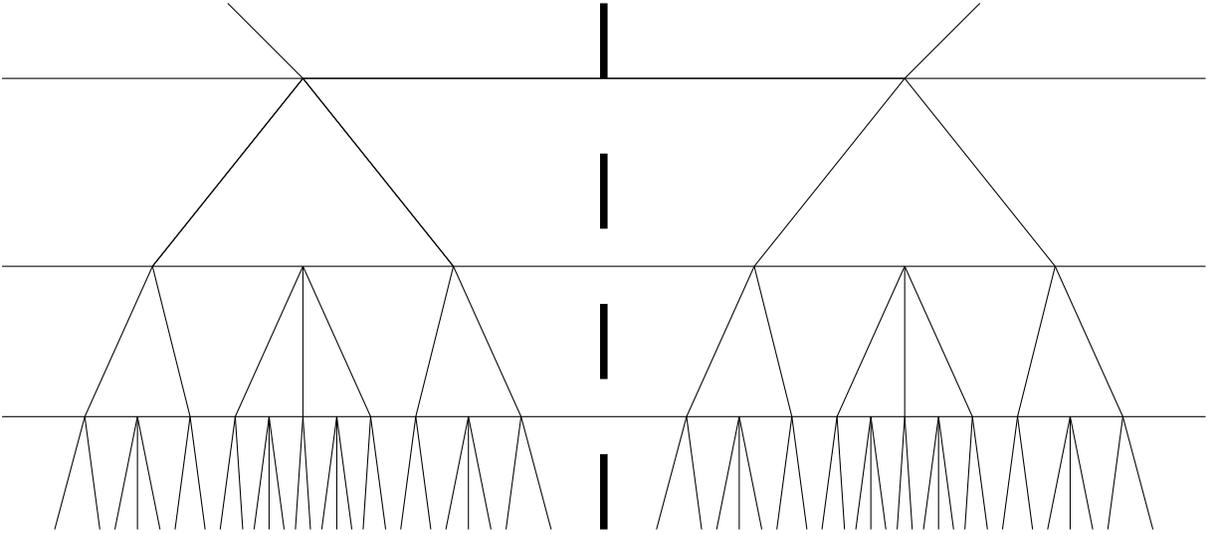

\subsection{The shifting and unfolding procedure}

Here we explain how to obtain the network around a dashed line at level $n+1$ from level $n$ through the shifting and unfolding procedure, and apply it to the pushing of the Hamiltonian.

The steps of the process are illustrated in Figure \ref{fig:ShiftUnfold} and work as follows: starting one level after the dashed line begins, take the string of 2-clusters and 3-clusters to the right of the dashed line (shifting), then symmetrize it with respect to the last 2-cluster (first unfolding). Then symmetrize everything with respect to the dashed line (second unfolding). Repeat the operation infinitely, to keep increasing the level.

This gives us an explicit way of constructing the network around dashed lines. Let us now study the new trapezes this shifting and unfolding procedure creates. First, all trapezes of level $n$ will be replicated through the shifting and unfolding procedure. After the shift, they will each give rise to a new trapeze, which will double the number of trapezes. Then each unfolding will double that new number again. Therefore, the first step to construct all the trapeze shapes at level $n+1$ is to take the trapeze shapes at level $n$, and copy them through by shifting them and unfolding them. Almost all trapeze shapes at level $n+1$ can be obtained in this way, except the ones centered on the dashed line, which have to be added.

\begin{figure}[H]
\centering
\begin{subfigure}{1\textwidth}
\centering
\begin{tikzpicture}[scale=0.45]
\draw (1,1)--(0,0);
\draw (1,1)--(2,0);
\draw (4,1)--(3,0);
\draw (4,1)--(4,0);
\draw (4,1)--(5,0);
\draw (7,1)--(6,0);
\draw (7,1)--(8,0);
\draw (9,1)--(9,0.8);
\draw (9,0.6)--(9,0.4);
\draw (9,0.2)--(9,0);
\draw (18-1,1)--(18-0,0);
\draw (18-1,1)--(18-2,0);
\draw (18-4,1)--(18-3,0);
\draw (18-4,1)--(18-4,0);
\draw (18-4,1)--(18-5,0);
\draw (18-7,1)--(18-6,0);
\draw (18-7,1)--(18-8,0);
\end{tikzpicture}
\vspace{3mm}
\subcaption{Start with a given level.}
\vspace{5mm}
\begin{tikzpicture}[scale=0.45]
\draw (10+1,1)--(10+0,0);
\draw (10+1,1)--(10+2,0);
\draw (10+4,1)--(10+3,0);
\draw (10+4,1)--(10+4,0);
\draw (10+4,1)--(10+5,0);
\draw (10+7,1)--(10+6,0);
\draw (10+7,1)--(10+8,0);
\draw (9,1)--(9,0.8);
\draw (9,0.6)--(9,0.4);
\draw (9,0.2)--(+9,0);
\draw (10+18-1,1)--(10+18-0,0);
\draw (10+18-1,1)--(10+18-2,0);
\draw (10+18-4,1)--(10+18-3,0);
\draw (10+18-4,1)--(10+18-4,0);
\draw (10+18-4,1)--(10+18-5,0);
\draw (10+18-7,1)--(10+18-6,0);
\draw (10+18-7,1)--(10+18-8,0);
\end{tikzpicture}
\vspace{3mm}
\subcaption{Shift everything to the right of the dashed line.}
\vspace{5mm}
\begin{tikzpicture}[scale=0.45]
\draw (10+1,1)--(10+0,0);
\draw (10+1,1)--(10+2,0);
\draw (10+4,1)--(10+3,0);
\draw (10+4,1)--(10+4,0);
\draw (10+4,1)--(10+5,0);
\draw (10+7,1)--(10+6,0);
\draw (10+7,1)--(10+8,0);
\draw (9,1)--(9,0.8);
\draw (9,0.6)--(9,0.4);
\draw (9,0.2)--(+9,0);
\draw (10+18-1,1)--(10+18-0,0);
\draw (10+18-1,1)--(10+18-2,0);
\draw (10+18-4,1)--(10+18-3,0);
\draw (10+18-4,1)--(10+18-4,0);
\draw (10+18-4,1)--(10+18-5,0);
\draw (10+18-7,1)--(10+18-6,0);
\draw (10+18-7,1)--(10+18-8,0);
\draw (29,0)--(30,1);
\draw (30,0)--(30,1);
\draw (31,0)--(30,1);
\draw (32,0)--(33,1);
\draw (34,0)--(33,1);
\draw (36,0)--(37,1);
\draw (38,0)--(37,1);
\draw (39,0)--(40,1);
\draw (40,0)--(40,1);
\draw (41,0)--(40,1);
\draw (42,0)--(43,1);
\draw (44,0)--(43,1);
\end{tikzpicture}
\vspace{3mm}
\subcaption{Symmetrize everything with respect to the last 2-cluster.}
\vspace{5mm}
\begin{tikzpicture}[scale=0.23]
\draw (10+1,1)--(10+0,0);
\draw (10+1,1)--(10+2,0);
\draw (10+4,1)--(10+3,0);
\draw (10+4,1)--(10+4,0);
\draw (10+4,1)--(10+5,0);
\draw (10+7,1)--(10+6,0);
\draw (10+7,1)--(10+8,0);
\draw (9,1)--(9,0.8);
\draw (9,0.6)--(9,0.4);
\draw (9,0.2)--(+9,0);
\draw (10+18-1,1)--(10+18-0,0);
\draw (10+18-1,1)--(10+18-2,0);
\draw (10+18-4,1)--(10+18-3,0);
\draw (10+18-4,1)--(10+18-4,0);
\draw (10+18-4,1)--(10+18-5,0);
\draw (10+18-7,1)--(10+18-6,0);
\draw (10+18-7,1)--(10+18-8,0);
\draw (29,0)--(30,1);
\draw (30,0)--(30,1);
\draw (31,0)--(30,1);
\draw (32,0)--(33,1);
\draw (34,0)--(33,1);
\draw (36,0)--(37,1);
\draw (38,0)--(37,1);
\draw (39,0)--(40,1);
\draw (40,0)--(40,1);
\draw (41,0)--(40,1);
\draw (42,0)--(43,1);
\draw (44,0)--(43,1);
\draw (10+1-36,1)--(10+0-36,0);
\draw (10+1-36,1)--(10+2-36,0);
\draw (10+4-36,1)--(10+3-36,0);
\draw (10+4-36,1)--(10+4-36,0);
\draw (10+4-36,1)--(10+5-36,0);
\draw (10+7-36,1)--(10+6-36,0);
\draw (10+7-36,1)--(10+8-36,0);
\draw (10+18-1-36,1)--(10+18-0-36,0);
\draw (10+18-1-36,1)--(10+18-2-36,0);
\draw (10+18-4-36,1)--(10+18-3-36,0);
\draw (10+18-4-36,1)--(10+18-4-36,0);
\draw (10+18-4-36,1)--(10+18-5-36,0);
\draw (10+18-7-36,1)--(10+18-6-36,0);
\draw (10+18-7-36,1)--(10+18-8-36,0);
\draw (29-36,0)--(30-36,1);
\draw (30-36,0)--(30-36,1);
\draw (31-36,0)--(30-36,1);
\draw (32-36,0)--(33-36,1);
\draw (34-36,0)--(33-36,1);
\draw (36-36,0)--(37-36,1);
\draw (38-36,0)--(37-36,1);
\draw (39-36,0)--(40-36,1);
\draw (40-36,0)--(40-36,1);
\draw (41-36,0)--(40-36,1);
\draw (42-36,0)--(43-36,1);
\draw (44-36,0)--(43-36,1);
\end{tikzpicture}
\vspace{3mm}
\subcaption{Symmetrize everything with respect to the dashed line.}
\vspace{3mm}
\end{subfigure}
\caption{Shifting and unfolding levels of the HaPPY code.}
\label{fig:ShiftUnfold}
\end{figure}
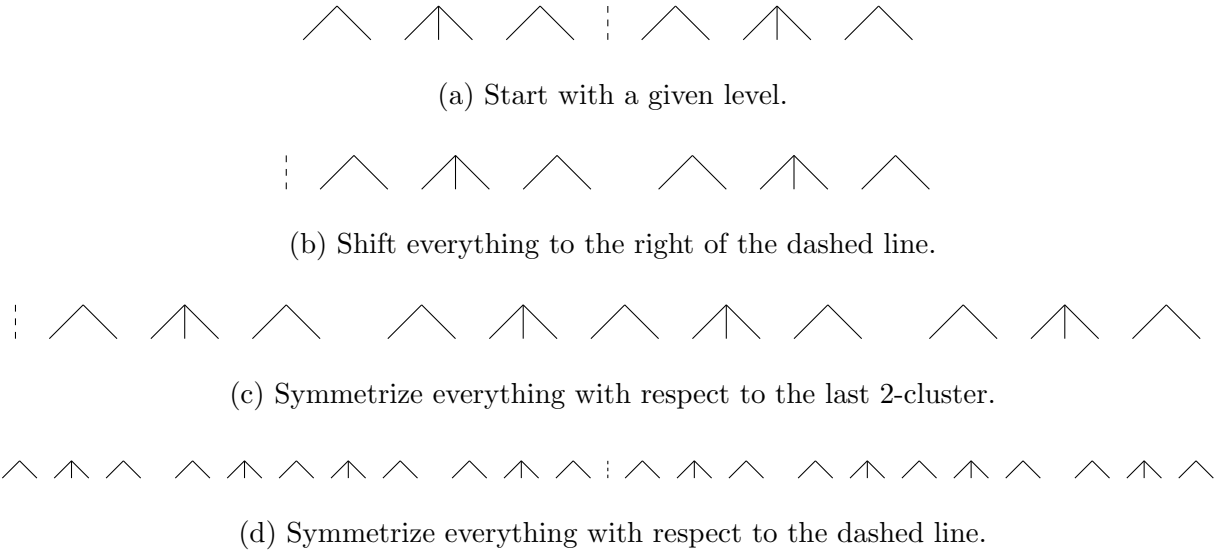 

These observations give rise to an explicit way of recursively building an infinite string of Paulis along each dashed line. The $n$-th step of the procedure will correspond to the mapping of all bulk trapeze operators contained in the entanglement wedge defined by the dashed line, and whose bulk level is less than $p+n$ if the dashed line starts at level $p$. This procedure is summarized in Figure \ref{fig:StepsShiftUnfold}. We first initialize with the pattern 
\begin{align*}
\begin{tikzpicture}
\node[draw=none,label={[label distance=0mm]east: 1YY11111111YY1,}] (C) at (0,0) {};
\draw (1.95,-0.2)--(1.95,-0.4);
\draw (1.95,0.6)--(1.95,0.4);
\draw (1.95,0.2)--(1.95,0);
\end{tikzpicture}
\end{align*}
where we have put our reference area as the dashed line in the middle. Now suppose that the $n$-th iteration of the process has been made, and we have a string of blocks of Paulis acting on the boundary qubits at level $n$. Shift every existing block of Paulis to the right of the dashed line. Symmetrize with respect to the axis between the rightmost 1 and the rightmost $Y$, then symmetrize back with respect to the dashed line. Then add an additional pattern 
$$S_k11(...)11S_k$$
centered with respect to the dashed line, with an appropriate number of $1$'s in between, for each $k\leq n+1$. The number of 1's separating the two $S_k$'s is defined in the following way: if $k\leq n$, take the number of 1's to be the same as it was defined before, and if $k=n+1$, define it so that the first $S_{n+1}$ is in the middle of the leftmost string of the form $S_n(...)S_n$, and the second $S_{n+1}$ is in the middle of the rightmost string of the form $S_n(...)S_n$.

This procedure can be seen clearly in Figure \ref{fig:StepsShiftUnfold}.

\begin{figure}[H]
\vspace{3mm}
\centering
\begin{subfigure}{1\textwidth}
\begin{center}
\begin{tikzpicture}
\node[draw=none,label={[label distance=0mm]east: \texttt{1111111111111111111\textcolor{blue}{1YY11111111YY1}1111111111111111111}}] (C) at (0,0) {};
\draw (5.95,-0.2)--(5.95,-0.4);
\draw (5.95,0.6)--(5.95,0.4);
\draw (5.95,0.2)--(5.95,0);
\end{tikzpicture}
\end{center}
\vspace{0.5mm}
\subcaption{Start with the initial string.}
\vspace{1.6mm}
\hspace{8cm}
\begin{center}
\begin{tikzpicture}
\node[draw=none,label={[label distance=0mm]east: \texttt{11111111111111111111111111\textcolor{blue}{1YY11111111YY1}111111111111}}] (C) at (-5.75,0) {};
\draw (0.2,-0.2)--(0.2,-0.4);
\draw (0.2,0.6)--(0.2,0.4);
\draw (0.2,0.2)--(0.2,0);
\end{tikzpicture}
\end{center}
\vspace{0.5mm}
\subcaption{Shift everything to the right of the dashed line.}
\vspace{1.6mm}
\hspace{8cm}
\begin{center}
\begin{tikzpicture}
\node[draw=none,label={[label distance=0mm]east: \texttt{11111111111111111111111111\textcolor{blue}{1YY11111111YY1}111111111111}}] (C) at (-5.75,0) {};
\node[draw=none,label={[label distance=0mm]east:
\texttt{11111111111111111111111111111111111111\textcolor{blue}{1YY11111111YY1}}}] (C) at (-5.75,-0.5) {};
\draw (0.2,-0.6)--(0.2,-0.8);
\draw (0.2,-0.2)--(0.2,-0.4);
\draw (0.2,0.6)--(0.2,0.4);
\draw (0.2,0.2)--(0.2,0);
\end{tikzpicture}
\end{center}
\vspace{0.5mm}
\subcaption{Symmetrize everything with respect to the last 2-cluster.}
\vspace{1.6mm}
\begin{center}
\begin{tikzpicture}
\node[draw=none,label={[label distance=0mm]east: \texttt{11111111111111111111111111\textcolor{blue}{1YY11111111YY1}111111111111}}] (C) at (-5.75,0) {};
\node[draw=none,label={[label distance=0mm]east: \texttt{\textcolor{blue}{1YY11111111YY1}11111111111111111111111111111111111111}}] (C) at (-5.75,1) {};
\node[draw=none,label={[label distance=0mm]east: \texttt{111111111111\textcolor{blue}{1YY11111111YY1}11111111111111111111111111}}] (C) at (-5.75,0.5) {};
\node[draw=none,label={[label distance=0mm]east:
\texttt{11111111111111111111111111111111111111\textcolor{blue}{1YY11111111YY1}}}] (C) at (-5.75,-0.5) {};
\draw (0.2,-0.2)--(0.2,-0.4);
\draw (0.2,-0.6)--(0.2,-0.8);
\draw (0.2,0.6)--(0.2,0.4);
\draw (0.2,0.2)--(0.2,0);
\draw (0.2,0.4)--(0.2,0.6);
\draw (0.2,0.8)--(0.2,1.0);
\draw (0.2,1.2)--(0.2,1.4);
\end{tikzpicture}
\end{center}
\vspace{0.5mm}
\subcaption{Symmetrize everything with respect to the dashed line.}
\vspace{1.6mm}
\begin{center}
\begin{tikzpicture}
\node[draw=none,label={[label distance=0mm]east: \texttt{11111111111111111111111111\textcolor{blue}{1YY11111111YY1}111111111111}}] (C) at (-5.75,0) {};
\node[draw=none,label={[label distance=0mm]east: \texttt{\textcolor{blue}{1YY11111111YY1}11111111111111111111111111111111111111}}] (C) at (-5.75,1.5) {};
\node[draw=none,label={[label distance=0mm]east: \texttt{1111111111111111111\textcolor{blue}{1YY11111111YY1}1111111111111111111}}] (C) at (-5.75,0.5) {};
\node[draw=none,label={[label distance=0mm]east: \texttt{111111111111\textcolor{blue}{1YY11111111YY1}11111111111111111111111111}}] (C) at (-5.75,1) {};
\node[draw=none,label={[label distance=0mm]east:
\texttt{11111111111111111111111111111111111111\textcolor{blue}{1YY11111111YY1}}}] (C) at (-5.75,-0.5) {};
\draw (0.2,-0.2)--(0.2,-0.4);
\draw (0.2,-0.6)--(0.2,-0.8);
\draw (0.2,0.6)--(0.2,0.4);
\draw (0.2,0.2)--(0.2,0);
\draw (0.2,0.4)--(0.2,0.6);
\draw (0.2,0.8)--(0.2,1.0);
\draw (0.2,1.2)--(0.2,1.4);
\draw (0.2,1.6)--(0.2,1.8);
\end{tikzpicture}
\end{center}
\vspace{0.5mm}
\subcaption{Add the original string in the center.}
\vspace{1.6mm}
\begin{center}
\begin{tikzpicture}
\node[draw=none,label={[label distance=0mm]east: \texttt{11111111111111111111111111\textcolor{blue}{1YY11111111YY1}111111111111}}] (C) at (-5.75,0) {};
\node[draw=none,label={[label distance=0mm]east: \texttt{\textcolor{blue}{1YY11111111YY1}11111111111111111111111111111111111111}}] (C) at (-5.75,1.5) {};
\node[draw=none,label={[label distance=0mm]east: \texttt{1111111111111111111\textcolor{blue}{1YY11111111YY1}1111111111111111111}}] (C) at (-5.75,0.5) {};
\node[draw=none,label={[label distance=0mm]east: \texttt{111111111111\textcolor{blue}{1YY11111111YY1}11111111111111111111111111}}] (C) at (-5.75,1) {};
\node[draw=none,label={[label distance=0mm]east:
\texttt{11111111111111111111111111111111111111\textcolor{blue}{1YY11111111YY1}}}] (C) at (-5.75,-0.5) {};
\node[draw=none,label={[label distance=0mm]east:
\texttt{\textcolor{red}{11111YXXY1111111111111111111111111111111111YXXY11111}}}] (C) at (-5.75,-1) {};
\draw (0.2,-0.2)--(0.2,-0.4);
\draw (0.2,-0.6)--(0.2,-0.8);
\draw (0.2,-1)--(0.2,-1.2);
\draw (0.2,0.6)--(0.2,0.4);
\draw (0.2,0.2)--(0.2,0);
\draw (0.2,0.4)--(0.2,0.6);
\draw (0.2,0.8)--(0.2,1.0);
\draw (0.2,1.2)--(0.2,1.4);
\draw (0.2,1.6)--(0.2,1.8);
\end{tikzpicture}
\end{center}
\vspace{0.5mm}
\subcaption{Add a string made of two copies of YXXY, in such a way that the first YXXY is centered on the leftmost 1YY11111111YY1, and the second YXXY is centered on the rightmost 1YY11111111YY1.}
\vspace{0.9mm}
\end{subfigure}
\caption{The first step of the shifting and unfolding procedure for the trapeze Hamiltonian. Each Pauli acts on a single boundary qubit. Each individual string represents a term in the decomposition of the Hamiltonian into a sum. Portions in blue keep track of the initial pattern through the steps of the procedure, while the new pattern is added in red in the last step.}
\label{fig:StepsShiftUnfold}
\end{figure}

We remind the reader that the trapeze Hamiltonian is a sum over trapeze operators. So when mapped to the boundary, it yields a sum, accordingly. Each individual row on Figure \ref{fig:StepsShiftUnfold} corresponds to a single term of this sum, i.e. to the image of an individual trapeze operator.

Note that from this construction, it is not completely obvious that the strings of Paulis obtained at step $n+1$ will include the ones obtained at step $n$, as the pattern is first shifted. However, due to the first unfolding, the right hand side of the dashed line is perfectly symmetric at every step, which in turn, implies that after the second unfolding, the left hand side of the dashed line is completely identical to its right hand side, which shows that all patterns that do not cross the dashed line are conserved. Moreover, all the patterns that cross the dashed line are added at the end of the procedure. This shows that the symmetry of the unfolding operation guarantees that all strings will be conserved from step $n$ to step $(n+1)$. In a way, the algorithm is \textit{blind to the shift}. Consequently, we can actually see this process as a construction of a sum of explicit, finite strings of Pauli matrices on a fixed, infinite string of qubits which is centered on the dashed line. Moreover, the successive appearances of the elements of the sequence $(S_n)$ make the resulting operator completely scale invariant: the HaPPY code seems to naturally call for a conformal description of the boundary theory.

The shifting and unfolding procedure once again unveils the fractal structure of the HaPPY code. In particular, the fact that the successive maps between levels can be seen as a simple addition of new operators on a fixed, infinite string of qubits is entirely due to scale invariance. Within our setting, this fact does not have a great physical interest, as in fact the operators on this fixed string come from a different level of the bulk at every step. However, this fact allows one to think of the HaPPY code in terms of \textit{reverse engineering} -- to construct an infinite number of bulk levels from a fixed boundary string of qubits -- and recover the same shifted and unfolded Hamiltonian. This other viewpoint, which is more convenient for entropic considerations, is explored in detail in our companion paper \cite{MonicaElliott2}.

\subsection{Failures of the bulk-to-boundary map in an entanglement wedge starting at a trapeze}
We have shown that the bulk-to-boundary map creates a scale invariant shifted and unfolded Hamiltonian. However, one would like to see whether this shifted and unfolded Hamiltonian does anything physical to the boundary theory. We shall see that uberholography works too well: in an entanglement wedge defined by a trapeze, the shifted and unfolded Hamiltonian will not create any long-range correlation on the boundary. This fact shows that the HaPPY code does not keep all its promises: the boundary representative of a physical, long-range entangled theory in the bulk, does not show crucial features that one would expect of a CFT, like algebraic decay of the correlation functions.

As we have seen above, the bulk-to-boundary map of the trapeze Hamiltonian has an uberholographic structure, as it maps a bulk theory into a fractal. More concretely, the successive bulk-to-boundary maps create strings of Pauli operators, which have a lot of $1$'s in between nontrivial Pauli matrices. A striking consequence of this is that two overlapping trapeze operators in the bulk will map to operators with disjoint support in the boundary, as depicted in Figure \ref{fig:TwoOverlapBulkTrapezes}. On the contrary, one of the nice features of the trapeze Hamiltonian is that it couples all the qubits in the bulk. Unfortunately, the disentangling property of the bulk-to-boundary map is unlikely to allow the long-range entanglement to survive on the boundary.

\begin{figure}[H]
\vspace{3mm}
\centering
\begin{tikzpicture}
\draw (-5,1)--(-4,0);
\draw (5,1)--(4,0);
\draw (-8,0)--(8,0);
\draw (-8,-2.5)--(8,-2.5);
\draw (-8,-4.5)--(8,-4.5);
\draw (-4,0)--(4,0);
\draw (-6,-2.5)--(-4,0);
\draw (-4,0)--(-2,-2.5);
\draw (2,-2.5)--(4,0);
\draw (4,0)--(6,-2.5);
\draw (-6,-2.5)--(-4,0);
\draw (-4,0)--(-2,-2.5);
\draw (-6,-2.5)--(-6-0.9,-4.5);
\draw (-6,-2.5)--(-6+0.5,-4.5);
\draw (-2+0.9,-4.5)--(-2,-2.5);
\draw (-2-0.5,-4.5)--(-2,-2.5);
\draw (-4,-2.5)--(-4+0.9,-4.5);
\draw (-4,-2.5)--(-4,-4.5);
\draw (-4,-2.5)--(-4-0.9,-4.5);
\draw (6,-2.5)--(6+0.9,-4.5);
\draw (6,-2.5)--(6-0.5,-4.5);
\draw (2+0.5,-4.5)--(2,-2.5);
\draw (2-0.9,-4.5)--(2,-2.5);
\draw (4,-2.5)--(4+0.9,-4.5);
\draw (4,-2.5)--(4,-4.5);
\draw (4,-2.5)--(4-0.9,-4.5);
\draw (-6-0.9-0.4,-6)--(-6-0.9,-4.5);
\draw (-6-0.9+0.2,-6)--(-6-0.9,-4.5);
\draw (-6-0.2-0.3,-6)--(-6-0.2,-4.5);
\draw (-6-0.2,-6)--(-6-0.2,-4.5);
\draw (-6-0.2+0.3,-6)--(-6-0.2,-4.5);
\draw (-6+0.5-0.2,-6)--(-6+0.5,-4.5);
\draw (-6+0.5+0.2,-6)--(-6+0.5,-4.5);
\draw (-4-0.9-0.2,-6)--(-4-0.9,-4.5);
\draw (-4-0.9+0.1,-6)--(-4-0.9,-4.5);
\draw (-4-0.65,-6)--(-4-0.45,-4.5);
\draw (-4-0.45,-6)--(-4-0.45,-4.5);
\draw (-4-0.25,-6)--(-4-0.45,-4.5);
\draw (-4-0.1,-6)--(-4,-4.5);
\draw (-4+0.1,-6)--(-4,-4.5);
\draw (-4+0.65,-6)--(-4+0.45,-4.5);
\draw (-4+0.45,-6)--(-4+0.45,-4.5);
\draw (-4+0.25,-6)--(-4+0.45,-4.5);
\draw (-4+0.9-0.1,-6)--(-4+0.9,-4.5);
\draw (-4+0.9+0.2,-6)--(-4+0.9,-4.5);
\draw (-2-0.5-0.2,-6)--(-2-0.5,-4.5);
\draw (-2-0.5+0.2,-6)--(-2-0.5,-4.5);
\draw (-2+0.2-0.3,-6)--(-2+0.2,-4.5);
\draw (-2+0.2,-6)--(-2+0.2,-4.5);
\draw (-2+0.2+0.3,-6)--(-2+0.2,-4.5);
\draw (-2+0.9-0.2,-6)--(-2+0.9,-4.5);
\draw (-2+0.9+0.4,-6)--(-2+0.9,-4.5);
\draw (8-6-0.9-0.4,-6)--(8-6-0.9,-4.5);
\draw (8-6-0.9+0.2,-6)--(8-6-0.9,-4.5);
\draw (8-6-0.2-0.3,-6)--(8-6-0.2,-4.5);
\draw (8-6-0.2,-6)--(8-6-0.2,-4.5);
\draw (8-6-0.2+0.3,-6)--(8-6-0.2,-4.5);
\draw (8-6+0.5-0.2,-6)--(8-6+0.5,-4.5);
\draw (8-6+0.5+0.2,-6)--(8-6+0.5,-4.5);
\draw (8-4-0.9-0.2,-6)--(8-4-0.9,-4.5);
\draw (8-4-0.9+0.1,-6)--(8-4-0.9,-4.5);
\draw (8-4-0.65,-6)--(8-4-0.45,-4.5);
\draw (8-4-0.45,-6)--(8-4-0.45,-4.5);
\draw (8-4-0.25,-6)--(8-4-0.45,-4.5);
\draw (8-4-0.1,-6)--(8-4,-4.5);
\draw (8-4+0.1,-6)--(8-4,-4.5);
\draw (8-4+0.65,-6)--(8-4+0.45,-4.5);
\draw (8-4+0.45,-6)--(8-4+0.45,-4.5);
\draw (8-4+0.25,-6)--(8-4+0.45,-4.5);
\draw (8-4+0.9-0.1,-6)--(8-4+0.9,-4.5);
\draw (8-4+0.9+0.2,-6)--(8-4+0.9,-4.5);
\draw (8-2-0.5-0.2,-6)--(8-2-0.5,-4.5);
\draw (8-2-0.5+0.2,-6)--(8-2-0.5,-4.5);
\draw (8-2+0.2-0.3,-6)--(8-2+0.2,-4.5);
\draw (8-2+0.2,-6)--(8-2+0.2,-4.5);
\draw (8-2+0.2+0.3,-6)--(8-2+0.2,-4.5);
\draw (8-2+0.9-0.2,-6)--(8-2+0.9,-4.5);
\draw (8-2+0.9+0.4,-6)--(8-2+0.9,-4.5);
\draw[draw,line width=1mm,color=blue] (-2,-2.5)--(2,-2.5);
\draw[draw,line width=1mm,color=blue] (-2,-2.5)--(-2+0.9,-4.5);
\draw[draw,line width=1mm,color=blue] (-2,-2.5)--(-2-0.5,-4.5);
\draw[draw,line width=1mm,color=blue] (2,-2.5)--(2-0.9,-4.5);
\draw[draw,line width=1mm,color=blue] (-2-0.5,-4.5)--(2.5,-4.5);
\draw[draw,line width=1mm,color=blue] (2+0.5,-4.5)--(2,-2.5);
\draw[draw,line width=1mm,color=purple] (2,-2.5)--(2-0.9,-4.5);
\draw[draw,line width=1mm,color=purple] (2+0.5,-4.5)--(2,-2.5);
\draw[draw,line width=1mm,color=purple] (2-0.9,-4.5)--(2+0.5,-4.5);
\draw[draw,line width=1mm,color=red] (2+0.5,-4.5)--(4,-4.5);
\draw[draw,line width=1mm,color=red] (4,-2.5)--(2,-2.5);
\draw[draw,line width=1mm,color=red] (4,-4.5)--(4,-2.5);
\draw[draw,line width=1mm,color=red] (4-0.9,-4.5)--(4,-2.5);
\node[draw=none,label={[label distance=0.5mm]south:\tiny \textcolor{red}{YY}}] (C) at (3.825,-6) {};
\node[draw=none,label={[label distance=0.5mm]south:\tiny \textcolor{red}{YY}}] (C) at (1.4,-6) {};
\node[draw=none,label={[label distance=0.5mm]south:\tiny \textcolor{blue}{YY}}] (C) at (2.2,-6) {};
\node[draw=none,label={[label distance=0.5mm]south:\tiny \textcolor{blue}{YY}}] (C) at (-2.2,-6) {};
\end{tikzpicture}
\vspace{3mm}

\caption{Two overlapping bulk trapezes are mapped to disjoint boundary regions.}
\label{fig:TwoOverlapBulkTrapezes}
\end{figure}
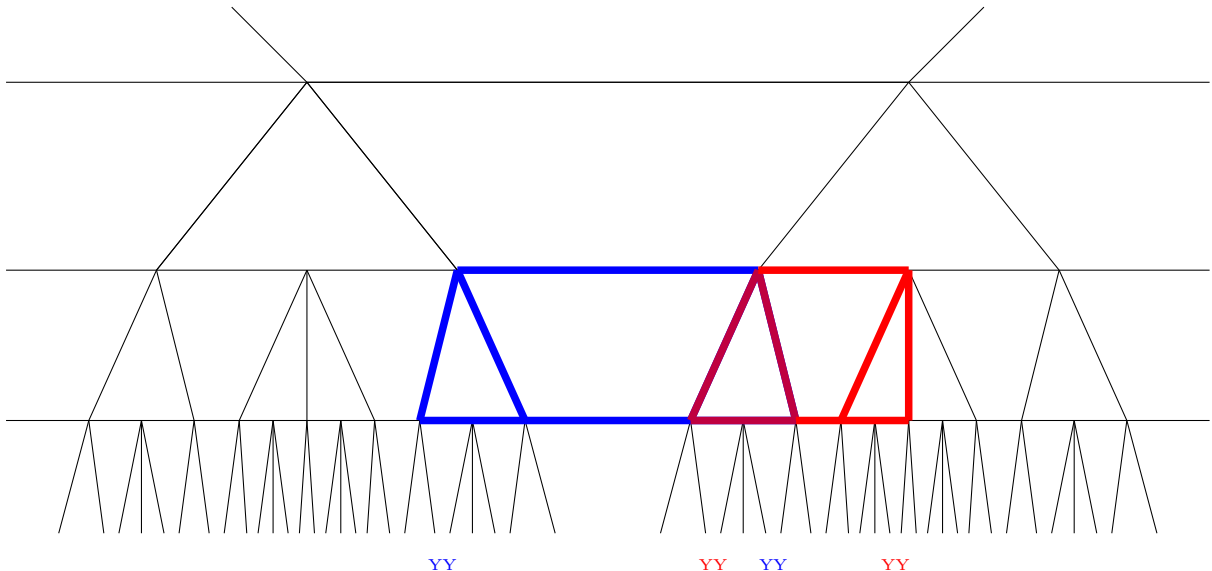

In the remainder of this section, we show that the shifted and unfolded Hamiltonian does not reproduce any feature of a CFT very well. In particular, the obtained theory does not exhibit any long range order.

The basic idea is that uberholography works too well in the case of the trapeze Hamiltonian. Because of the disentangling structure of the bulk-to-boundary map that we previously defined, only very few bulk trapezes will be mapped to overlapping regions of the boundary. We expect that these very minimal overlaps will not have sufficiently large enough entanglement to create any long-range entanglement on the boundary. Nevertheless, proving this claim rigorously seems to be complicated, or even impossible, due to the very high number of cases that can arise. However, by just removing a small proportion of the bulk trapezes (but still getting a qualitatively similar bulk theory), we show this claim in Appendix \ref{sec:trapezeEW}: for any qubit at a finite distance from the dashed line, there exists an upper bound on the number of qubits it can be correlated to by the shifted and unfolded Hamiltonian. Moreover, this upper bound will only depend on the distance of the qubit to the dashed line, and not on the number of bulk levels, hence it carries over to the infinite-dimensional setting. As an illustration, Figure \ref{fig:trapezeoverlapbondary} shows a few examples of nontrivial finite sets of qubits which are coupled together by trapeze interactions. 

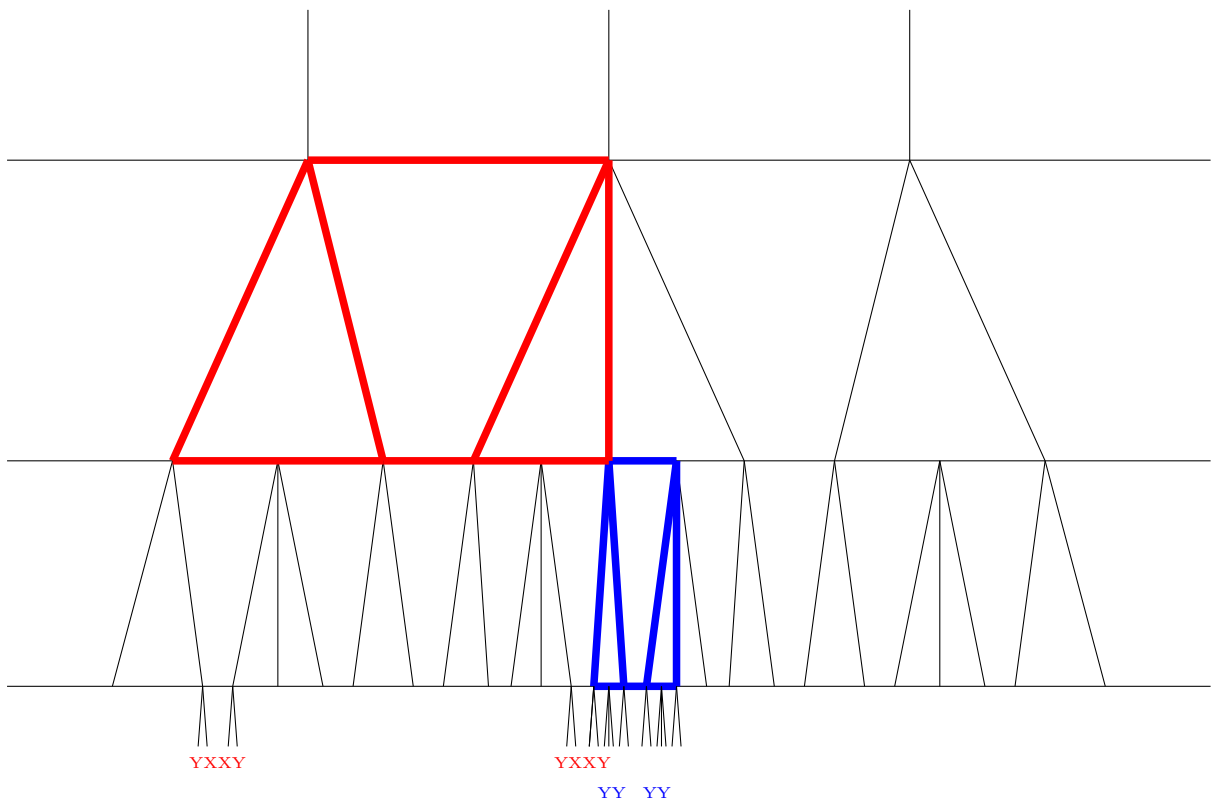
\begin{figure}[H]
\vspace{3mm}
\centering
\begin{tikzpicture}[scale=2]
\draw (2,-2.5)--(2,-1.5);
\draw (4,-2.5)--(4,-1.5);
\draw (6,-2.5)--(6,-1.5);
\draw (0,-2.5)--(8,-2.5);
\draw (0,-4.5)--(8,-4.5);
\draw (0,-6)--(8,-6);
\draw (6,-2.5)--(6+0.9,-4.5);
\draw (6,-2.5)--(6-0.5,-4.5);
\draw (2+0.5,-4.5)--(2,-2.5);
\draw (2-0.9,-4.5)--(2,-2.5);
\draw (4,-2.5)--(4+0.9,-4.5);
\draw (4,-2.5)--(4,-4.5);
\draw (4,-2.5)--(4-0.9,-4.5);
\draw (8-6-0.9-0.4,-6)--(8-6-0.9,-4.5);
\draw (8-6-0.9+0.2,-6)--(8-6-0.9,-4.5);
\draw (8-6-0.2-0.3,-6)--(8-6-0.2,-4.5);
\draw (8-6-0.2,-6)--(8-6-0.2,-4.5);
\draw (8-6-0.2+0.3,-6)--(8-6-0.2,-4.5);
\draw (8-6+0.5-0.2,-6)--(8-6+0.5,-4.5);
\draw (8-6+0.5+0.2,-6)--(8-6+0.5,-4.5);
\draw (8-4-0.9-0.2,-6)--(8-4-0.9,-4.5);
\draw (8-4-0.9+0.1,-6)--(8-4-0.9,-4.5);
\draw (8-4-0.65,-6)--(8-4-0.45,-4.5);
\draw (8-4-0.45,-6)--(8-4-0.45,-4.5);
\draw (8-4-0.25,-6)--(8-4-0.45,-4.5);
\draw (8-4-0.1,-6)--(8-4,-4.5);
\draw (8-4+0.1,-6)--(8-4,-4.5);
\draw (8-4+0.65,-6)--(8-4+0.45,-4.5);
\draw (8-4+0.45,-6)--(8-4+0.45,-4.5);
\draw (8-4+0.25,-6)--(8-4+0.45,-4.5);
\draw (8-4+0.9-0.1,-6)--(8-4+0.9,-4.5);
\draw (8-4+0.9+0.2,-6)--(8-4+0.9,-4.5);
\draw (8-2-0.5-0.2,-6)--(8-2-0.5,-4.5);
\draw (8-2-0.5+0.2,-6)--(8-2-0.5,-4.5);
\draw (8-2+0.2-0.3,-6)--(8-2+0.2,-4.5);
\draw (8-2+0.2,-6)--(8-2+0.2,-4.5);
\draw (8-2+0.2+0.3,-6)--(8-2+0.2,-4.5);
\draw (8-2+0.9-0.2,-6)--(8-2+0.9,-4.5);
\draw (8-2+0.9+0.4,-6)--(8-2+0.9,-4.5);
\draw[draw,line width=1mm,color=blue] (4,-4.5)--(4.1,-6);
\draw[draw,line width=1mm,color=blue] (4,-4.5)--(3.9,-6);
\draw[draw,line width=1mm,color=blue] (4+0.45,-4.5)--(4+0.45,-6);
\draw[draw,line width=1mm,color=blue] (4+0.45,-4.5)--(4+0.25,-6);
\draw[draw,line width=1mm,color=blue] (4,-4.5)--(4.45,-4.5);
\draw[draw,line width=1mm,color=blue] (3.9,-6)--(4.45,-6);
\draw[draw,line width=1mm,color=red] (2,-2.5)--(2-0.9,-4.5);
\draw[draw,line width=1mm,color=red] (2+0.5,-4.5)--(2,-2.5);
\draw[draw,line width=1mm,color=red] (2-0.9,-4.5)--(2+0.5,-4.5);
\draw[draw,line width=1mm,color=red] (2+0.5,-4.5)--(4,-4.5);
\draw[draw,line width=1mm,color=red] (4,-2.5)--(2,-2.5);
\draw[draw,line width=1mm,color=red] (4,-4.5)--(4,-2.5);
\draw[draw,line width=1mm,color=red] (4-0.9,-4.5)--(4,-2.5);
\draw (3.9,-6)--(3.87,-6.4);
\draw (3.9,-6)--(3.93,-6.4);
\draw (4,-6)--(3.97,-6.4);
\draw (4,-6)--(4,-6.4);
\draw (4,-6)--(4.03,-6.4);
\draw (4.1,-6)--(4.07,-6.4);
\draw (4.1,-6)--(4.13,-6.4);
\draw (3.9+0.35,-6)--(3.87+0.35,-6.4);
\draw (3.9+0.35,-6)--(3.93+0.35,-6.4);
\draw (4+0.35,-6)--(3.97+0.35,-6.4);
\draw (4+0.35,-6)--(4+0.35,-6.4);
\draw (4+0.35,-6)--(4.03+0.35,-6.4);
\draw (4.1+0.35,-6)--(4.07+0.35,-6.4);
\draw (4.1+0.35,-6)--(4.13+0.35,-6.4);
\draw (2-0.7,-6)--(2-0.73,-6.4);
\draw (2-0.7,-6)--(2-0.67,-6.4);
\draw (2-0.5,-6)--(2-0.53,-6.4);
\draw (2-0.5,-6)--(2-0.47,-6.4);
\draw (3.9,-6)--(3.87,-6.4);
\draw (3.9-0.15,-6)--(3.93-0.15,-6.4);
\draw (3.9-0.15,-6)--(3.87-0.15,-6.4);
\node[draw=none,label={[label distance=0.5mm]south:\tiny \textcolor{red}{YXXY}}] (C) at (3.825,-6.3) {};
\node[draw=none,label={[label distance=0.5mm]south:\tiny \textcolor{red}{YXXY}}] (C) at (1.4,-6.3) {};
\node[draw=none,label={[label distance=0.5mm]south:\tiny \textcolor{blue}{YY}}] (C) at (4.02,-6.5) {};
\node[draw=none,label={[label distance=0.5mm]south:\tiny \textcolor{blue}{YY}}] (C) at (4.32,-6.5) {};
\end{tikzpicture}
\vspace{3mm}
\caption{Two trapezes in the bulk that overlap on the boundary. The supports of the boundary images of the red and blue trapezes share one common qubit.}
\label{fig:trapezeoverlapbondary}
\end{figure}

In AdS/CFT, we expect the bulk-to-boundary map to correlate qubits enough that a fully entangled theory in the bulk will have a boundary dual with algebraic decay. Here, we have shown that it is not always the case for the HaPPY code, as a theory with nontrivial behavior in the bulk seems to be mapped to something which exhibits no long-range order in the boundary. Indeed, the shifted and unfolded Hamiltonian on the boundary is scale invariant, but doesn't seem to exhibit any long-range entanglement.

\subsection{A bulk reconstruction with nontrivial boundary correlation functions}
In an entanglement wedge defined by a trapeze, there is no easily recognizable infinite subset of boundary qubits which is entangled. This result is very important, but still depends on the fact that at each level, the expansion is done around the straight dashed line of Figure \ref{fig:dottedline}. For another choice of dashed line, which is not straight, we show in this subsection that the HaPPY code can achieve a bit better: it can entangle together an infinite but fractal subset of the boundary qubits together. Within this infinite subset, we estimate the corresponding thermal correlation functions.

\begin{figure}[H]
\vspace{5mm}
\centering
\begin{tikzpicture}
\draw (0,0)--(14,0);
\draw (0,-2.5)--(14,-2.5);
\draw (0,-4.5)--(14,-4.5);
\draw (4,0)--(2,-2.5);
\draw[draw,line width=1mm,color=red] (4,0)--(4,-2.5);
\draw[draw,line width=1mm,color=red] (4,0)--(6,-2.5);
\draw[draw,line width=1mm,color=red] (10,0)--(8,-2.5);
\draw[draw,line width=1mm,color=red] (10,0)--(12,-2.5);
\draw (2,-2.5)--(2-0.5,-4.5);
\draw (2,-2.5)--(2+0.3,-4.5);
\draw[draw,line width=1mm,color=red] (4,-2.5)--(4-0.3,-4.5);
\draw[draw,line width=1mm,color=red] (4,-2.5)--(4+0.3,-4.5);
\draw[draw,line width=1mm,color=red] (3,-2.5)--(3.5,-4.5);
\draw (3,-2.5)--(2.5,-4.5);
\draw[draw,line width=1mm,color=red] (3,-2.5)--(3,-4.5);
\draw (5,-2.5)--(5.5,-4.5);
\draw (5,-2.5)--(4.5,-4.5);
\draw (5,-2.5)--(5,-4.5);
\draw (6,-2.5)--(6-0.3,-4.5);
\draw (6,-2.5)--(6+0.3,-4.5);
\draw (14-2,-2.5)--(14-2+0.3,-4.5);
\draw (14-2,-2.5)--(14-2-0.3,-4.5);
\draw (14-4,-2.5)--(14-4+0.5,-4.5);
\draw (14-4,-2.5)--(14-4-0.5,-4.5);
\draw (14-4,-2.5)--(14-4,-4.5);
\draw (14-6,-2.5)--(14-6+0.3,-4.5);
\draw (14-6,-2.5)--(14-6-0.3,-4.5);
\draw[draw,line width=1mm,color=red] (3,-4.5)--(2.9,-6);
\draw[draw,line width=1mm,color=red] (3,-4.5)--(3.1,-6);
\draw[draw,line width=1mm,color=red] (2.7,-4.5)--(2.7,-6);
\draw[draw,line width=1mm,color=red] (2.7,-4.5)--(2.8,-6);
\draw (2.7,-4.5)--(2.6,-6);
\draw[draw,line width=1mm,color=red] (2.7,-4.5)--(4.3,-4.5);
\draw[draw,line width=1mm,color=red] (2.7,-6)--(3.1,-6);
\draw[draw,line width=1mm,color=red] (2.7,-4.5)--(3,-4.5);
\draw[draw,line width=1mm,color=red] (3,-2.5)--(12,-2.5);
\draw[draw,line width=1mm,color=red] (4,0)--(14,0);
\draw[draw,line width=1mm,color=red] (10,0)--(10,1);
\draw[draw,line width=1mm] (8.5,1)--(5.5,-1);
\draw[draw,line width=1mm] (5,-2)--(3.3,-4);
\draw[draw,line width=1mm] (3,-4.5)--(2.7,-6.5);
\end{tikzpicture}
\vspace{5mm}

\begin{tikzpicture}[color=red,scale=0.4]
\draw (-4,0)--(-4,1);
\draw (4,0)--(4,1);
\draw (-8,0)--(8,0);
\draw (-8,-2.5)--(-6,-2.5)--(-4,-2.5)--(-2,-2.5)--(2,-2.5)--(4,-2.5)--(6,-2.5)--(8,-2.5);
\draw (-4,0)--(4,0);
\draw (-6,-2.5)(-4,0);
\draw (-4,0)(-2,-2.5);
\draw (2,-2.5)--(4,0);
\draw (4,0)--(6,-2.5);
\draw (-6,-2.5)--(-4,0);
\draw (-4,0)--(-2,-2.5);
\draw (-6,-2.5)--(-6-0.9,-4.5);
\draw (-6,-2.5)--(-6+0.9,-4.5);
\draw (-2+0.9,-4.5)--(-2,-2.5);
\draw (-2-0.9,-4.5)--(-2,-2.5);
\draw (-4,-2.5)--(-4+0.9,-4.5);
\draw (-4,-2.5)--(-4,-4.5);
\draw (-4,-2.5)--(-4-0.9,-4.5);
\draw (6,-2.5)--(6+0.9,-4.5);
\draw (6,-2.5)--(6-0.9,-4.5);
\draw (2+0.9,-4.5)--(2,-2.5);
\draw (2-0.9,-4.5)--(2,-2.5);
\draw (4,-2.5)--(4+0.9,-4.5);
\draw (4,-2.5)--(4,-4.5);
\draw (4,-2.5)--(4-0.9,-4.5);
\node[draw=none,label={[label distance=0.5mm]south:\tiny YY...}] (C) at (-4-1,-4) {};
\node[draw=none,label={[label distance=0.5mm]south:\tiny ...YY}] (C) at (5,-4) {};
\node[draw=none,label={[label distance=0.5mm]south:\tiny YXXY...}] (C) at (6.7,-5) {};
\node[draw=none,label={[label distance=0.5mm]south:\tiny ...YXXY}] (C) at (20,-5) {};
\node[draw=none,label={[label distance=0.5mm]south:\tiny YXXY11YY11YXXY...}] (C) at (23.7,-6) {};
\end{tikzpicture}
\vspace{3mm}
\caption{The boundary mappings of the sequence of red trapezes form an infinite string of entangled boundary qubits. The dashed line, which shows where the bulk is centered at each step, is shifted so that the trapezes in red pile up. This choice entangles an unbounded number of boundary qubits as the tensor network goes infinite, as shown by the overlaps of Paulis on the bottom of the figure, which correspond to the boundary images of the red trapezes at a fixed level.}
\label{fig:InfiniteBdryQubits}
\end{figure}
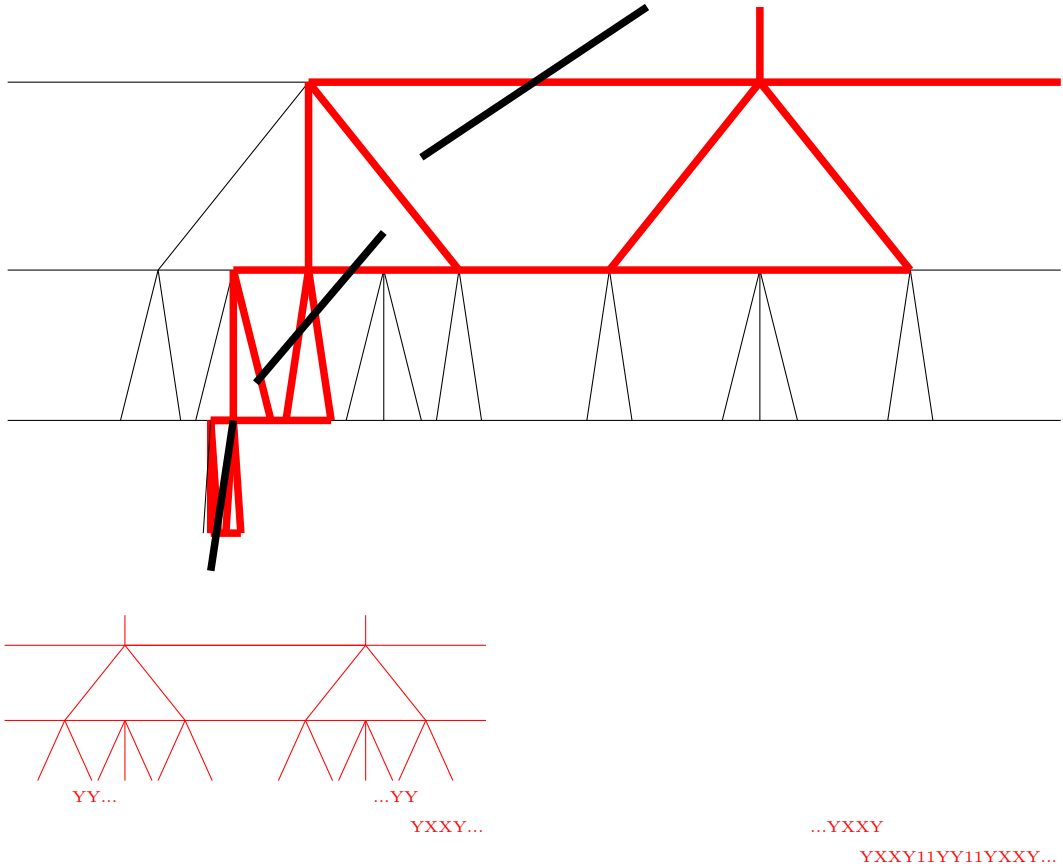

If the bulk structure had been different than the one of the infinite entanglement wedge of a trapeze, it would be possible to have an infinite subregion of the boundary entangled together, thanks to successive trapezes made of one 2-cluster and one side of a 3-cluster. This sequence of bulk trapezes is illustrated in Figure \ref{fig:InfiniteBdryQubits}. Therefore, if one expands at each level around the shifted dashed line of Figure \ref{fig:InfiniteBdryQubits}, all these trapezes will appear at each level, and the lowest one will be at a fixed finite distance of the shifted dashed line. An always growing string of entangled boundary qubits will then be created, and in the infinite-dimensional limit, its size will go to infinity.  Moreover, this string will have a fractal structure, as each new interaction will come from an upper level in the bulk. The separation between the qubits in the interaction will therefore be exponentially growing.

If we restrict ourselves to this fractal subset of the boundary, we expect our correlations to qualitatively behave like a 1d Ising model. At inverse temperature $\beta$, we therefore expect a correlation function between $x$ and $y$ to be of the form:$$C(x,y)\sim e^{-\beta d_{\mathrm{fractal}}(x,y)},$$where $d_{\mathrm{fractal}}$ denotes the number of qubits within the fractal that separate $x$ and $y$. As the gaps are exponentially growing, we then expect that$$d_{\mathrm{fractal}}(x,y)\sim \left| x-y\right|,$$ which in turn yields:$$C(x,y)\sim \left| x-y\right|^{-\beta}.$$ This algebraic decay is reminiscent of a Conformal Field Theory correlation function (although it should be smoothed out by a hyperbolic function), which allows us to say that in a much weaker sense, in the very rare cases when the HaPPY code does couple infinite strings of boundary qubits, it does it approximately like a CFT would, at least at low enough temperature.

This confirms our belief that a finite-dimensional modeling of the HaPPY tensor network, with limitations, does still comprise some holographic aspects, in a particularly low temperature setting. If we consider a low temperature state, we expect to be able to capture some of the long-range physics. In other words, even though the physics of the original HaPPY code does not contain as much entanglement as fully holographic theories, it gives enough information to still see long-range effects. It is quite remarkable to notice that uberholography still ``knows" about the underlying CFT structure, and that even when it decouples almost everything like in the HaPPY code, a relic of approximate CFT behavior can still be found in fractals.

\subsection{A bulk-to-boundary map with infinite energy}

Similarly, we can take a shifted dashed line which alternates between trapezes made of one 2-cluster and one 3-cluster and trapezes made of two 2-clusters, as shown on Figure \ref{fig:TrpezesToQubit}. In this picture, another situation arises: an infinite number of bulk operators will be mapped onto the same boundary qubit, at every level! This shows that an infinite-dimensional Hamiltonian on a HaPPY-like code will in some cases also need to be regulated at a certain scale, just like in regular AdS/CFT.

\begin{figure}[H]
\vspace{5mm}
\centering
\begin{tikzpicture}
\draw (-5,1)--(-4,0);
\draw (5,1)--(4,0);
\draw (-8,0)--(8,0);
\draw (-8,-2.5)--(8,-2.5);
\draw (-8,-4.5)--(8,-4.5);
\draw (-4,0)--(4,0);
\draw (-6,-2.5)--(-4,0);
\draw (-4,0)--(-2,-2.5);
\draw (2,-2.5)--(4,0);
\draw (4,0)--(6,-2.5);
\draw (-6,-2.5)--(-4,0);
\draw (-4,0)--(-2,-2.5);
\draw (-6,-2.5)--(-6-0.9,-4.5);
\draw (-6,-2.5)--(-6+0.5,-4.5);
\draw (-2+0.9,-4.5)--(-2,-2.5);
\draw (-2-0.5,-4.5)--(-2,-2.5);
\draw (-4,-2.5)--(-4+0.9,-4.5);
\draw (-4,-2.5)--(-4,-4.5);
\draw (-4,-2.5)--(-4-0.9,-4.5);
\draw (6,-2.5)--(6+0.9,-4.5);
\draw (6,-2.5)--(6-0.5,-4.5);
\draw (2+0.5,-4.5)--(2,-2.5);
\draw (2-0.9,-4.5)--(2,-2.5);
\draw (4,-2.5)--(4+0.9,-4.5);
\draw (4,-2.5)--(4,-4.5);
\draw (4,-2.5)--(4-0.9,-4.5);
\draw (-6-0.9-0.4,-6)--(-6-0.9,-4.5);
\draw (-6-0.9+0.2,-6)--(-6-0.9,-4.5);
\draw (-6-0.2-0.3,-6)--(-6-0.2,-4.5);
\draw (-6-0.2,-6)--(-6-0.2,-4.5);
\draw (-6-0.2+0.3,-6)--(-6-0.2,-4.5);
\draw (-6+0.5-0.2,-6)--(-6+0.5,-4.5);
\draw (-6+0.5+0.2,-6)--(-6+0.5,-4.5);
\draw (-4-0.9-0.2,-6)--(-4-0.9,-4.5);
\draw (-4-0.9+0.1,-6)--(-4-0.9,-4.5);
\draw (-4-0.65,-6)--(-4-0.45,-4.5);
\draw (-4-0.45,-6)--(-4-0.45,-4.5);
\draw (-4-0.25,-6)--(-4-0.45,-4.5);
\draw (-4-0.1,-6)--(-4,-4.5);
\draw (-4+0.1,-6)--(-4,-4.5);
\draw (-4+0.65,-6)--(-4+0.45,-4.5);
\draw (-4+0.45,-6)--(-4+0.45,-4.5);
\draw (-4+0.25,-6)--(-4+0.45,-4.5);
\draw (-4+0.9-0.1,-6)--(-4+0.9,-4.5);
\draw (-4+0.9+0.2,-6)--(-4+0.9,-4.5);
\draw (-2-0.5-0.2,-6)--(-2-0.5,-4.5);
\draw (-2-0.5+0.2,-6)--(-2-0.5,-4.5);
\draw (-2+0.2-0.3,-6)--(-2+0.2,-4.5);
\draw (-2+0.2,-6)--(-2+0.2,-4.5);
\draw (-2+0.2+0.3,-6)--(-2+0.2,-4.5);
\draw (-2+0.9-0.2,-6)--(-2+0.9,-4.5);
\draw (-2+0.9+0.4,-6)--(-2+0.9,-4.5);
\draw (8-6-0.9-0.4,-6)--(8-6-0.9,-4.5);
\draw (8-6-0.9+0.2,-6)--(8-6-0.9,-4.5);
\draw (8-6-0.2-0.3,-6)--(8-6-0.2,-4.5);
\draw (8-6-0.2,-6)--(8-6-0.2,-4.5);
\draw (8-6-0.2+0.3,-6)--(8-6-0.2,-4.5);
\draw (8-6+0.5-0.2,-6)--(8-6+0.5,-4.5);
\draw (8-6+0.5+0.2,-6)--(8-6+0.5,-4.5);
\draw (8-4-0.9-0.2,-6)--(8-4-0.9,-4.5);
\draw (8-4-0.9+0.1,-6)--(8-4-0.9,-4.5);
\draw (8-4-0.65,-6)--(8-4-0.45,-4.5);
\draw (8-4-0.45,-6)--(8-4-0.45,-4.5);
\draw (8-4-0.25,-6)--(8-4-0.45,-4.5);
\draw (8-4-0.1,-6)--(8-4,-4.5);
\draw (8-4+0.1,-6)--(8-4,-4.5);
\draw (8-4+0.65,-6)--(8-4+0.45,-4.5);
\draw (8-4+0.45,-6)--(8-4+0.45,-4.5);
\draw (8-4+0.25,-6)--(8-4+0.45,-4.5);
\draw (8-4+0.9-0.1,-6)--(8-4+0.9,-4.5);
\draw (8-4+0.9+0.2,-6)--(8-4+0.9,-4.5);
\draw (8-2-0.5-0.2,-6)--(8-2-0.5,-4.5);
\draw (8-2-0.5+0.2,-6)--(8-2-0.5,-4.5);
\draw (8-2+0.2-0.3,-6)--(8-2+0.2,-4.5);
\draw (8-2+0.2,-6)--(8-2+0.2,-4.5);
\draw (8-2+0.2+0.3,-6)--(8-2+0.2,-4.5);
\draw (8-2+0.9-0.2,-6)--(8-2+0.9,-4.5);
\draw (8-2+0.9+0.4,-6)--(8-2+0.9,-4.5);
\draw (-8,-6)--(8,-6);
\draw[draw,line width=1mm,color=red] (-4,0)--(4,0);
\draw[draw,line width=1mm,color=red] (-4,0)--(-6,-2.5);
\draw[draw,line width=1mm,color=red] (-4,0)--(-2,-2.5);
\draw[draw,line width=1mm,color=red] (4,0)--(2,-2.5);
\draw[draw,line width=1mm,color=red] (4,0)--(6,-2.5);
\draw[draw,line width=1mm,color=red] (-6,-2.5)--(6,-2.5);
\draw[draw,line width=1mm,color=red] (-6+0.5,-4.5)--(-5.3,-6);
\draw[draw,line width=1mm,color=red] (-6+0.5,-4.5)--(-5.7,-6);
\draw[draw,line width=1mm,color=red] (-6+0.5+0.6,-4.5)--(-5.4+0.6,-6);
\draw[draw,line width=1mm,color=red] (-6+0.5+0.6,-4.5)--(-5.7+0.6,-6);
\draw[draw,line width=1mm,color=red] (-6+0.5-0.2,-6)--(-5.4+0.6,-6);
\draw[draw,line width=1mm,color=red] (-6+0.5,-4.5)--(-4-0.9,-4.5);
\draw[draw,line width=1mm] (1.5,1)--(-1.5,-1);
\draw[draw,line width=1mm] (-2.5,-1.5)--(-4,-3);
\draw[draw,line width=1mm] (-4.5,-3.5)--(-5.2,-5);
\draw[draw,line width=1mm] (-5.2,-5.5)--(-5.4,-6.5);
\draw[draw,line width=1mm] (-4.5,-3.5)--(-5.2,-5);
\end{tikzpicture}
\vspace{10mm}

\begin{tikzpicture}[color=red,scale=0.4]
\draw (-4,0)--(-4,1);
\draw (4,0)--(4,1);
\draw (-8,0)--(8,0);
\draw (-8,-2.5)--(-6,-2.5)--(-4,-2.5)--(-2,-2.5)--(2,-2.5)--(4,-2.5)--(6,-2.5)--(8,-2.5);
\draw (-4,0)--(4,0);
\draw (-6,-2.5)(-4,0);
\draw (-4,0)(-2,-2.5);
\draw (2,-2.5)--(4,0);
\draw (4,0)--(6,-2.5);
\draw (-6,-2.5)--(-4,0);
\draw (-4,0)--(-2,-2.5);
\draw (-6,-2.5)--(-6-0.9,-4.5);
\draw (-6,-2.5)--(-6+0.9,-4.5);
\draw (-2+0.9,-4.5)--(-2,-2.5);
\draw (-2-0.9,-4.5)--(-2,-2.5);
\draw (-4,-2.5)--(-4+0.9,-4.5);
\draw (-4,-2.5)--(-4,-4.5);
\draw (-4,-2.5)--(-4-0.9,-4.5);
\draw (6,-2.5)--(6+0.9,-4.5);
\draw (6,-2.5)--(6-0.9,-4.5);
\draw (2+0.9,-4.5)--(2,-2.5);
\draw (2-0.9,-4.5)--(2,-2.5);
\draw (4,-2.5)--(4+0.9,-4.5);
\draw (4,-2.5)--(4,-4.5);
\draw (4,-2.5)--(4-0.9,-4.5);
\node[draw=none,label={[label distance=0.5mm]south:\tiny YY...}] (C) at (-4-1,-4) {};
\node[draw=none,label={[label distance=0.5mm]south:\tiny YXXY...}] (C) at (-4-1,-5) {};
\node[draw=none,label={[label distance=0.5mm]south:\tiny YXXY...}] (C) at (-4-1,-6) {};
\node[draw=none,label={[label distance=0.5mm]south:\tiny ...}] (C) at (-4-1,-7) {};
\end{tikzpicture}
\vspace{5mm}
\caption{All the red trapezes are mapped onto a similar qubit, at every level. The dashed line, which shows where the bulk is centered at each step, is shifted so that the trapezes in red pile up. As shown here, this results in an infinite number of nontrivial Pauli actions accumulating on the same qubits, which makes the energy infinite when the tensor network grows. Each string of boundary Paulis on the bottom of the figure represents the action of one red trapeze at a fixed level, and all strings act nontrivially on the same boundary qubit.}
\label{fig:TrpezesToQubit}
\end{figure}
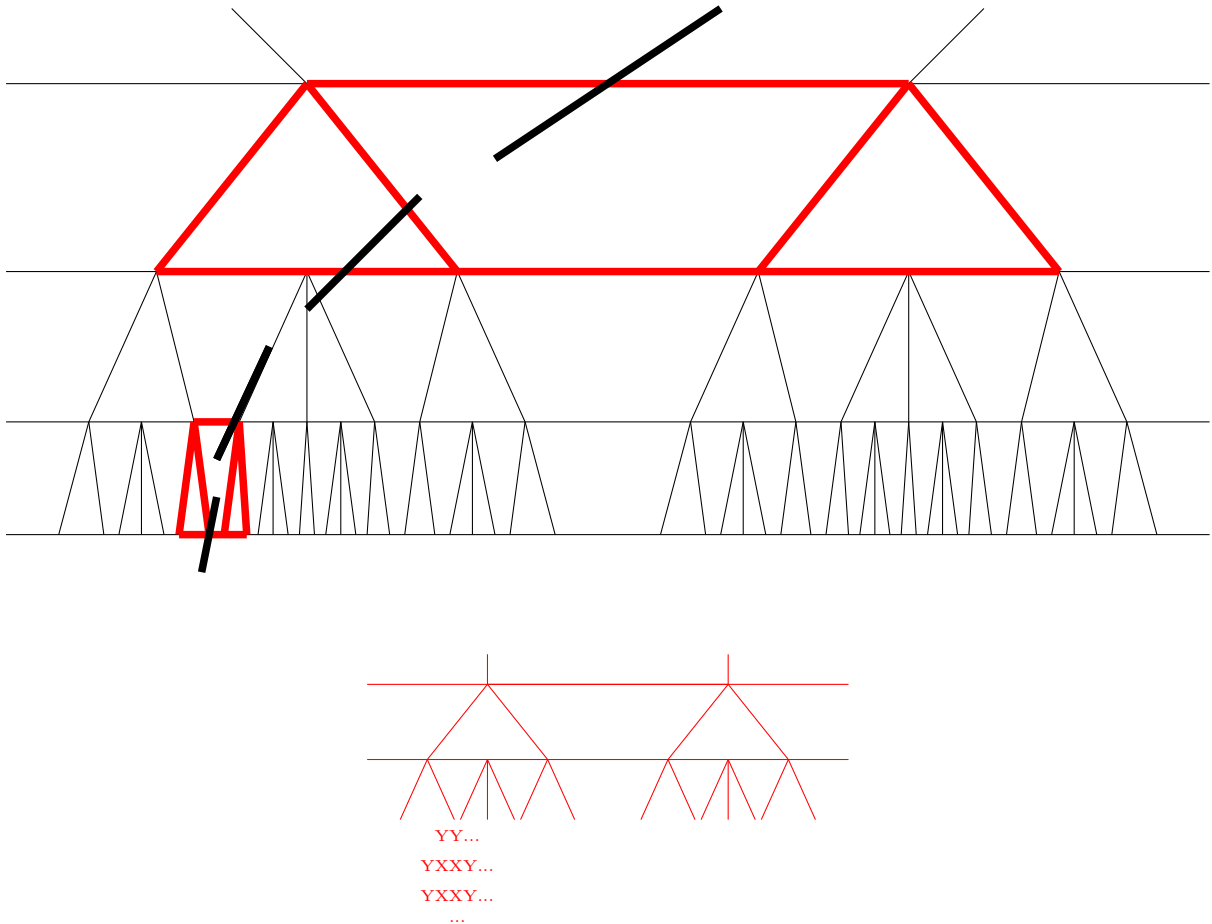

As a summary, in the infinite-dimensional setting, the HaPPY code seems to fail to send the trapeze Hamiltonian, which describes a nontrivial physical theory, to a CFT-like boundary Hamiltonian. Whilst the shifting and unfolding procedure shows that a scale-invariant pattern can be found, long range entanglement is going to be lost on the boundary. It is a consequence of the uberholographic structure of the bulk-to-boundary mapping, which creates too much disentanglement in the case of the trapeze Hamiltonian to have enough boundary correlations. With different bulk structures than the trapeze entanglement wedge, one can get a bit more correlations, on infinite, yet fractal and rare, subsets of the boundary. One can also, in some cases, send an infinite number of trapezes to the same boundary qubits, which shows that the HaPPY code, in spite of its disentangling structure, is not immune to regulation either.

On the brighter side, one could see this result as a success of the quantum error correction perspective on quantum gravity. Indeed, as we show in \cite{MonicaElliott2}, the infinite-dimensional codes we describe still have a Ryu--Takayanagi-like behavior, at least in the reverse engineering context, as they can preserve relative entropies between the bulk and the boundary, according to the Jafferis-Lewkowycz-Maldacena-Suh prescription. It is therefore interesting to notice that within the framework of quantum error correction, the boundary theory only barely needs to be entangled to give rise to a physically nontrivial theory in the bulk. This result could open a path to more models which could describe gravitational properties without the entanglement structure of a CFT.

\section{Discussion}
\label{sec:conclusion}

In this paper, we showed that the infinite-dimensional HaPPY code fails to send a physical theory with long-range entanglement in the bulk to a CFT-resembling theory on the boundary. In the case of the trapeze Hamiltonian, we showed that if we just remove a very small number of trapeze operators in the bulk, which shouldn't change the qualitative bulk behavior of the theory, the mapped operator on the infinite boundary is scale invariant, but doesn't exhibit any uniform algebraic decay behavior, as it only couples finite numbers of boundary qubits together. This shows that the HaPPY code has its limits as a toy model of AdS/CFT. It would be nice to generalize our result to a more important class of interacting theories in the bulk, to see how much the HaPPY code fails to map dynamics.

In fact, it is interesting from a physical viewpoint that it is possible to encode a bulk theory with interaction and long-range decay into independent finite-dimensional interaction terms. Since, as we show in \cite{MonicaElliott2}, the infinite-dimensional HaPPY code still inherits nice entropy properties from its finite-dimensional counterpart, it is also interesting to see that it is a different toy model of the semiclassical regime of quantum gravity than AdS/CFT. Another interesting research direction would therefore be to look for other infinite-dimensional quantum error correcting codes which simulate bulk physics well in terms of entropy, maybe even outside of an explicit bulk to boundary setting.

Note that in this paper, we have not touched upon the entropic properties of the infinite-dimensional HaPPY code. This important question is extensively treated in our companion paper \cite{MonicaElliott2}, in a more general context which works for any model based on operator pushing. Interestingly, dynamics still play a key role, as the existence of thermal states on the boundary, which is directly related to the Hamiltonian, is the cornerstone of our construction of bulk and boundary Hilbert spaces.

The HaPPY code has a direct implication to AdS$_2$/CFT$_1$. A CFT$_1$ is a one-dimensional theory with a global $SL(2,\mathbb{R})$ symmetry, which is the isometry group of the Poincare disk. We must find a boundary state (defined by a state of the bulk qubits and the HaPPY tensor network) for which boundary correlators of primary operators in the continuum limit satisfy the Ward identities of $SL(2,\mathbb{R})$. The correlators are then interpreted as correlators in an $SL(2,\mathbb{R})$-invariant vacuum. The AdS$_2$/CFT$_1$ interpretation of the HaPPY code is that using the HaPPY code to compute boundary correlators is equivalent to performing a discretized 
 path integral on the Poincare disk. The CFT$_1$ correlators must be thought of as Euclidean time correlators because Euclidean time is periodic. The bulk indices allow one to insert bulk operators in the Euclidean path integral.

However, one cannot really interpret the HaPPY code in terms of entanglement wedge reconstruction in such a setup, as an entanglement wedge lies in space but not in time. Henceforth, this AdS$_2$/CFT$_1$ model should be viewed as a time slice of the higher-dimensional AdS$_3$/CFT$_2$, where the former Euclidean time becomes an actual space direction. It is then interesting to note that a Hamiltonian on the boundary theory defines a preferred time direction for the dynamics of the system, and can subsequently evolve the bulk hypersurface via subregion duality. It is not clear whether this time direction is the same as the time direction of, for instance, Lorentz boosts or geometric transformations of the boundary theory. However, if the system is in a thermal state, the time evolution corresponds to modular time, and results from Bisognano, Wichmann, Hislop and Longo summarized in \cite{Haag} allow one to link modular time and geometric transformations like Lorentz boosts, for specific subsets of the boundary such as Rindler wedges and diamonds. In this sense, Hamiltonian evolution can correspond to actual relativistic transformations of the tensor network underlying the HaPPY code. We hope to return to this point in future work and draw a link with the result of Jafferis-Lewkowycz-Maldacena-Suh \cite{Jafferis:2015del}.

\section*{Acknowledgments}
The authors are grateful to Vincent Chen, Adrian Franco-Rubio, David Kolchmeyer, and Matilde Marcolli for discussions and Craig Lawrie for helpful comments on this paper. M.J.K. is supported by a Sherman Fairchild Postdoctoral Fellowship. This material is based upon work supported by the U.S. Department of Energy, Office of Science, Office of High Energy Physics, under Award Number DE-SC0011632. 
E.G. is funded by ENS Paris and would like to thank Matilde Marcolli for her guidance and constant support.
Research at Perimeter Institute is supported in part by the Government of Canada through the Department of Innovation, Science and Economic Development Canada and by the Province of Ontario through the Ministry of Colleges and Universities.

\appendix

\section{Short-range entanglement in an entanglement wedge defined by a trapeze}
\label{sec:trapezeEW}
In this appendix, we show that the shifted and unfolded Hamiltonian only creates short-range entanglement near the dashed line if we carefully remove a small number of bulk trapezes.

The idea is to remove a small number of bulk trapezes that may create long-range entanglement and observe its differences in physics. As we have shown in Section \ref{sec:bulktoboundarytrapeze}, the shifted and unfolded Hamiltonian comes from piled up 2-clusters on each side of the dashed line. These 2-clusters are neighbored by 3-clusters, shown in blue on Figure \ref{fig:TrapezeRemoved}. The only trapezes which are below the 3-clusters and couple the right and the left of them are represented in red on Figure \ref{fig:TrapezeRemoved}.

Our proposal is to remove these red trapeze operators from the bulk Hamiltonian. Then we find that it does not change the physics since they are in a small proportion. Indeed, under each blue 3-cluster, there are only two such trapezes every other line. Note that these trapezes appear in a perfectly regular way when the bulk level increases. At sufficiently long distance in the renormalization group flow, as it is defined in Section \ref{sec:MERA}, it is not possible to distinguish that these trapezes have been removed. This is expected from the perspective of AdS/CFT from tensor networks as it requires an understanding via RG flow. Thus, there will still be long-range entanglement in the bulk and one expects the bulk mean field theory analysis we performed to still be meaningful in this new context.

Now, take a qubit on an arbitrary truncated boundary level. We will find an upper bound on the number of qubits it can be correlated to by the shifted and unfolded Hamiltonian, which is independent of the level, and depends only on its distance to the dashed line. In order to do so, we will use Figure \ref{fig:CouplingFail} as a graphical support for our reasoning.

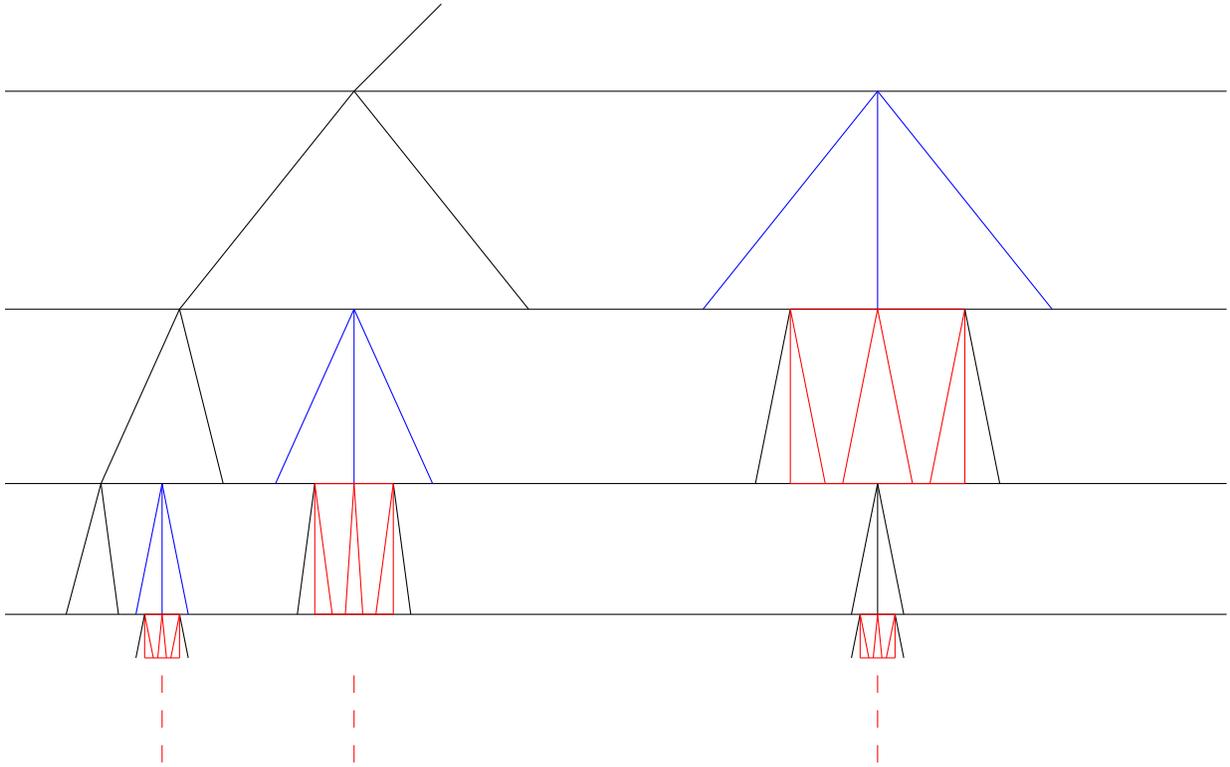
\begin{figure}[H]
\centering
\begin{tikzpicture}[scale=1.16]
\draw (0,0)--(14,0);
\draw (0,-2.5)--(14,-2.5);
\draw (0,-4.5)--(14,-4.5);
\draw (0,-6)--(14,-6);
\draw (2,-2.5)--(4,0);
\draw (4,0)--(6,-2.5);
\draw (2+0.5,-4.5)--(2,-2.5);
\draw (2-0.9,-4.5)--(2,-2.5);
\draw[color=blue] (4,-2.5)--(4+0.9,-4.5);
\draw[color=blue] (4,-2.5)--(4,-4.5);
\draw[color=blue] (4,-2.5)--(4-0.9,-4.5);
\draw (8-6-0.9-0.4,-6)--(8-6-0.9,-4.5);
\draw (8-6-0.9+0.2,-6)--(8-6-0.9,-4.5);
\draw[color=blue] (8-6-0.2,-6)--(8-6-0.2,-4.5);
\draw[color=blue] (8-6-0.2+0.3,-6)--(8-6-0.2,-4.5);
\draw[color=blue] (8-6-0.2-0.3,-6)--(8-6-0.2,-4.5);
\draw[color=red] (8-4+0.1,-6)--(8-4,-4.5);
\draw[color=red] (8-4-0.1,-6)--(8-4,-4.5);
\draw (8-4-0.65,-6)--(8-4-0.45,-4.5);
\draw[color=red] (8-4-0.45,-6)--(8-4-0.45,-4.5);
\draw[color=red] (8-4-0.25,-6)--(8-4-0.45,-4.5);

\draw (8-4+0.65,-6)--(8-4+0.45,-4.5);
\draw[color=red] (8-4+0.45,-6)--(8-4+0.45,-4.5);
\draw[color=red] (8-4+0.25,-6)--(8-4+0.45,-4.5);
\draw[color=blue] (10,0)--(8,-2.5);
\draw[color=blue] (10,0)--(10,-2.5);
\draw[color=blue] (10,0)--(12,-2.5);
\draw[color=red] (10,-2.5)--(10-0.4,-4.5);
\draw[color=red] (10,-2.5)--(10+0.4,-4.5);
\draw[color=red] (11,-2.5)--(11-0.4,-4.5);
\draw[color=red] (11,-2.5)--(11,-4.5);
\draw (11,-2.5)--(11.4,-4.5);
\draw (9,-2.5)--(9-0.4,-4.5);
\draw[color=red] (9,-2.5)--(9,-4.5);
\draw[color=red] (9,-2.5)--(9+0.4,-4.5);
\draw (10,-4.5)--(10,-6);
\draw (10,-4.5)--(10-0.3,-6);
\draw (10,-4.5)--(10+0.3,-6);
\draw (10.2,-6)--(10.3,-6.5);
\draw[color=red] (10.2,-6)--(10.1,-6.5);
\draw[color=red] (10.2,-6)--(10.2,-6.5);
\draw (9.8,-6)--(9.7,-6.5);
\draw[color=red] (9.8,-6)--(9.8,-6.5);
\draw[color=red] (9.8,-6)--(9.9,-6.5);
\draw[color=red] (10,-6)--(9.95,-6.5);
\draw[color=red] (10,-6)--(10.05,-6.5);
\draw[color=red] (1.9,-6.5)--(2,-6);
\draw (2.1,-6.5)--(2,-6);
\draw[color=red] (2,-6.5)--(2,-6);
\draw[color=red] (1.8,-6)--(1.85,-6.5);
\draw[color=red] (1.8,-6)--(1.75,-6.5);
\draw[color=red] (1.6,-6)--(1.7,-6.5);
\draw[color=red] (1.6,-6)--(1.6,-6.5);
\draw (1.6,-6)--(1.5,-6.5);
\draw[color=red] (9,-4.5)--(11,-4.5);
\draw[color=red] (9,-2.5)--(11,-2.5);
\draw[color=red] (3.55,-4.5)--(4.45,-4.5);
\draw[color=red] (3.55,-6)--(4.45,-6);
\draw[color=red] (1.6,-6)--(2,-6);
\draw[color=red] (1.6,-6.5)--(2,-6.5);
\draw[color=red] (9.8,-6)--(10.2,-6);
\draw[color=red] (9.8,-6.5)--(10.2,-6.5);
\draw[color=red] (10,-6.7)--(10,-6.9);
\draw[color=red] (10,-7.1)--(10,-7.3);
\draw[color=red] (10,-7.7)--(10,-7.5);
\draw[color=red] (1.8,-6.7)--(1.8,-6.9);
\draw[color=red] (1.8,-7.1)--(1.8,-7.3);
\draw[color=red] (1.8,-7.7)--(1.8,-7.5);
\draw[color=red] (4,-6.7)--(4,-6.9);
\draw[color=red] (4,-7.1)--(4,-7.3);
\draw[color=red] (4,-7.7)--(4,-7.5);
\draw (4,0)--(5,1);
\end{tikzpicture}
\vspace{3mm}
\caption{Pairs of red trapezes are removed (for clarity, the irrelevant parts of the bulk are not drawn). The red trapezes correspond to the trapezes that couple the left and the right of the 3-clusters drawn in blue on the boundary. These 3-clusters are the ones immediately next to the 2-clusters neighboring the dashed line. Under each 3-cluster, there is only two such trapezes every other line, which corresponds to a small proportion of the bulk.}
\label{fig:TrapezeRemoved}
\end{figure}

Take a boundary qubit. Because of the structure of the entanglement wedge, there exists one of the blue bulk 3-clusters of Figure \ref{fig:TrapezeRemoved} whose entanglement wedge is on its right. Now, as we have carefully taken away the red trapeze operators of Figure \ref{fig:TrapezeRemoved} from the bulk Hamiltonian, it means that no lower trapeze than the trapeze in blue on Figure \ref{fig:CouplingFail} can be used to connect our boundary qubit to the right of the red line. Actually, it is the only one that can: higher level trapezes such as the one in green have a support which is too much on the right. Now, say our blue trapeze does couple our boundary qubit to the right of the first red line. Then, the qubits it is coupled to will lie in the entanglement wedge drawn as a blue curve. By a similar argument, the only trapeze which can couple this entanglement wedge to the right of the second red line is the green trapeze. However, the entanglement wedge drawn as the inside of the green curve is clearly disjoint from the inside of the blue curve. Therefore, the second red line gives an upper bound on the boundary region which can be coupled to our initial boundary qubit.

\begin{figure}[H]
\centering
\vspace{3mm}

\begin{tikzpicture}[scale=2]
\draw (0,0)--(4,0);
\draw[color=green,line width=1mm] (4,0)--(7,0);
\draw (7,0)--(8,0);
\draw (0,-2.5)--(8,-2.5);
\draw[color=blue,line width=1mm] (2,-2.53)--(4,-2.53);
\draw[color=blue,line width=1mm] (3.97,-2.5)--(3.97,-4.5);
\draw[color=green,line width=1mm] (2,-2.47)--(7,-2.47);
\draw (0,-4.5)--(8,-4.5);
\draw[color=green,line width=1mm] (2,-2.5)--(4,0);
\draw[color=green,line width=1mm] (4,0)--(5,-2.5);
\draw[color=blue,line width=1mm] (2+0.5,-4.5)--(2,-2.5);
\draw[color=blue,line width=1mm] (2-0.9,-4.5)--(2,-2.5);
\draw (4,-2.5)--(4+0.9,-4.5);
\draw (4,-2.5)--(4,-4.5);
\draw[color=blue,line width=1mm] (4,-2.5)--(4-0.9,-4.5);
\draw[color=blue,line width=1mm] (2-0.9,-4.5)--(4,-4.5);
\draw (8-6-0.9-0.4,-6)--(8-6-0.9,-4.5);
\draw (8-6-0.9+0.2,-6)--(8-6-0.9,-4.5);
\draw (8-6-0.2,-6)--(8-6-0.2,-4.5);
\draw (8-6-0.2+0.3,-6)--(8-6-0.2,-4.5);
\draw (8-6-0.2-0.3,-6)--(8-6-0.2,-4.5);
\draw (8-4+0.1+0.03,-6)--(8-4+0.03,-4.5);
\draw (8-4-0.1+0.03,-6)--(8-4,-4.5);
\draw (8-4-0.65,-6)--(8-4-0.45,-4.5);
\draw (8-4-0.45,-6)--(8-4-0.45,-4.5);
\draw (8-4-0.25,-6)--(8-4-0.45,-4.5);
\draw (4+0.9,-4.5)--(4+0.9+0.2,-6);
\draw (4+0.9,-4.5)--(4+0.9-0.1,-6);
\draw (8-4+0.65,-6)--(8-4+0.45,-4.5);
\draw (8-4+0.45,-6)--(8-4+0.45,-4.5);
\draw (8-4+0.25,-6)--(8-4+0.45,-4.5);
\draw (7,0)--(8,-2.5);
\draw[color=green,line width=1mm] (7,0)--(7,-2.5);
\draw[color=green,line width=1mm] (7,0)--(6,-2.5);
\draw (4,0)--(5,1);
\draw[color=red,line width=1mm] (1.8,-4.5)--(1.8,-7);
\draw[color=green,line width=1mm] (2+0.5,-4.5) to[out=20,in=160] (4-0.9,-4.5);
\draw[color=green,line width=1mm] (2+0.3,-6) to[out=-90,in=-120] (2+0.5,-4.5);
\draw[color=green,line width=1mm] (4-0.9,-4.5) to[out=-60,in=-90] (4-0.7,-6);
\draw[color=blue,line width=1mm] (3.75,-6) to[out=20,in=160] (3.9,-6);
\draw[color=blue,line width=1mm] (3.7,-7) to[out=-90,in=-120] (3.75,-6);
\draw[color=blue,line width=1mm] (3.9,-6) to[out=-60,in=-90] (3.95,-7);
\draw[color=red,line width=1mm] (4.03,-2.5)--(4.03,-7);
\end{tikzpicture}
\vspace{2mm}

\caption{Why a boundary qubit cannot be coupled to an infinite number of other ones. A boundary qubit on the left of the first red line can only be coupled its right by the blue trapeze. The qubits it is coupled to can only be coupled to the right of the second red line by the trapeze in green. Now, the entanglement wedges in green and in blue, that we defined for a given trapeze operator as the portion of the bulk for which the bulk-to-boundary map is nontrivial, are disjoint: in fact, this coupling is not possible. Therefore the boundary qubit cannot be coupled to anything past the second red line.}
\vspace{1mm}

\label{fig:CouplingFail}
\end{figure}
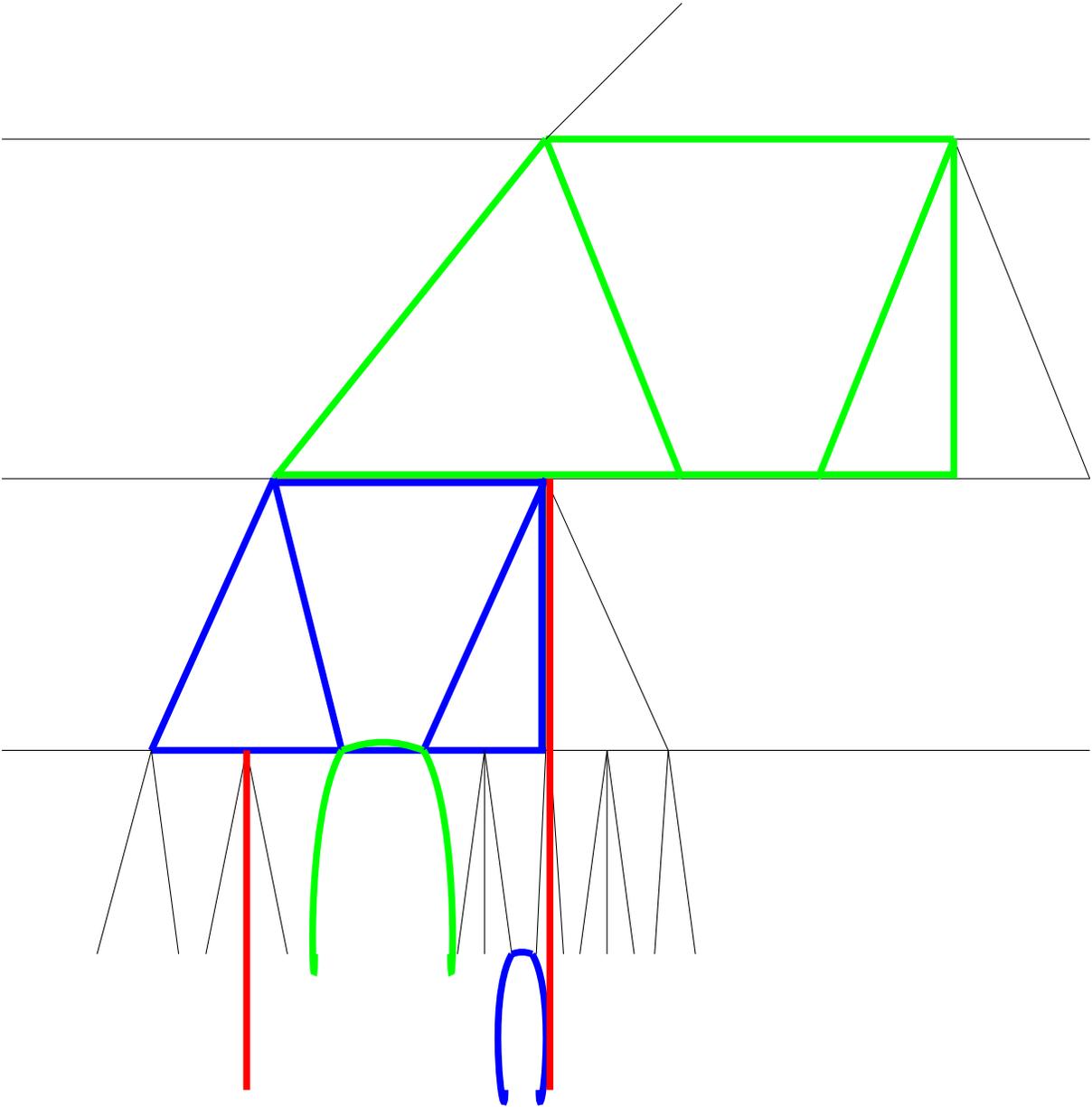

In short, we have proven that a boundary qubit on the left of a red line of removed trapezes cannot be entangled to anything on the right of the next red line. This gives an upper bound on the number of qubits it can be entangled to. Moreover, this upper bound is level independent, in the sense that the position of the red line just depends on the distance of the boundary qubit to the dashed line (which would be on the extreme left in this picture). This achieves to prove that all long-range entanglement is lost during the shifting and unfolding procedure.

The proof of this appendix will turn out to be useful to show the existence of boundary KMS states in the setting of our companion paper \cite{MonicaElliott2}, hence we shall end with a brief remark regarding this slightly different setting. The interested reader is referred to \cite{MonicaElliott2} for more details. Suppose that instead of pushing the bulk to the boundary from successive levels, one reconstructs a trapeze entanglement wedge from a fixed string of qubits on the boundary, through a sort of reverse engineering procedure. As we saw, the shifted and unfolded Hamiltonian can be thought of as a superposition of operators on this fixed infinite string of qubits. In such a context, what this proof tells us is that the shifted and unfolded Hamiltonian only couples finite subsets of that infinite string together. This gives another interpretation of what short-range correlations means.

\end{document}